\documentclass[
10pt,
aps,
prx,
twocolumn,
superscriptaddress,
floatfix,
nofootinbib
]{revtex4-2}
\usepackage{pretex}

\usepackage[english]{babel}
\usepackage{silence}
\usepackage{bbm}
\usepackage{algorithm}
\def\be{\begin{eqnarray}}
\def\ee{\end{eqnarray}}
\DeclareMathOperator{\Tr}{Tr} 
\DeclareMathOperator{\EX}{\mathbb{E}}

\newcommand{\bb}[1]{\mathbb{#1}}

\usepackage[capitalise]{cleveref}
\crefformat{equation}{Eq.~(#2#1#3)}
\crefformat{section}{Sec.~#2#1#3}
\Crefformat{equation}{Equation~(#2#1#3)}
\crefformat{figure}{Fig.~#2#1#3}
\crefrangeformat{equation}{Eqs.~#3(#1)#4--#5(#2)#6}
\Crefformat{section}{Section~#2#1#3}

\begin{document}

\title{A unified loss formalism for improving variational quantum algorithms}
\title{Quantum Tilted Loss: A Unified Loss Family for Variational Quantum Algorithms}
\title{Quantum Tilted Loss in Variational Optimization: Theory and Applications}

\author{Yixian Qiu}
\email{yixian\_qiu@u.nus.edu}
\affiliation{Centre for Quantum Technologies, National University of Singapore, 3 Science Drive 2, Singapore 117543}
\author{Josep Lumbreras}
\email{josep.lz@ntu.edu.sg}
\affiliation{Centre for Quantum Technologies, National University of Singapore, 3 Science Drive 2, Singapore 117543}
\affiliation{Nanyang Quantum Hub, School of Physical and Mathematical Sciences, Nanyang Technological University, Singapore}
\author{Xiufan Li}
\email{lixiufan@u.nus.edu}
\affiliation{Centre for Quantum Technologies, National University of Singapore, 3 Science Drive 2, Singapore 117543}
\author{Patrick Rebentrost}
\email{cqtfpr@nus.edu.sg}
\affiliation{Centre for Quantum Technologies, National University of Singapore, 3 Science Drive 2, Singapore 117543}
\affiliation{Department of Computer Science, National University of Singapore, 13 Computing Drive, Singapore 117417}

\date{\today}

\begin{abstract}
Variational quantum algorithms (VQAs) are leading strategies for using near-term quantum devices, with a well-studied bottleneck being their trainability. Standard expectation-value objectives with expressive circuits frequently encounter barren plateaus in the optimization landscape during training. To address this challenge, we introduce the Quantum Tilted Loss (QTL), an operator-level generalization of classical exponential tilting designed to systematically reshape the optimization landscape. By tuning a single continuous parameter, QTL can amplify gradient signals in structured settings while preserving the problem's true global minima. We provide a theoretical foundation that unifies standard expectation minimization with popular tunable heuristics, such as Conditional Value-at-Risk (CVaR) and Gibbs formulations. Deploying this framework requires balancing the geometric benefits of a sharpened landscape against the statistical cost of estimating nonlinear gradients from finite quantum measurements. We formalize this trainability-estimability trade-off, demonstrating how aggressive tilting fundamentally shifts the optimization bottleneck from landscape flatness to sample complexity. Thus, the operational bottleneck shifts from vanishing gradients to measurement sampling variance. Finally, we exhibit through numerical simulations that ascending tilt schedules can outperform fixed-tilt training in finite-shot regimes.
\end{abstract}

\maketitle
\tableofcontents

\section{Introduction}

Variational quantum algorithms (VQAs) are a leading approach for near-term quantum computation. By combining shallow parameterized circuits with a classical outer-loop optimizer, VQAs operate within the coherence and gate-fidelity constraints of current hardware without the overhead of full fault-tolerant compilation~\cite{Cerezo_2021review}. Foundational algorithms such as the Variational Quantum Eigensolver (VQE)~\cite{Peruzzo_2014} and the Quantum Approximate Optimization Algorithm (QAOA)~\cite{farhi2014} have demonstrated significant potential across quantum chemistry, condensed matter physics, and combinatorial optimization. Despite this promise, practical performance is fundamentally bottlenecked by trainability. Standard expectation-value objectives typically exhibit exponentially vanishing gradients as system size or circuit depth increases, leading to the barren plateau phenomenon and rendering gradient-based optimization ineffective at scale~\cite{McClean2018,Wang2021,Marrero2021,Holmes_2022,Larocca_2025}.

This challenge has motivated substantial effort not only in ansätz, optimizer design and initialization~\cite{Kandala_2017,Romero2018,Tang2021,Stokes_2020,Grant2019}, but also in the formulation of alternative objective functions. The standard linear expectation-value objective, while physically intuitive, actively constructs optimization landscapes dominated by barren plateaus and narrow gorges when applied to highly expressive circuits~\cite{McClean2018,Sim2019,Arrasmith2021effectofbarren,Larocca_2025}. Consequently, loss design provides an avenue for improving variational optimization. Previous efforts have introduced various heuristic modifications, including local cost functions~\cite{Cerezo2021}, the CVaR objective~\cite{Barkoutsos2020}, and Gibbs-type objectives for ansätz search~\cite{gibbs_objective}, and filtering variational objectives~\cite{Amaro2022}. Further explorations have utilized Rényi divergences~\cite{Kieferova2021}, Hilbert-Schmidt overlaps~\cite{Khatri2019}, and Quantum Wasserstein Distances~\cite{De_Palma_2021}. While effective within specific domains, these diverse approaches highlight the need for a cohesive underlying theory. 

Such a theory naturally emerges by examining the geometry of optimization. While specialized ansätz designs attempt to restrict the parameter search space, modifying the loss function directly re-engineers the underlying topology of the landscape itself. Alternative quantum objectives, such as CVaR and Gibbs-based costs, typically achieve this reshaping by introducing tunable, non-linear parameters to emphasize specific spectral outcomes. To systematically analyze this mechanism, classical risk theory provides a rigorous foundation. Originally developed to model and hedge against extreme tail risks, the statistical concept of exponential tilting reweights probability distributions by an exponential factor, normalized via a log-partition function~\cite{esscher1932probability, Bental2007, AhmadiJavid2011}. In classical machine learning, this principle is discussed as tilted empirical risk minimization~\cite{li2020tilted,li2023tilted}, using a continuous hyperparameter to smoothly interpolate between average-loss minimization and the amplification of tail outcomes. While classical applications primarily aim to achieve statistical robustness, the underlying mathematical operation fundamentally alters the objective's geometry. Consequently, classical exponential tilting provides a mathematical framework we can repurpose: it offers a language to understand how tunable, non-linear transformations dynamically reshape quantum optimization landscapes.

In this work, we introduce the Quantum Tilted Loss (QTL) as an analogue of classical entropic risk. Rather than proposing QTL as a universal remedy for barren plateaus, we use this framework to study how nonlinear parameterized transformations reshape the variational landscape, and what is the statistical cost they introduce. Across different tasks, tuning a single continuous risk parameter $\gamma$ allows the objective to shift continuously between standard expectation-value minimization, variance suppression, and ground-state projection. Specifically, applying a negative $\gamma$ emphasizes low-energy outcomes for ground-state minimization, while a positive $\gamma$ emphasizes high-energy outcomes.

However, fixing a flat optimization landscape is not free. When we use the tilt parameter to make the gradients larger and easier for the optimizer to follow, we also increase the statistical noise in our measurements. This means we need many more measurement shots to get a reliable gradient estimate. In practice, aggressive tilting does not remove the barren plateau problem by itself. 
It shifts the bottleneck from trainability to estimability: from a flat landscape to a larger sampling overhead.

\begin{figure*}[t]
    \centering
    \includegraphics[trim={0cm 10pt 0cm 10pt}, clip, width=0.7\linewidth]{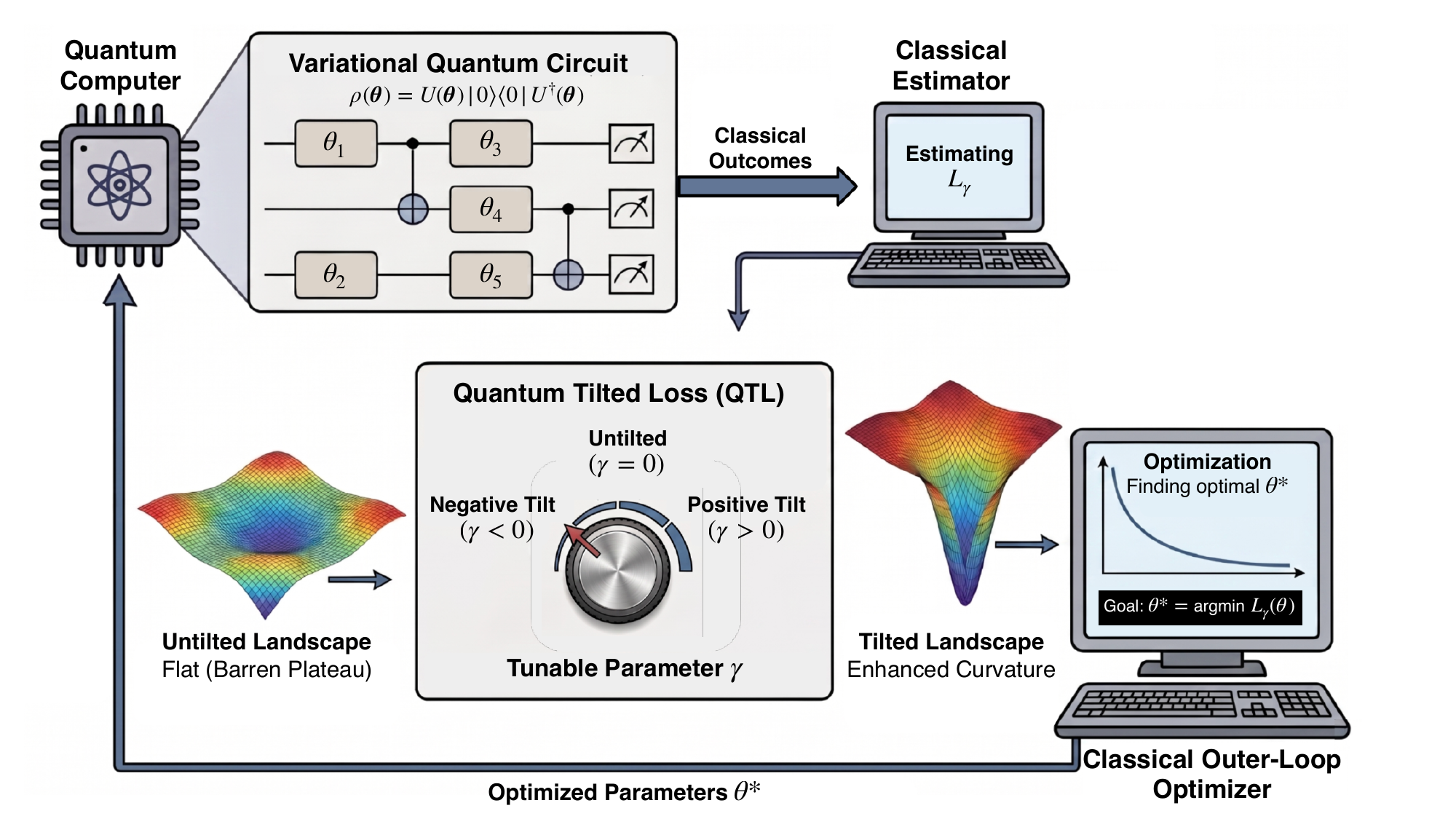}
    \caption{Schematic of the Quantum Tilted Loss (QTL) optimization loop. A parameterized quantum state is prepared via a Variational Quantum Circuit (VQC) as $|\psi(\bm{\theta})\rangle = U(\bm{\theta})|0\rangle$, corresponding to the density matrix $\rho(\bm{\theta}) = U(\bm{\theta})|0\rangle\langle0|U^\dagger(\bm{\theta})$. Classical measurement outcomes are fed into a loss calculator to evaluate the QTL, defined as $\mathcal{L}_\gamma(O,\rho(\bm{\theta})) = \frac{1}{\gamma}\log \text{Tr} \left( e^{\gamma O} \rho(\bm{\theta}) \right)$. By tuning the continuous risk parameter $\gamma$, the nonlinear exponential reweighting actively reshapes the optimization landscape. While the untilted expectation-value objective ($\gamma \to 0$) suffers from flat regions characteristic of barren plateaus, a finite tilt enhances local curvature and amplifies gradient signals. For energy minimization tasks, a negative tilt ($\gamma < 0$) is used to emphasize low-energy outcomes. A classical outer-loop optimizer subsequently updates the parameters to find the optimal $\bm{\theta}^*$ and updates the VQC to continue the optimization.}
    \label{fig:qtl_schematic}
\end{figure*}

We analyze the theoretical properties, optimization geometry, and operational limitations of the QTL framework, yielding the following main contributions:
\begin{itemize}
   \item \textbf{A theoretical framework for tunable quantum losses.}
   We introduce QTL as an outcome-level entropic tilting of the spectrum of a measured observable. 
   The objective recovers the standard expectation value at zero tilt and approaches extremal spectral outcomes at large tilt. 
   We prove that QTL is faithful, so that it preserves the global minimizers of the target Hamiltonian. 
   We also show that QTL connects existing tunable objectives: it admits a Gibbs/free-energy interpretation and gives a smooth entropic counterpart to lower-tail CVaR~\cite{Barkoutsos2020}, with an explicit correction term. 
   This turns loss reshaping into a unified object that can be analyzed theoretically.

   \item \textbf{Trainability--estimability trade-off in landscape reshaping.} 
   We study nonlinear objectives under finite-shot quantum measurements. 
   Stronger tilting can amplify gradient signals and help the optimizer escape local plateaus. 
   At the same time, estimating the resulting nonlinear gradients becomes harder, with a shot overhead that grows exponentially with the tilt.

   \item \textbf{Structured benchmark and finite-shot behavior.}
   We introduce a global-projector benchmark with observable
   $O_G=\mathbb{I}-|0\rangle\langle 0|^{\otimes n}$ and a tensor-product ansatz. 
   The model is analytically tractable and isolates the small-overlap mechanism behind landscape reshaping. 
   We use it to show how QTL preserves the global minimum while modifying the gradient and Hessian structure. 
   QAOA experiments for MaxCut then illustrate the same finite-shot trade-off: modest tilts can improve the performance, whereas overly large fixed tilts degrade performance. 
\end{itemize}
\noindent
Through this analysis, supported by numerical experiments, we provide a mathematically grounded approach to loss design, clarifying exactly when, why, and at what statistical cost tilting can improve variational quantum optimization.

\section{Quantum Tilted Loss}
\subsection{Notations and preliminaries}\label{sec:pre}

Here we assume a finite-dimensional Hilbert space. Let us denote the set of positive semidefinite operators acting on this space by
$\mathcal{P}$ and the subset of density operators (with unit trace) by $\mathcal{S}$. For a density operator $\sigma \in \mathcal{S}$ and $\rho \in \mathcal{P}$. The Umegaki quantum relative entropy between the quantum states $\sigma$ and $\rho$ is denoted by $D(\sigma \| \rho)$. We use $\|\cdot\|_\infty$ to denote the spectral norm (operator norm) of an operator.

\subsection{Definition and properties}\label{sec:def_pro}

To construct a unified, parameterized objective for variational quantum optimization, we draw upon the classical entropic risk measure \cite{Fllmer2016}. In classical probability and decision theory, the entropic risk of a loss random variable $X \sim P$ at risk sensitivity $\gamma \in \mathbb{R} \setminus \{0\}$ is defined as $\rho_\gamma(X) := \frac{1}{\gamma}\log \mathbb{E}_P[e^{\gamma X}]$~\cite{AhmadiJavid2011}. This mathematical structure translates directly to statistical machine learning: evaluating this functional over an empirical training distribution yields Tilted Empirical Risk Minimization (TERM)~\cite{li2020tilted, li2023tilted}. Across both domains, the parameter $\gamma$ smoothly interpolates the objective between the standard expected value ($\gamma \to 0$) and the extremal tail outcomes ($|\gamma| \to \infty$). 

Crucially, this risk measure admits a variational representation via Kullback-Leibler divergence regularization, where the optimal exponentially reweighted distribution is known as the classical Esscher transform (we provide a comprehensive review of these classical representations in Appendix~\ref{app:classical_risk}). The operator-algebraic analogue of this process—the quantum Esscher transform~\cite{Qiu2024Esscher}—skews a density matrix toward target spectral subspaces using an exponentiated observable. Motivated by the partition function that normalizes this quantum tilting process, we define the following natural generalization. All proofs for this subsection are provided in Appendix~\ref{app:propery-QTL}.

\begin{definition}[Quantum $\gamma$-Tilted Loss (QTL)]\label{def:QTL}
For all $\gamma \in \mathbb{R}$, given a quantum state $\rho$ and an Hermitian observable $O$, the $\gamma$-quantum tilted loss is defined as: 
\begin{align}\label{eq:qtl_def}
    \mathcal{L}_\gamma(O, \rho) := \begin{cases}
\frac{1}{\gamma} \log \mathrm{Tr}(e^{\gamma O} \rho) \quad\text{if } \gamma \neq 0, \\
\lim_{\gamma \to 0} \frac{1}{\gamma} \log \mathrm{Tr}(e^{\gamma O} \rho) = \mathrm{Tr}(O \rho).
\end{cases}
\end{align}

\end{definition}

In the below lemma we check that the above expression for $\mathcal{L}_\gamma(O, \rho)$ is well defined.

\begin{lemma}[Convergence of QTL]\label{Lemma:convergence} For any fixed observable $O$ and quantum state $\rho$, the QTL smoothly converges to the standard expectation value in the zero-tilt limit as
\be
\lim _{\gamma \rightarrow 0} \mathcal{L}_\gamma(O,\rho) := \mathcal{L}_0(O,\rho)= \tr \left( O \rho \right).
\ee
\end{lemma}

The QTL is the quantum analogue of tilted empirical risk \cite{li2023tilted}: $\gamma$ sets the emphasis on the tail of the outcome spectrum of Hermitian observable $O$, with $\gamma\rightarrow 0$ giving the expected value. Crucially, we adopt the following sign convention for minimization: a negative $\gamma$ ($\gamma < 0$) emphasizes low-energy outcomes, while a positive $\gamma$ ($\gamma > 0$) emphasizes high-energy outcomes, mirroring the classical LogSumExp continuum between mean and extremal losses.

Some elementary properties of the QTL are collected in the lemmas below.

\begin{lemma}[Basic Properties of QTL]\label{lem:prop}
The QTL satisfies the following:
\begin{enumerate}
    \item Non-negativity: $\mathcal{L}_\gamma(O,\rho) \ge 0$ if the eigenvalues of $O$ are non-negative, i.e., $O\succeq 0$.
    \item Additivity: $ \mathcal{L}_\gamma(O,\rho)=\mathcal{L}_\gamma(O_A,\rho_A)+\mathcal{L}_\gamma(O_B,\rho_B)$ if $\rho=\rho_A\otimes\rho_B$ and $O=O_A\otimes\mathbb{I}_B+\mathbb{I}_A\otimes O_B$.
    \item Shift: $\mathcal{L}_\gamma(O+c\mathbb{I},\rho)=\mathcal{L}_\gamma(O,\rho)+c$ for any constant $c\in\mathbb R$.
    \item Monotonicity in $\gamma$: $\mathcal{L}_\gamma(O,\rho)$ is non-decreasing on $\mathbb{R}$.
\end{enumerate}
\end{lemma}

\begin{lemma}[$\gamma$-Limits of QTL]\label{lem:limit}
Let $O = \sum_i \lambda_i \Pi_i$ be the spectral decomposition of an observable $O$, where $\Pi_i$ are the corresponding eigenspace projectors. Given a quantum state $\rho$, let $p_i = \mathrm{Tr}(\rho \Pi_i)$. Define $m_\rho$ and $M_\rho$ as the minimum and maximum eigenvalues of $O$ that have a strictly positive  weight on $\rho$ as $m_\rho = \min \{ \lambda_i \mid p_i > 0 \}, M_\rho = \max \{ \lambda_i \mid p_i > 0 \}$. The QTL exhibits the following limiting behavior:
\begin{align}
    \lim_{\gamma \to -\infty} \mathcal{L}_\gamma(O,\rho) &= m_\rho \quad (\text{Minimum})\\\nonumber
    \lim_{\gamma \to 0} \mathcal{L}_\gamma(O,\rho) &= \tr(O\rho) \quad (\text{Expectation Value}) \\\nonumber
    \lim_{\gamma \to +\infty} \mathcal{L}_\gamma(O,\rho) &= M_\rho \quad (\text{Maximum}) 
\end{align}
\end{lemma}

\begin{lemma}[Faithfulness of QTL]\label{lem:faithfulness}
Let $O$ be an observable with minimum eigenvalue $o_{\min}$ and corresponding eigenspace projector $\Pi_{\min}$. For any non-zero tilt parameter $\gamma \in \mathbb{R} \setminus \{0\}$, the global minimum of the Quantum Tilted Loss over the set of all density operators $\mathcal{S}$ is exactly the ground state energy of $O$:
\begin{align}
    \min_{\rho \in \mathcal{S}} \mathcal{L}_\gamma(O,\rho) = o_{\min}.
\end{align}
Furthermore, a state $\rho^\star \in \mathcal{S}$ achieves this minimum if and only if it is entirely supported on the ground state eigenspace, $\tr(\rho^\star \Pi_{\min}) = 1$.
\end{lemma}

 This faithfulness property is fundamentally why the QTL serves as a viable objective for Variational Quantum Algorithms. When replacing the standard expectation-value objective $\tr(O\rho(\bm{\theta}))$ with the QTL $\mathcal{L}_\gamma(O,\rho(\bm{\theta}))$, Lemma~\ref{lem:faithfulness} guarantees that the target of the optimization remains completely invariant. As long as the parameterized ansätz $\rho(\bm{\theta})$ is sufficiently expressive to represent the ground state of $O$, minimizing the QTL recovers the exact same global minimum as standard expectation-value minimization. We formally analyze how the parameter $\gamma$ modifies the surrounding optimization landscape during training in Section~\ref{sec:trainability}.

The QTL also possesses a Gibbs variational property, which provides an alternative representation in terms of the Umegaki quantum relative entropy.

\begin{lemma}[Gibbs variational bound]
\label{lem:gibbs-variational}
Fix a real number $\gamma \neq 0$. Let $\rho$ be a full-rank density operator. Then the quantum tilted loss has the following variational representation:

For $\gamma>0$, it is given by a supremum:
\begin{align}
\mathcal{L}_\gamma 
\ge \sup_{\sigma} \Big\{ \tr[\sigma O] - \frac{1}{\gamma}D(\sigma\Vert \rho) \Big\}.
\end{align}

For $\gamma<0$, it is given by an infimum:
\begin{align}
\mathcal{L}_\gamma  \le \inf_{\sigma}\Big\{\tr[\sigma O] - \frac{1}{\gamma} D(\sigma\Vert \rho) \Big\},
\end{align}
where the optimization is over all density operators $\sigma$, and $D(\sigma\Vert\rho):=\tr[\sigma(\log\sigma-\log\rho)]$ is the Umegaki quantum relative entropy. Moreover, equality holds for both cases iff $[\rho,O]=0$, in which case the unique optimizer is $\sigma^\star\propto e^{\log\rho+\gamma O}$.
\end{lemma}

This bound makes explicit a tunable trade-off: optimizing $\mathcal{L}_\gamma$ can be interpreted as optimizing the expected observable value under an auxiliary state $\sigma$ while penalizing its
deviation from the reference state $\rho$ via relative entropy, with penalty strength $1/\gamma$.
(For $\gamma<0$, an analogous infimum form holds with the same integrand.)

\begin{remark}
By the quantum Gibbs variational principle, a natural fully commutative comparator is
\begin{align}
\mathcal{L}^c_\gamma (O,\rho) :=\tfrac{1}{\gamma}\log\tr[e^{\log \rho + \gamma O }].
\end{align}
As shown in Lemma~\ref{lem:gibbs-variational}, this comparator corresponds to a state-level tilting directly regularized by the quantum relative entropy. The unique optimizer of the variational problem in Lemma \ref{lem:gibbs-variational} does not have a simple closed form in the general non-commutative case. However, the optimizer simplifies to the quantum Esscher transformed state \cite{Qiu2024Esscher} (for full rank $\rho$) which is
\be
\sigma^\star = \frac{e^{\log\rho + \gamma O}}{\tr[e^{\log\rho + \gamma O}]}.
\ee
\end{remark}

By contrast, the QTL in Definition~\ref{def:QTL} is entirely determined by the classical outcome distribution of measuring $O$ on $\rho$. Indeed, writing the spectral decomposition $O=\sum_i o_i\Pi_i$ and defining the Born probabilities $p_i=\Tr(\rho\Pi_i)$ gives $\mathcal{L}_\gamma(O,\rho)=\frac{1}{\gamma}\log\Big(\sum_i p_i e^{\gamma o_i}\Big)$. Thus, its variational form is a \emph{classical KL-regularized} reweighting over outcomes: for $\gamma>0$,
\begin{align}
\mathcal{L}_\gamma(O,\rho) = \sup_{q\in\Delta}\left\{\sum_i q_i o_i-\frac{1}{\gamma}D_{\mathrm{KL}}(q\|p)\right\},
\end{align}
which is the analogue of the \emph{classical TERM \cite{li2020tilted, li2023tilted}} (see~\cref{eq:term_var} in Appendix) with the reference distribution replaced by $p$. Hence, Golden--Thompson implies that $\mathcal{L}^c_\gamma(O,\rho)\le \mathcal{L}_\gamma(O,\rho)$ for $\gamma>0$, with equality whenever $[O,\rho]=0$. In this sense, Definition~\ref{def:QTL} can be viewed as a ``measurement-first'' (outcome-level) entropic tilting, whereas $\mathcal{L}^c_\gamma$ performs a ``state-level'' tilting.  For $\gamma<0$, the same statements hold with the supremum replaced by an infimum, and the Golden--Thompson comparison reverses accordingly.

\subsection{Risk and thermodynamic interpretations}
\textbf{Risk interpretation: Connection to machine learning and economic risk.} In classical risk theory, $\mathcal{L}_\gamma$ is the entropic risk measure and a canonical convex-risk construction. QTL is its measurement-outcome analogue under the Born distribution \cite{Fllmer2016}. 
Let $O=\sum_i o_i \Pi_i$ be the spectral decomposition of a Hermitian observable, and define the
measurement-induced distribution $p_i:=\tr(\rho\,\Pi_i)$ over the eigenvalues $\{o_i\}$.
Then we can express the QTL as
\be\label{eq:qtl_classical_entropic}
\mathcal{L}_\gamma(O,\rho)
=\frac{1}{\gamma}\log \tr\!\big(\rho e^{\gamma O}\big)
=\frac{1}{\gamma}\log\sum_i p_i e^{\gamma o_i},
\ee
Under the standard sign convention that interprets $X$ as a loss, \cref{eq:qtl_classical_entropic} is exactly the entropic (exponential) risk functional. This formulation naturally admits a classical certainty-equivalent interpretation from the exponential utility—specifically, representing Constant Absolute Risk Aversion (CARA). We provide a detailed review of this economic framework in Appendix~\cref{app:classical_risk}.
This yields a risk-sensitive generalization of the mean as $\gamma\to 0$ recovers $\mathbb{E}[X]$, and
as we increase the magnitude of $\gamma$ it emphasizes specific tail outcomes: negative $\gamma$ isolates the lower tail (low-energy), and positive $\gamma$ isolates the upper tail.

\textbf{Physical interpretation: partition functions and free energy.}
When $O$ is a Hamiltonian $H$ and $\rho$ is a quantum state, define
\be
Z_\gamma:=\tr\!\left(e^{\gamma H}\rho\right).
\ee
This quantity is a tilted moment-generating function of the energy
under the Born distribution. Derivatives of $Z_\gamma$ at
$\gamma=0$ generate raw energy moments, while derivatives of
$K(\gamma):=\log Z_\gamma=\gamma L_\gamma$ generate energy
cumulants~\cite{Kubo1962}. This log-generating-function viewpoint
is also standard in large-deviation and statistical-mechanical
formulations, where free energies are naturally interpreted as
scaled cumulant-generating functions~\cite{Touchette2009}. Related
quantum statistical-function formalisms provide a more recent
operator-level perspective~\cite{Emori2026}. As a concrete example, consider the $n$-qubit Gibbs state at inverse temperature $\beta$,
\be
\rho_{\beta,H}=\frac{e^{-\beta H}}{Z_\beta},\qquad Z_\beta:=\tr\!\left(e^{-\beta H}\right).
\ee
For any $\gamma$,
\begin{align}
\mathcal{L}_\gamma(H,\rho_{\beta,H})
&=\frac{1}{\gamma}\log\tr\!\left(e^{\gamma H}\frac{e^{-\beta H}}{Z_\beta}\right)\\ \nonumber
&=\frac{1}{\gamma}\Big(\log Z_{\beta-\gamma}-\log Z_\beta\Big),
\label{eq:qtl_gibbs_relation}
\end{align}
where $Z_{\beta-\gamma}=\tr(e^{-(\beta-\gamma)H})$.
In the special case $\gamma=\beta$,
\be
\mathcal{L}_\beta(H,\rho_{\beta,H})
=\frac{1}{\beta}\big(\log d-\log Z_\beta\big),
\ee
with $d=2^n$ the Hilbert-space dimension.
Defining the Helmholtz free energy as
$F(\beta):=-(1/\beta)\log Z_\beta$, we obtain
\be
\mathcal{L}_\beta(H,\rho_{\beta,H})=F(\beta)+\frac{1}{\beta}\log d.
\ee
Thus, in the matched Gibbs-state setting and for fixed $\beta$ and $d$, minimizing $\mathcal{L}_\beta(H,\rho_{\beta,H})$ is equivalent to minimizing the Helmholtz free energy up to an additive constant.

\section{Relations to other loss functions}\label{sec:relations_losses}

\subsection{Relation to financial risk measure CVaR}
In this section, we compare the Quantum Tilted Loss (QTL) with alternative risk-sensitive objective functions, focusing in particular on the Conditional Value-at-Risk (CVaR) employed in~\cite{Barkoutsos2020}. In that work, CVaR was introduced as a replacement for the standard mean-energy objective in VQE and QAOA, with the specific goal of emphasizing the lowest-energy outcomes of the distribution. Unlike the traditional “expected shortfall’’ definition from finance, which concentrates on the upper tail of losses, the variant used in~\cite{Barkoutsos2020} averages only over the lowest $\alpha$-fraction of measured energies, thereby reshaping the variational landscape to amplify favourable (low-energy) samples. In contrast, the QTL modifies the objective through exponential reweighting: for negative tilt parameters, it smoothly amplifies low-energy events, while for positive tilt it emphasises high-energy fluctuations. Below, we clarify these relationships and discuss how QTL and CVaR are related.

In order to introduce the notion of $\mathrm{CVaR}$, we first define its continuous formulation and subsequently describe how both $\mathrm{CVaR}$ and the QTL can be estimated empirically from measurement samples. Let $\rho(\bm{\theta})$ be a parametrized ansätz state, where $\bm{\theta} \in \mathbb{R}^d$ is a vector of trainable parameters, and let $O$ denote the Hamiltonian of interest. Measuring the observable $O$ on the parameterized state $\rho(\bm{\theta})$ yields a classical random variable $E(\bm{\theta})$, whose underlying probability distribution $P_{\bm{\theta}}$ is given by Born's rule. Adopting general integral notation to align with standard risk literature (understood with respect to the discrete spectral measure), its expectation is
\begin{align}
\mathbb{E}_{P_{\bm{\theta}}}\!\left[E(\bm{\theta})\right]
    &= \int_{-\infty}^{\infty} e\, dP_{\bm{\theta}}(e) = \Tr\!\big( O\,\rho(\bm{\theta}) \big),
\end{align}
which recovers the usual mean energy of the ansätz.  
The cumulative distribution function associated with $P_{\bm{\theta}}$ can be expressed as
\begin{align}
F_{\bm{\theta}}(e)
    \;=\;\Pr\!\left( E(\bm{\theta}) \le e \right)
    \;=\;\int_{-\infty}^{e} dP_{\bm{\theta}}(x) .
\end{align}

The (lower-tail) Conditional Value-at-Risk at confidence level $\alpha \in (0,1)$ represents the conditional expectation $\mathbb{E}[...]$ over the strictly lowest $\alpha$-fraction of outcomes. For continuous distributions, this is defined directly using the generalized inverse CDF $q_\alpha(\bm{\theta}) = F_{\bm{\theta}}^{-1}(\alpha)$ (commonly referred to in risk theory as the Value-at-Risk, $\mathrm{VaR}_{\alpha}(\bm{\theta})$) as the conditional expectation below this quantile. However, for discrete distributions with probability mass at the quantile---which is the exact setting for finite-spectrum Hamiltonians---the definition must partially include the atom at $q_\alpha(\bm{\theta})$ to ensure exactly a probability mass of $\alpha$ is averaged. Formally, the rigorous $\mathrm{CVaR}_{\alpha}$ objective is defined as
\begin{align}
\label{eq:cvar_discrete}
\mathrm{CVaR}_{\alpha}(\bm{\theta})
    : =  \frac{1}{\alpha} (\mathbb{E}_{P_{\bm{\theta}}} \bigl[\,E(\bm{\theta})\; \mathbb{I}\!\{E(\bm{\theta}) < q_\alpha(\bm{\theta})\}\bigr] \nonumber \\
    + q_\alpha(\bm{\theta}) \bigl(\alpha - \Pr(E(\bm{\theta}) < q_\alpha(\bm{\theta}))\bigr) ).
\end{align}
This expression correctly splits the probability atom if necessary. Equivalently, utilizing the quantile function $F_{\bm{\theta}}^{-1}(u)$, it admits a unified integral representation valid for any distribution \cite{rockafellar2000optimization}:
\begin{align}
\label{eq:cvar_integral}
\mathrm{CVaR}_{\alpha}(\bm{\theta})
    \;=\; \frac{1}{\alpha} \int_{0}^{\alpha} F_{\bm{\theta}}^{-1}(u)\, du.
\end{align}
Thus $\mathrm{CVaR}_{\alpha}(\bm{\theta})$ rigorously averages over exactly the lowest $\alpha$-fraction of energy outcomes, thereby emphasising low-energy configurations in variational optimisation, as done in \cite{Barkoutsos2020}.

\bigskip

Then in the following theorem we can relate the CVaR to the QTL by adding a corrective term that ensures a rigorous lower bound between both quantities.

\begin{theorem}[Negative $\gamma$ regime]\label{thm:cvar_vs_tilt}
Let $E(\bm{\theta})$ be a real-valued random energy with law $P_{\bm{\theta}}$. 
For any $\gamma < 0$ and any $\alpha\in(0,1)$, define the tilted loss
\begin{align}
    \mathcal{L}_\gamma(E,P_{\bm{\theta}}) 
    &:= \frac{1}{\gamma}\log \EX_{P_{\bm{\theta}}} \left[ e^{\gamma E} \right] \\\nonumber
    &= \frac{1}{\gamma}\log \left( \int_{-\infty}^{\infty} e^{\gamma E} \, dP_{\bm{\theta}}(E) \right)
\end{align}
and the lower-tail Conditional Value-at-Risk $\mathrm{CVaR}_{\alpha}(\bm{\theta})$ as in \cref{eq:cvar_integral}.  
Then the following inequality holds
\begin{align}
    \mathrm{CVaR}_{\alpha}(\bm{\theta})
    \geq
    \mathcal{L}_\gamma(\bm{\theta})
    + \frac{1}{\gamma}\log\!\left( \frac{1}{\alpha} \right).
\end{align}
\end{theorem}

The proof relies on restricting the moment-generating function to the lowest $\alpha$-fraction of outcomes and applying Jensen's inequality. This technique adapts the classical bounding method originally introduced by~\cite{AhmadiJavid2011} to our negative-tilt, lower-tail convention. The complete derivation is provided in Appendix~\ref{app:cvar_vs_tilt}.

We note that while the base QTL does not strictly bound CVaR on its own, this theorem proves that by adding the correction term shown in the inequality, it acts as a smooth lower-tail counterpart of CVaR. It links the standard mean-energy objective and the minimum-energy objective through exponential reweighting, rather than by truncating to the lowest $\alpha$-fraction of outcomes.

\subsubsection{Relation to alternative measures EVaR and VaR}

In risk theory, $\mathrm{CVaR}$ was introduced by Rockafellar and Uryasev~\cite{rockafellar2000optimization} as a coherent alternative to the traditional Value-at-Risk ($\mathrm{VaR}$), addressing its lack of subadditivity and insensitivity to outcomes beyond the risk threshold. We note that the bound in our previous theorem admits a direct connection to another well-established metric: the \emph{Entropic Value-at-Risk} ($\mathrm{EVaR}$) introduced by Ahmadi-Javid~\cite{AhmadiJavid2011}. 

By selecting the tilt parameter to match the extremal case, the right-hand side of our inequality becomes exactly the lower-tail Entropic Value-at-Risk defined as
\begin{align}
   \mathrm{EVaR}_\alpha(\bm{\theta}) &:= \sup_{\gamma < 0} \left\{ \mathcal{L}_\gamma(\bm{\theta}) + \frac{1}{\gamma}\log\!\left( \frac{1}{\alpha} \right) \right\}  \\
   &\leq  \mathrm{CVaR}_{\alpha}(\bm{\theta}) . \nonumber
\end{align}
This provides a Chernoff-type entropic lower bound on $\mathrm{CVaR}_\alpha$. Note that while the classical definition in~\cite{AhmadiJavid2011} uses an infimum over $\gamma > 0$, our supremum over $\gamma < 0$ is mathematically equivalent after a change of variables ($\gamma \mapsto -\gamma$) and aligns with the lower-tail convention standard in quantum optimization~\cite{Barkoutsos2020}.

Furthermore, $\mathrm{EVaR}_\alpha$ can be understood directly as the optimal threshold obtained by inverting the Chernoff bound. If we interpret the previously defined $\alpha$-quantile of the energy distribution as the lower-tail $\mathrm{VaR}_\alpha(\bm{\theta})$—a standard metric that emerged for market-risk reporting~\cite{pritsker1997evaluating}—the inverted Chernoff bound guarantees that $\Pr(E(\bm{\theta}) \le t) \le \alpha$ for all $t \le \mathrm{EVaR}_\alpha(\bm{\theta})$. This implies the strict ordering:
\begin{align}
\mathrm{EVaR}_\alpha(\bm{\theta}) \le 
\mathrm{CVaR}_\alpha(\bm{\theta}) \le
\mathrm{VaR}_\alpha(\bm{\theta}) .
\end{align}
Thus, in our lower-tail convention, $\mathrm{EVaR}$ provides a Chernoff-derived lower bound on $\mathrm{VaR}$, showing that optimizing the negative-tilt QTL smoothly controls the extreme low-energy tail of the ansätz.

\subsubsection{Tightness of QTL and CVaR via Gibbs States}

Returning to the quantum tilted loss, we further clarify its relationship with the lower-tail $\mathrm{CVaR}$ objective by showing that the inequality established in Theorem~\ref{thm:cvar_vs_tilt} is in fact essentially tight. This tightness is demonstrated through a concrete example that naturally arises in quantum settings: Gibbs states. While classical references like \cite{AhmadiJavid2011} do not provide explicit instances where the bound is tight, Gibbs states form a canonical family for which the behaviour of the tilted loss can be analysed exactly.

Intuitively, the Boltzmann weights of a Gibbs state already induce an exponential bias toward low-energy outcomes. This structure aligns perfectly with the exponential reweighting generated by a negative tilt. 

To formalize this, consider a Hamiltonian $O = \sum_{i=1}^d E_i |E_i\rangle\!\langle E_i|$ with ordered eigenvalues $E_{\min} = E_1 \le \cdots \le E_{\max} := E_d$. For an inverse temperature $\beta>0$, we define the Gibbs state as $\rho_\beta := e^{-\beta O}/Z(\beta)$. When measuring $O$ on $\rho_\beta$, we obtain a classical random energy $X$, with the ground state probability given by $p_1 := \langle E_1|\rho_\beta|E_1\rangle$. We can evaluate the tilted loss specifically for this state as:
\be
    \mathcal{L}_\gamma(\rho_\beta) := \frac{1}{\gamma}\log \Tr\!\big(e^{\gamma O}\rho_\beta\big).
\ee

By pushing the tilt to the large negative limit, we show that the entropic bound converges exactly to the ground state energy, matching the $\mathrm{CVaR}$ for sufficiently small tail probabilities. We provide the formal proof of this Theorem in Appendix~\ref{app:gibbs_proof}.

\begin{theorem}[Tightness of the entropic $\mathrm{CVaR}$ bound for Gibbs states]\label{thm:gibbs_tight}
Under the definitions above, for any confidence level $\alpha \in (0,p_1]$, the $\mathrm{CVaR}$ inequality becomes tight for Gibbs states $\rho_\beta$ as
\be
    \mathrm{CVaR}_\alpha(X) = E_{\min} = \sup_{\gamma<0} \left\{ \mathcal{L}_\gamma(\rho_\beta) + \frac{1}{\gamma}\log\!\left(\frac{1}{\alpha}\right) \right\}.
\ee
\end{theorem}

\subsubsection{Empirical CVaR and tilted loss for diagonal Hamiltonians}

For a general Hamiltonian $O$, evaluating the tilted loss $\mathcal{L}_\gamma(\bm{\theta}) = \frac{1}{\gamma}\log\Tr(e^{\gamma O}\rho(\bm{\theta}))$ is typically hard because computing $e^{\gamma O}$ requires diagonalising $O$. However, in variational experiments where $O$ is diagonal in the computational basis, $O = \sum_{z\in\{0,1\}^n} o(z)\,|z\rangle\!\langle z|$, each measurement outcome corresponds directly to an eigenvalue of $O$.

Suppose we prepare a parametrised quantum state $\rho(\bm{\theta})$ on $n$ qubits and measure it in the computational basis. Repeating this procedure $K$ times yields independent bit strings $z_k$, generating $K$ independent energy samples $\{E_k(\bm{\theta}) = o(z_k)\}_{k=1}^K$ drawn from the measurement-induced distribution $P_{\bm{\theta}}$. 

From these shared samples, the empirical mean energy is $\widehat{E}_K = \frac{1}{K}\sum_{k=1}^K E_k(\bm{\theta})$, which satisfies $\widehat{E}_K \xrightarrow[K\to\infty]{\text{a.s.}} \mathbb{E}_{P_{\bm{\theta}}}[E]$ by the strong law of large numbers. We can analogously define the empirical estimators for both CVaR and the tilted loss:

\paragraph{Empirical $\mathrm{CVaR}$ loss.} Following~\cite{Barkoutsos2020}, we sort the sampled energies in increasing order, $E_{(1)}(\bm{\theta}) \le E_{(2)}(\bm{\theta}) \le \cdots \le E_{(K)}(\bm{\theta})$, and define the empirical lower-tail $\mathrm{CVaR}$ estimator at level $\alpha\in(0,1)$ as
\begin{align}
\widehat{\mathrm{CVaR}}_{\alpha,K}(\bm{\theta})
    := \frac{1}{\lceil \alpha K\rceil}
        \sum_{i=1}^{\lceil \alpha K\rceil} E_{(i)}(\bm{\theta}).
\end{align}

\paragraph{Empirical tilted loss.} Focusing on the negative $\gamma < 0$ regime for energy minimization, we define the empirical tilted loss directly from the same $K$ energy samples as
\begin{align}
\widehat{\mathcal{L}}_{\gamma,K}(\bm{\theta})
    := \frac{1}{\gamma}\log\!\left(
        \frac{1}{K}\sum_{k=1}^K e^{\gamma E_k(\bm{\theta})}
    \right).
\end{align}

By the law of large numbers, both $\widehat{\mathrm{CVaR}}_{\alpha,K}(\bm{\theta})$ and $\widehat{\mathcal{L}}_{\gamma,K}(\bm{\theta})$ converge almost surely to their exact counterparts as $K\to\infty$. With both estimators constructed from the same empirical distribution, we can formulate the empirical analogue of the exact $\mathrm{CVaR}$--tilted-loss inequality.

\begin{theorem}[Empirical lower-tail CVaR vs.\ tilted loss]\label{thm:empirical_cvar_tilted}
Let $\alpha\in(0,1)$, let $\rho(\bm{\theta})$ be any parametrised quantum state, and let $O$ be a Hamiltonian diagonal in the computational basis. Measuring $\rho(\bm{\theta})$ in this basis yields $K$ independent energy samples $\{E_k(\bm{\theta})\}_{k=1}^K$ from $P_{\bm{\theta}}$. Then for every $\gamma<0$,
\begin{align}
\widehat{\mathrm{CVaR}}_{\alpha,K}(\bm{\theta})
\;\ge\;
\widehat{\mathcal{L}}_{\gamma,K}(\bm{\theta})
    + \frac{1}{\gamma}\log\!\left(\frac{1}{\alpha}\right).
\end{align}
\end{theorem}

Thus the empirical estimators satisfy the same ordering as the exact quantities in Theorem~\ref{thm:cvar_vs_tilt}, and we provide the corresponding proof in Appendix~\ref{app:empirical_cvar_tilt}.

\subsection{Relation to Petz–\Renyi $\alpha$-relative entropy}
We explore the connection between the Quantum Tilted Loss (QTL) and the family of Petz-\Renyi relative entropies, which are a cornerstone of quantum information theory. The Petz–\Renyi $\alpha$-relative entropy of order $\alpha \in(0,1)\cup (1,+\infty)$ is defined as
\begin{align}
D_\alpha(\sigma \|\rho) := \frac{1}{\alpha-1} \log [\tr (\sigma^\alpha \rho^{1-\alpha})].
\end{align}

A direct comparison between the QTL and this entropy is possible if we consider the specific case where the target operator is given by the logarithm of a positive operator $\sigma$, i.e., $O = \log\sigma$. Under this assumption, the QTL takes the form
\begin{align}
\mathcal{L}_\gamma(\log \sigma,\rho) = \frac{1}{\gamma}\log\tr (\sigma^{\gamma} \rho) .
\end{align}
This expression can be interpreted as a tilted, un-normalized version of the Petz–Rényi $\alpha$-relative entropy, measuring the overlap between the data state $\rho$ and the model state $\sigma$. However, unlike relative entropy, the QTL is not symmetric in its arguments and does not necessarily measure a "distance" between states. Despite these structural differences, the connection between the two quantities is mathematically rigorous: the QTL is strictly bounded above by the Petz-Rényi divergence for positive tilts. We formalize this relationship in the following theorem.

\begin{theorem}\label{thm:QTL-vs-petz}
Let $\rho,\sigma\in\mathcal S$ be density operators and let
$\gamma\in(0,1)\cup(1,\infty)$. Define
\be
L_\gamma(\log\sigma,\rho)
:=
\frac{1}{\gamma}\log\Tr(\sigma^\gamma\rho).
\ee
If $\gamma>1$, assume 
$\operatorname{supp}(\sigma)\subseteq \operatorname{supp}(\rho)$ so that 
$D_\gamma(\sigma\Vert\rho)<\infty$. Then
\be\label{eq:QTL-petz1}
\mathcal{L}_\gamma(\log \sigma,\rho) \le \frac{\gamma-1}{\gamma}D_\gamma(\sigma\Vert\rho).
\ee
Equality holds iff $\rho$ is pure.
\end{theorem}
The proof is provided in Appendix~\ref{app:petz-renyi-tilt}.

\onecolumngrid

\begin{table}[h]
\centering
\caption{Comparison of objective functions used for trainability and optimization.}
\label{tab:objective_comparison}
\renewcommand{\arraystretch}{1.2}
\resizebox{\linewidth}{!}{%
\begin{tabular}{llllll}
\hline
\textbf{Dimension} & \textbf{Expectation Value} & \textbf{Local Cost} & \textbf{F-VQE / Filtering} & \textbf{CVaR} & \textbf{QTL (this work)} \\
\hline
Math definition 
& $\mathrm{tr}(O\rho)$ 
& $\mathrm{tr}(O_\text{local}\rho)$ 
& $\mathrm{tr}(f(O)\rho)$ 
& $\mathbb{E}[X \mid X \leq \mathrm{VaR}_{\alpha}]$ 
& $\frac{1}{\gamma}\log \mathrm{tr}(e^{\gamma O}\rho)$ \\
Core idea 
& Arithmetic averaging 
& Local measurements 
& Spectral filtering 
& Best $\alpha$-fraction 
& Exponential reweighting \\
Properties 
& Linear, analytic 
& Linear, analytic 
& Filter-dependent 
& Nonlinear, non-smooth 
& Nonlinear, smooth \\
Landscape 
& Often flat (BPs) 
& Trainable (shallow) 
& Sharpened by filter choice 
& Tail-selective, nonsmooth 
& Smooth, sharpened \\
Flexibility 
& None 
& Locality $k$ 
& Filter function $f$ 
& Tail parameter $\alpha$ 
& Tilt parameter $\gamma$ \\
\hline
\end{tabular}%
}
\end{table}

\twocolumngrid

\subsection{Comparative Analysis of Objective Functions}

To contextualize the Quantum Tilted Loss (QTL) framework within the landscape of objective design for variational quantum optimization, we compare it with several commonly used alternatives, summarized in Table~\ref{tab:objective_comparison}. The standard expectation-value objective, $\tr(O\rho)$, is the most natural and widely adopted choice. Its main advantages are its linearity, analytic simplicity, and straightforward physical interpretation. However, because it treats all sampled outcomes through uniform arithmetic averaging, it is particularly susceptible to landscape flattening. In expressive ans\"atze or deep circuits, this often manifests as barren plateaus, where gradients become exponentially small and optimization becomes ineffective.

Local cost functions of the form $\tr(O_{\mathrm{local}}\rho)$, where $O_{\mathrm{local}}$ acts non-trivially only on a small subset of qubits, significantly alleviate barren plateaus in shallow architectures by restricting measurements to local subsystems \cite{Cerezo2021,Uvarov_2021}. Nevertheless, this remedy is inherently problem-dependent. Its effectiveness relies on the existence of a meaningful local decomposition, rendering it unsuitable for tasks with fundamentally global structures, such as combinatorial optimization problems involving nonlocal constraints or parity-type observables.

Filtering VQE (F-VQE) provides a particularly close quantum analogue to QTL because it also modifies variational optimization through spectral reweighting rather than changing the ansätz itself. In F-VQE, a filter function or filtering operator is used to amplify low-energy components of the variational state while suppressing energetically unfavorable components, thereby biasing the optimization toward high-quality solutions~\cite{Amaro2022}. This allows F-VQE to increase the probability of sampling low-energy bit strings, which is often more important than simply reducing the average energy. However, the effectiveness of F-VQE depends on the choice and implementation of the filter, and explicit filtering may introduce additional circuit, sampling, or approximation overhead.

Another important non-linear alternative is $\mathrm{CVaR}$, which focuses optimization on the lowest $\alpha$-fraction of energy outcomes \cite{Barkoutsos2020,rockafellar2000optimization,Saem_2026}. This hard tail truncation introduces a nonsmooth dependence on the variational parameters: as $\theta$ changes, samples may enter or leave the empirical lower-tail set, and the active quantile threshold can shift discontinuously. Consequently, empirical CVaR gradients can be sensitive to finite-shot noise, sorting effects, and quantile ties. 
Additional objective-engineering strategies include subspace-search VQE, variational quantum deflation, and contextual-subspace VQE, which modify the variational objective or restrict the effective search space rather than reweighting the measured energy distribution \cite{Nakanishi2019,Higgott2019,Kirby2021}.

The QTL addresses these limitations by introducing a controllable exponential reweighting via the tilt parameter $\gamma$, permitting genuinely global observables. Small $|\gamma|$ recovers a regime close to standard expectation-value optimization, whereas larger $|\gamma|$ smoothly emphasizes selected spectral regions. In this sense, QTL interpolates between weak and strong risk sensitivity without introducing the hard truncations of $\mathrm{CVaR}$ or the explicit state-spectrum estimation required by free-energy objectives. Because the exponential reweighting acts directly on the observable's spectrum, the QTL admits closed-form analytic expressions for its gradients. While these can be evaluated on hardware via standard parameter-shift rules, the nonlinear normalization introduces specific sampling complexities that we formally address in Section~\ref{sec:estimation}. Ultimately, the QTL provides a unifying and flexible framework: it preserves the broad applicability of global objectives, avoids the discontinuities of truncation-based methods, and introduces a tunable mechanism for reshaping the optimization landscape in an analytically tractable manner.

\section{Estimability, limits and regimes of QTL}\label{sec:estimation}

As established in Lemma~\ref{lem:faithfulness}, the QTL faithfully preserves the global minimum of the target observable, serving as a rigorous drop-in replacement for the standard expectation-value objective in algorithms like VQE. To operationally implement this objective, however, it must be efficiently evaluable using finite quantum measurements. Compared with the linear quantity $\Tr(O\rho(\bm\theta))$, estimating the QTL is statistically and operationally more demanding because the nonlinear operator $e^{\gamma O}$ is not directly accessible from a standard Pauli decomposition of $O$. We therefore focus on three practically relevant regimes: structured observables, small-tilt perturbative expansions, and diagonal Hamiltonians. Technical proofs and refined concentration results are deferred to Appendix~\ref{app:qtl_estimation_appendix}.

\subsection{Exact reduction for finite-spectrum observables}

While estimating $\Tr(e^{\gamma O}\rho(\bm\theta))$ for a generic observable generally requires exponential resources, the exponential operator can fundamentally be mapped to a finite polynomial. For any Hermitian operator $O$ with exactly $r$ distinct eigenvalues, its minimal polynomial has degree $r$. Consequently, any analytic matrix function, including the exponential, can be exactly reduced modulo the minimal polynomial to a polynomial of degree at most $r-1$. 

Concretely, for every tilt parameter $\gamma\in\mathbb{R}$, there exists a polynomial $p_{r-1}$ such that
\begin{align}
e^{\gamma O} = p_{r-1}(O).
\end{align}
By linearity of the trace, the tilted partition function is exactly determined by the first $r-1$ moments of the observable as
\begin{align}
\Tr(e^{\gamma O}\rho)
=
\sum_{k=0}^{r-1} a_k(\gamma)\Tr(O^k\rho),
\end{align}
where the coefficients $a_k(\gamma)$ depend only on $\gamma$ and the spectrum of $O$.

However, this exact algebraic reduction is not unconditionally efficient. For a generic full-rank Hamiltonian, the number of distinct eigenvalues $r$ can scale with the Hilbert space dimension $d=2^n$, rendering the experimental measurement of $r-1$ high-degree moments $\Tr(O^k\rho)$ operationally prohibitive on near-term hardware. The practical utility of this reduction is therefore restricted to highly structured observables with spectral degeneracy. When $r$ is small and strictly independent of system size $n$—such as for rank-1 global projectors ($r=2$) or single Pauli strings ($r=2$)—this algebraic reduction collapses the estimation complexity of the QTL to that of standard, low-degree polynomial expectation values.

\subsection{Small-tilt perturbative expansion}

Because the exact algebraic reduction requires measuring high-order moments $\Tr(O^k\rho)$—a task that is resource-prohibitive for generic VQE Hamiltonians—we must adopt an alternative estimation strategy for standard applications. Rather than seeking an exact global representation of $e^{\gamma O}$, we can directly exploit the operational structure of the variational circuit and restrict our focus to the regime of small tilt parameters.

We assume standard hardware access to a parameterized ansätz state prepared by a quantum circuit,
\begin{align}
\rho(\bm\theta)=U(\bm\theta)\rho_0U^\dagger(\bm\theta),
\end{align}
and consider a generic target observable given by a Pauli decomposition $O=\sum_j c_j P_j$. When $|\gamma|$ is small, the QTL admits a tractable perturbative expansion. If $X$ denotes the classical random variable obtained by measuring $O$ in its eigenbasis on $\rho(\bm\theta)$, then formally
\begin{align}
\mathcal{L}_\gamma(O,\rho(\bm\theta))
=
\frac{1}{\gamma}\log \mathbb{E}[e^{\gamma X}],
\end{align}
so $\gamma\mathcal{L}_\gamma$ acts as the cumulant-generating function of $X$. Truncating the expansion at the second order yields
\begin{align}\label{eq:qtl_small_gamma_main}
\mathcal{L}_\gamma(O,\rho(\bm\theta))
=
\Tr(O\rho(\bm\theta))
+
\frac{\gamma}{2}\text{Var}_{\rho(\bm\theta)}(O)
+
\mathcal{O}(\gamma^2).
\end{align}
Thus, in the small-tilt regime, estimating the QTL reduces to estimating the first two moments of $O$ directly on the parameterized state. For a Pauli expansion $O=\sum_j c_jP_j$, the first moment is directly accessible, while the variance
\begin{align}
\text{Var}_{\rho(\bm\theta)}(O)
=
\sum_{j,k} c_jc_k\,\Tr(P_jP_k\rho(\bm\theta)) - \big(\Tr(O\rho(\bm\theta))\big)^2
\end{align}
can be estimated from standard Pauli measurements, with a quadratically growing number of terms.

\subsection{Finite tilt in the diagonal case}
\label{subsecFiniteTilt}
A nonperturbative but still tractable regime arises when $O$ is diagonal in the
computational basis,
\be
O=\sum_{z\in\{0,1\}^n} E(z)\ket{z}\!\bra{z}.
\ee
Measuring $\rho(\bm\theta)$ in that basis produces samples
$z\sim p_{\bm\theta}(z):=\bra{z}\rho(\bm\theta)\ket{z}$, so that
\be \label{eqTiltDiagonal}
Z_\gamma(\bm\theta):=\Tr(e^{\gamma O}\rho(\bm\theta))
=
\mathbb{E}[e^{\gamma X}],
\quad X:=E(z).
\ee
Hence QTL becomes a classical moment-generating-function estimation problem.

Given i.i.d.\ samples $X_1,\dots,X_m$, define
\be
\widehat Z_\gamma:=\frac1m\sum_{i=1}^m e^{\gamma X_i},
\qquad
\widehat{\mathcal L}_\gamma:=\frac1\gamma\log \widehat Z_\gamma.
\ee
Let $\lambda_{\min}\le X\le \lambda_{\max}$ and
$\Delta:=\lambda_{\max}-\lambda_{\min}$.
Then standard concentration implies the worst-case scaling
\be\label{eq:qtl_main_scaling}
m
=
\tilde O\!\left(
\frac{(e^{|\gamma|\Delta}-1)^2}{\gamma^2\varepsilon^2}
\log\frac1\delta
\right)
\ee
for estimating $\mathcal{L}_\gamma$ to additive error $\varepsilon$ with
failure probability at most $\delta$.
A variance-adaptive refinement, which replaces pure range dependence by
$\text{Var}(e^{\gamma X})$, is given in
Appendix~\ref{app:diagonal_qtl_bounds}.

\cref{eq:qtl_main_scaling} shows that large finite tilt may be
statistically expensive in the worst case.
Accordingly, polynomial estimability requires either additional structure or a
regime in which $|\gamma|\Delta$ remains moderate.

\subsection{Gradient estimation of QTL}
\label{subsecGradientQTL}
For variational optimization, one must estimate both the objective value and its gradient. For QTL, we obtain
\be
\partial_{\theta_k}\mathcal{L}_\gamma(\bm\theta)
=
\frac{1}{\gamma}\,
\frac{\partial_{\theta_k}Z_\gamma(\bm\theta)}{Z_\gamma(\bm\theta)}.
\label{eq:qtl-grad-chain}
\ee

For the usual expectation-value loss, the parameter-shift rule applies directly because the loss itself is an expectation value of a fixed observable \cite{mitarai2018,schuld2019,Wierichs2022generalparameter}. By contrast, QTL is nonlinear in $\rho(\bm\theta)$ because of the outer logarithm, so the shift rule does not apply directly to $\mathcal{L}_\gamma(\bm\theta)$. The key observation is that $Z_\gamma(\bm\theta)$ still has the form of an expectation value, now for the observable $e^{\gamma O}$, and therefore admits the standard parameter-shift representation.

Specifically, if $\theta_k$ enters through a gate $e^{-i\theta_k V_k}$ with
$V_k^2=v_k^2 \mathbb I$, and if
\be
\bm\theta_k^\pm:=\bm\theta\pm s_k\mathbf e_k,
\qquad
s_k:=\frac{\pi}{4v_k},
\ee
then
\be
\partial_{\theta_k}Z_\gamma(\bm\theta)
=
v_k\Bigl(
Z_\gamma(\bm\theta_k^+)-Z_\gamma(\bm\theta_k^-)
\Bigr).
\label{eq:qtl-z-shift-main}
\ee
Hence
\be
\partial_{\theta_k}\mathcal{L}_\gamma(\bm\theta)
=
\frac{v_k}{\gamma}\,
\frac{
Z_\gamma(\bm\theta_k^+)-Z_\gamma(\bm\theta_k^-)
}{
Z_\gamma(\bm\theta)
}.
\label{eq:qtl-grad-shift-main}
\ee
Thus, unlike the expectation-value case, the QTL gradient is not a simple central difference of $\mathcal{L}_\gamma$ itself, but a normalized difference of shifted tilted partition function values. The full statement and proof are given in Appendix~\ref{app:qtl-parameter-shift}.

\cref{eq:qtl-grad-shift-main} also shows that the sensitivity of the gradient estimate depends on the size of $Z_\gamma(\bm\theta)$: smaller values of $Z_\gamma(\bm\theta)$ lead to stronger amplification of estimation errors through the normalization. In contrast to $\mathrm{CVaR}$, which relies on empirical sorting and does not admit an exact parameter-shift rule, the QTL gradient can be evaluated exactly via shifted circuit evaluations.

\section{Trainability: Landscape Reshaping Mechanisms and Benchmark}\label{sec:trainability}

QTL cannot eliminate barren plateaus in variational optimization. QTL is able to reshape local optimization geometry through a nonlinear transformation of the objective, and, in specific structured settings, to improve gradient signal strength. This section develops that perspective in stages: we first analyze a structured projector benchmark where this landscape reshaping can be explicitly visualized and mathematically quantified (Sec.~\ref{subsec:benchmark}), and then explain the underlying geometric mechanism via Hessian analysis (Sec.~\ref{sec:hessian}).

Before proceeding, it is useful to distinguish two conceptually different aspects of trainability that are often conflated in the literature \cite{Anschuetz2022}:

\begin{itemize}
\item \textbf{Barren plateaus} concern the scaling of first-order gradient signals under random initialization. Formally, a cost function exhibits a barren plateau if $\mathrm{Var}_{\bm\theta}[\partial_k \mathcal{L}]$ vanishes exponentially in the system size $n$ for parameters $\bm\theta$ drawn from the initialization distribution~\cite{McClean2018}. This is a statement about the typical signal strength in a randomly chosen region of parameter space.

\item \textbf{Landscape reshaping} concerns how the objective function modifies local curvature, conditioning, and the geometry of optimization trajectories. Two objectives with the same minimizers can have very different Hessian spectra, condition numbers, and basin structures, leading to qualitatively different optimizer behavior.
\end{itemize}

\noindent
QTL acts directly on the latter: its logarithmic nonlinearity redistributes curvature and rescales gradients in a state-dependent manner, without modifying the circuit ansätz itself. As we show below, this reshaping can provably improve gradient-variance scaling in a structured shallow model, provided the tilt schedule scales with the system size. Conversely, for any fixed, constant tilt, the gradient variance still vanishes exponentially, adhering to standard barren plateau bounds. More critically, even when an aggressive tilt schedule successfully restores the landscape gradient, it triggers a statistical penalty. We formally establish this trade-off between trainability and accuracy at the end of this section (Sec.~\ref{sec:gradient_resolvability}), demonstrating how exponential tilt ultimately shifts the optimization bottleneck from the optimization landscape to the finite-shot measurement sample complexity.

\subsection{Projector benchmark: visual intuition and exact analysis}\label{subsec:benchmark}

We can directly observe and rigorously quantify the landscape reshaping of the QTL using a structured projector benchmark. The key enabling structure is that the observable is a global projector, which makes the tilted loss an explicit scalar transform of the standard cost.

The target state is $|\psi_{\rm target}\rangle = |0\rangle^{\otimes n}$, and the variational ansätz is the tensor-product circuit
\begin{align}
V(\bm\theta)=\bigotimes_{j=1}^n e^{i\theta_j \sigma_x^{(j)}/2}, \quad |\psi(\bm\theta)\rangle:=V(\bm\theta)|0\rangle^{\otimes n}.
\end{align}
The observable is the global projector
\begin{align}
O_G:=\mathbb I-|0\rangle\langle 0|^{\otimes n},
\end{align}
with the corresponding standard global loss $C_G(\bm\theta):=\Tr\!\big(O_G\rho(\bm\theta)\big)$, where $\rho(\bm\theta):=|\psi(\bm\theta)\rangle\langle\psi(\bm\theta)|$. We replace this expectation loss by the QTL
\begin{align}
\mathcal L_\gamma(O_G,\rho(\bm\theta)) := \frac{1}{\gamma}\log \Tr\!\big(e^{\gamma O_G}\rho(\bm\theta)\big).
\end{align}

The qualitative effect of the QTL on this optimization landscape is immediately apparent through direct visualization. Figure~\ref{fig:LossLandscape} plots a simplified two-parameter slice (using $R_y(\theta_1)R_x(\theta_2)$ generators) of this global cost alongside its QTL counterpart. As $n$ grows, the standard global loss becomes increasingly flat over most of the parameter space—a visual manifestation of the exponentially vanishing gradients that characterize barren plateaus. 

\begin{figure*}[t]
  \vskip 0.2in
  \begin{center}
  \centerline{\includegraphics[width=\textwidth]{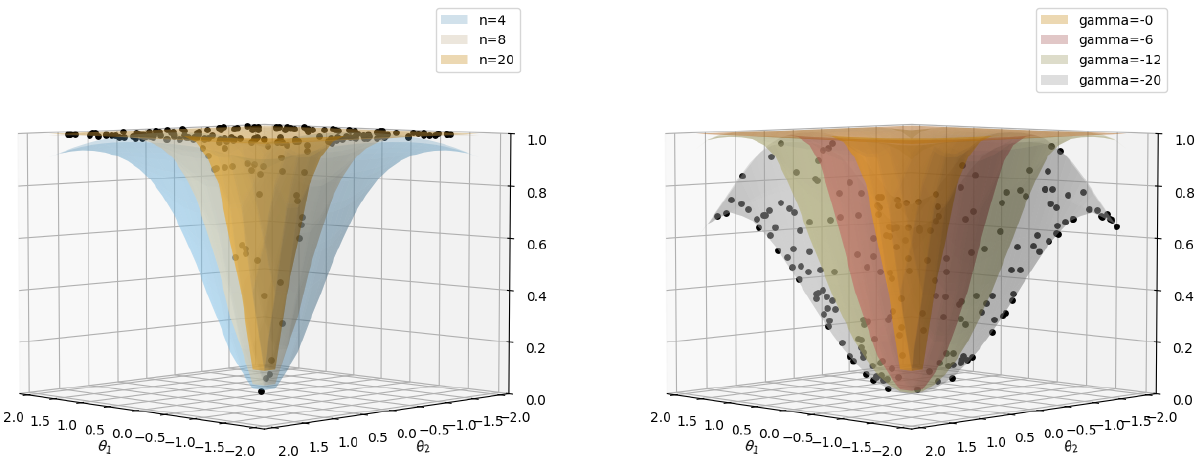}}
    \caption{Loss landscapes for a simple two-parameter ansätz $R_y(\theta_1)R_x(\theta_2)$ with 200 uniformly sampled parameter points (black dots). \textbf{Left:} standard global loss for $n=4,8,20$. \textbf{Right:} The corresponding global QTL for $n=20$ with tilt parameters $\gamma=0,-6,-12,-20$. As $\gamma$ decreases, the finite-size landscape is increasingly reshaped, with enhanced contrast between low-loss and high-loss regions.}
    \label{fig:LossLandscape}
  \end{center}
\end{figure*}

When a negative tilt is applied, the landscape is structurally reshaped: contrast between low-loss and high-loss regions is enhanced, and the basin around the optimum becomes sharply pronounced. The reshaping is highly nonuniform, compressing high-loss plateaus while deepening low-loss valleys, consistent with the logarithmic nonlinearity acting more strongly where the objective value is small. Crucially, the position of the global minimum is strictly preserved—the tilted and untilted landscapes share the exact same optimizer.

These visual observations are mathematically formalized by the structural properties of the projector model.

\begin{proposition}[Properties of projector QTL benchmark]
\label{prop:warmup_projector_qtl}
Let $O_G = \mathbb{I} - |0\rangle\langle 0|^{\otimes n}$ and the variational ansätz be the tensor-product circuit $V(\bm\theta)=\bigotimes_{j=1}^n e^{i\theta_j \sigma_x^{(j)}/2}$. This benchmark model has the following properties:
\begin{enumerate}
    \item Reduction: The projector QTL can be exactly reduced to a scalar transformation of the standard global cost $C_G(\bm\theta) = 1 - \prod_{j=1}^n \cos^2(\theta_j/2)$, given by
\begin{align}
\mathcal L_\gamma(O_G,\rho(\bm\theta)) = \phi_\gamma\!\big(C_G(\bm\theta)\big),
\end{align}
where the nonlinear reshaping function $\phi_\gamma: [0,1] \to \mathbb{R}$ is defined as
\begin{align}
\phi_\gamma(x) := \frac{1}{\gamma}\log\!\Big(1+(e^\gamma-1)x\Big).
\end{align}
    \item Monotonicity: The function $\phi_\gamma$ is strictly increasing on $[0,1]$ with $\phi_\gamma(0) = 0$ and $\phi_\gamma(x) > 0$ for $x > 0$. That is, the QTL is an exact scalar post-processing of the ordinary expectation value $C_G(\bm\theta)$---a monotone but nonlinear reshaping of the standard global cost.
    \item Faithfulness: $\mathcal L_\gamma(O_G,\rho(\bm\theta)) = 0 \iff C_G(\bm\theta) = 0$, so the tilted and untilted losses share the same global minimizers.
\end{enumerate}
\end{proposition}

\noindent
The proof is given in Appendix~\ref{app:proof_warmup_projector_qtl}. The reduction property explains the visual preservation of the minima: in this benchmark model, the tilted loss is an exact scalar post-processing of the ordinary expectation value $C_G(\bm\theta)$. Consequently, the loss itself can be estimated at the exact same cost as the standard global loss, regardless of how large $|\gamma|$ is. We return to this point in Sec.~\ref{sec:gradient_resolvability}, where the distinction between loss estimation and gradient estimation plays a central role.

Having established the structural properties of the projector QTL, we now turn to its gradient-variance scaling. Since $\mathcal L_\gamma = \phi_\gamma(C_G)$, the gradient is $\partial_k\mathcal L_\gamma = \phi_\gamma'(C_G)\,\partial_k C_G$, where the state-dependent prefactor $\phi_\gamma'(C_G(\bm\theta))$ is what enables the tilt to reshape gradient scaling. The following theorem identifies two distinct regimes.

\begin{theorem}[Gradient-variance regimes in the projector QTL benchmark]
\label{thm:critical_tilt_transition}
Consider the projector benchmark model of Proposition~\ref{prop:warmup_projector_qtl} under i.i.d.\ uniform random initialization $\theta_j \sim \mathrm{Unif}(-\pi,\pi)$.
\begin{itemize}
\item Fixed small tilt: For fixed $|\gamma| \ll 1$, the gradient variance of $\mathcal L_\gamma$ exhibits the same exponential-in-$n$ decay as the standard global cost $C_G$:
\begin{align}
\mathrm{Var}[\partial_{\theta_k}\mathcal L_\gamma]
=
\frac{1}{8}\Big(\frac{3}{8}\Big)^{n-1}\big(1+\mathcal{O}(\gamma)\big).
\end{align}

\item Linear negative tilt schedule: There exists an explicit tilt schedule $\gamma(n)=\Theta(-n)$ under which the gradient variance admits a polynomial lower bound in $n$, in contrast to the exponential upper bound at fixed tilt.
\end{itemize}
More concretely, for the explicit choice
\begin{align}
\gamma(n)=2(n-1)\log(3/8),
\end{align}
the appendix proves the lower bound
\begin{align}
\mathrm{Var}[\partial_{\theta_k}\mathcal L_{\gamma(n)}]
=
\Omega\!\left(\frac{1}{n^2}\right).
\end{align}
\end{theorem}

\noindent
The proof is given in Appendix~\ref{app:train}. The result relies on three specific ingredients: (i) the projector structure of $O_G$, which makes $e^{\gamma O_G}$ exactly computable; (ii) the shallow product ansätz, which makes the product $\prod_j \cos^2(\theta_j/2)$ explicit; and (iii) the $n$-dependent tilt schedule, which tunes the compensation. Removing any of these weakens or invalidates the conclusion. The resulting $\Omega(1/n^2)$ scaling is comparable to that typically associated with local cost functions, suggesting a local-cost-like scaling mechanism in this benchmark model. However, this is a scaling resemblance rather than a formal equivalence, since the connection to operator-theoretic notions of locality has not been established. Accordingly, Theorem~12 should be viewed as a structured benchmark result rather than a universal barren-plateau avoidance theorem. Its broader value is to isolate a concrete mechanism by which nonlinear tilting can compensate for exponentially small gradients when observable structure, circuit structure, and tilt schedule are properly aligned. This motivates turning from first-order variance scaling to the more general geometric question of how QTL reshapes local curvature and conditioning.

The projector benchmark above identifies one explicit mechanism by which QTL can improve gradient scaling, but its assumptions are deliberately restrictive. More generally, the role of tilting is geometric rather than universal: even when it does not eliminate barren plateaus in full generality, it can still reshape local curvature, anisotropy, and conditioning. This is captured naturally by the Hessian decomposition.

\subsection{Hessian decomposition and geometric mechanism}\label{sec:hessian}

Second-order structure gives a more refined view of the QTL landscape than gradients alone. In classical nonlinear optimization, the Hessian defines the local quadratic model, and its spectrum controls curvature, conditioning, and the behavior of Newton/quasi-Newton-type methods. In modern nonconvex machine learning, Hessian spectra have also been used to diagnose saddle plateaus, sharp/flat directions, and the evolution of optimization trajectories. Hessian information is often used to distinguish sharp and flat directions, identify ill-conditioning, and explain why optimization trajectories behave differently even when first-order signals look similar \cite{Nocedal1999,1Yann2014,sagun2018,ghorbani2019}.

Recently, tilted sharpness-aware minimization shows that exponential tilting can reweight local perturbation neighborhoods and favor flatter solutions in classical nonconvex landscapes, reinforcing our use of tilting as a landscape-reshaping operation rather than only a tail-risk measure~\cite{li2025sharp}. 
In variational quantum optimization, however, this perspective comes with an extra caveat: higher-order derivatives may also become hard to resolve statistically in barren-plateau regimes~\cite{Cerezo_2021}. Thus, for QTL, it is useful to separate two questions: how the tilt modifies the quantum local geometry, and when this modified curvature remains observable in practice.

Differentiating QTL once gives
\begin{align}
\nabla \mathcal L_\gamma(\bm\theta)
=
\frac{1}{\gamma\, Z_\gamma(\bm\theta)}\,\nabla Z_\gamma(\bm\theta),
\end{align}
and differentiating again yields
\begin{align}
\label{eq:qtl_hessian_clean}
\nabla^2 \mathcal L_\gamma(\bm\theta)
=
\frac{1}{\gamma\, Z_\gamma(\bm\theta)}\,\nabla^2 Z_\gamma(\bm\theta)
-\frac{1}{\gamma}
\frac{\nabla Z_\gamma(\bm\theta)\nabla Z_\gamma(\bm\theta)^\top}{Z_\gamma(\bm\theta)^2}.
\end{align}
Equivalently,
\begin{align}
\label{eq:qtl_hessian_grad_form}
\nabla^2 \mathcal L_\gamma(\bm\theta)
=
\frac{1}{\gamma\, Z_\gamma(\bm\theta)}\,\nabla^2 Z_\gamma(\bm\theta)
-\gamma\,
\nabla \mathcal L_\gamma(\bm\theta)\nabla \mathcal L_\gamma(\bm\theta)^\top .
\end{align}
This form makes the geometry transparent: the Hessian of QTL is the Hessian of the tilted partition function, rescaled by $1/(\gamma Z_\gamma)$, together with a rank-one correction aligned with the gradient direction.

Taking the quadratic form in an arbitrary direction $v\in\mathbb R^m$ gives
\begin{align}
\label{eq:qtl_directional_curvature}
v^\top \nabla^2 \mathcal L_\gamma(\bm\theta)\, v
=
&\frac{1}{\gamma\, Z_\gamma(\bm\theta)}
\,v^\top \nabla^2 Z_\gamma(\bm\theta)\, v \nonumber\\
&-\gamma\,
\bigl(\langle v,\nabla \mathcal L_\gamma(\bm\theta)\rangle\bigr)^2 .
\end{align}
Here we used the elementary identity
$
v^\top (aa^\top) v = (a^\top v)^2
$.
Equation~\cref{eq:qtl_directional_curvature} shows that QTL does not simply rescale curvature uniformly. Instead, it changes curvature most strongly in directions that overlap with the gradient. When $\gamma>0$, the second term in~\cref{eq:qtl_directional_curvature} is non-positive, so QTL suppresses curvature along gradient-aligned directions. When $\gamma<0$, the sign is reversed, and the same term enhances curvature in those directions. In both cases, directions orthogonal to $\nabla \mathcal L_\gamma(\bm\theta)$ are unaffected by the rank-one correction. This is the basic geometric mechanism behind the anisotropic landscape reshaping induced by the tilt.

A particularly simple picture emerges at stationary points. If $\bm\theta^\star$ satisfies
$
\nabla \mathcal L_\gamma(\bm\theta^\star)=0
$,
equivalently
$
\nabla Z_\gamma(\bm\theta^\star)=0
$,
then the rank-one term vanishes and
\begin{align}
\label{eq:qtl_hessian_stationary}
\nabla^2 \mathcal L_\gamma(\bm\theta^\star)
=
\frac{1}{\gamma\, Z_\gamma(\bm\theta^\star)}
\,\nabla^2 Z_\gamma(\bm\theta^\star).
\end{align}
Hence, at a stationary point, QTL preserves the eigendirections of $\nabla^2 Z_\gamma$; only the eigenvalues are rescaled by the scalar factor $1/(\gamma Z_\gamma(\bm\theta^\star))$. Away from stationarity, by contrast, the additional rank-one term changes the local quadratic model in a direction-dependent way.

This decomposition also clarifies the limits of what QTL can guarantee. 
A finite tilt can reshape the local second-order geometry, even in regimes where the first-order barren-plateau phenomenon persists. 
However, the identity in Eq.~\eqref{eq:qtl_hessian_grad_form} does not by itself imply favorable conditioning, nor does it guarantee that the modified curvature can be estimated efficiently. 
Recent Hessian analyses for VQAs emphasize this distinction between geometric structure and statistical resolvability: Hessian entries can admit exact second-order parameter-shift representations, while still becoming prohibitively costly to resolve for global objectives; by contrast, locality can preserve polynomially accessible curvature information~\cite{Huang2026}. 
The Hessian decomposition should therefore be understood as a local geometric mechanism, not as a standalone guarantee of efficient curvature-based optimization. 
More generally, noise--precision relations for VQAs show that the observable spectrum, Hessian geometry, and ansätz structure jointly affect the noise level required to achieve a target cost precision, paralleling QTL's tilt-dependent sampling overhead~\cite{Kosuke2023}. 
Determining the full Hessian spectrum, condition number, or optimizer convergence rate requires additional assumptions on the circuit ensemble, observable structure, and finite-shot estimation budget.

\subsection{Gradient Resolvability and Signal-to-Noise Ratio}\label{sec:gradient_resolvability}

While the Hessian analysis shows that QTL can geometrically reshape the cost landscape, a steeper landscape is a necessary but insufficient condition for successful variational training. This motivates analyzing not only whether a loss reshapes the optimization landscape, but also whether the resulting objective values and gradients can be reliably estimated from finite quantum measurements. 
Operationally, gradients are estimated stochastically from finite measurement samples. For gradient descent to navigate the landscape effectively, the true gradient signal must be statistically resolvable against measurement noise. This signal-to-noise perspective is standard in barren-plateau analyses, where trainability is often characterized by whether the gradient variance remains large enough to be distinguished from finite-shot fluctuations~\cite{Wang2021}.

We formally quantify this resolvability through the noise-to-signal ratio, 
\begin{align}
\mathcal{R} = \frac{\sigma^2_{\text{shot}}}{N \|\nabla \mathcal{L}_\gamma\|^2 },
\end{align}
where $\sigma^2_{\text{shot}}$ is the single-shot variance of the estimator and $N$ is the number of measurements. In a standard barren plateau, the signal $\|\nabla \mathcal{L}_\gamma\|^2$ vanishes exponentially, demanding an exponentially large shot count $N$ to maintain the resolvability condition $\mathcal{R} \lesssim 1$. 

The QTL addresses this by amplifying the gradient signal via the tilt parameter $\gamma$. However, this geometric reshaping does not come for free. Because the stochastic estimator for the QTL gradient relies on the exponentially reweighted observable $e^{\gamma O}$, its single-shot variance $\sigma^2_{\text{shot}}$ scales exponentially with $|\gamma|$. Consequently, if an aggressive tilt $|\gamma| = \Omega(n)$ is deployed to unconditionally avoid a barren plateau, the statistical variance of the estimator explodes. This effectively transfers the exponential bottleneck from the vanishing landscape gradients to the measurement sampling complexity.

This exposes a fundamental trainability-estimability trade-off for near-term variational algorithms: while exponential tilting geometrically reshapes the landscape to amplify gradient signals, resolving those non-linear gradients incurs a statistical overhead that scales with the strength of the tilt. Recent finite-sample analyses of nonlinear post-processing caution that nonlinear objectives can appear to improve barren-plateau scaling while remaining exponentially hard to estimate, which motivates our explicit trainability–estimability trade-off analysis \cite{Saem_2026}.  Crucially, this bottleneck is not unique to the QTL; it reflects a statistical penalty for any objective function that aggressively biases the measurement distribution. For example, the empirical $\mathrm{CVaR}$ objective~\cite{Barkoutsos2020} sharpens the optimization landscape by discarding all but the lowest $\alpha$-fraction of outcomes. As the tail parameter $\alpha$ is tightened to isolate the ground-state signal, the effective sample size is proportionally reduced, demanding a strict increase in the total measurement budget simply to extract a statistically reliable gradient direction. We provide the information-theoretic derivation of this sample complexity bound, its connection to hypothesis testing, and the operational distinction between near-term sampling and fault-tolerant estimation regimes in Appendix~\ref{app:gradient_resolvability}.

\section{Tilted-QAOA for MaxCut problem}
\label{sec:qaoa-maxcut-qtl}
In this section we conduct numerical experiments to evaluate the performance of our proposed quantum tilted loss on a prominent variational optimization task known as the Quantum Approximate Optimization Algorithm (QAOA) to solve the MaxCut problem ~\cite{farhi2014}. Although QAOA algorithm and its variants have been extensively studied both in the literature and in experiments~\cite{farhi2016qaoa, zhou2020qaoa, rabinovich2022ionnativeqaoa, adamattibridi2026expressivity}, their trainability behaviors in the presence of barren plateaus are characterized by the expressiveness of the circuit ansätz and the problem specifications, via the framework of dynamical Lie algebras~\cite{Larocca2022,fontana2024characterizing, ragone2024dla, mao2025qaoa}. Nonetheless, few works have focused on the performance of the QAOA algorithm with loss functions distinct to the expectation forms. 

Consider a graph $G=(V,E)$ with $|V|=n$ nodes. A cut is modeled by a string $x \in \{0,1\}^{n}$. For any two nodes that are connected by an edge $(i,j)$, both $i$ and $j$ are in the same set if $x_i=x_j$, while $(i,j)$ is in the cut if and only if $x_i\neq x_j$. Our task of finding the maximum number of edges that are cut by a partition can be modeled by the cut function $\text{Cut}(x)=\sum_{(i,j)\in E}x_i+x_j-2x_ix_j$. This problem is then equivalent to finding the ground state of the following diagonal cost Hamiltonian
\be
H_C = \sum_{(i,j)\in E} \frac{1}{2}\left(Z_i Z_j - \mathbb I\right),
\ee
where $Z_i$ are Pauli Z operators. In the QAOA algorithm, a mixer Hamiltonian is chosen as $H_M=\sum_{i\in V}X_i$, where $X_i$ is the Pauli X operator acting on qubit $i$. Let $p\in \bb{N}^{+}$ and $\bm\theta, \bm\tau \in[0, 2\pi)^{\otimes p}$. We then construct the depth-$p$ QAOA ansätz as a density matrix given by $\Psi_p(\bm\theta, \bm\tau)=\ket{\psi_p(\bm\theta, \bm\tau)}\bra{\psi_p(\bm\theta, \bm\tau)}$, where

\be
\ket{\psi_p(\bm\theta, \bm\tau)} = U(\theta_p, \tau_p) \cdots U(\theta_1, \tau_1) |+\rangle^{\otimes n},
\ee
where
\be
U(\theta_k, \tau_k) = U_M(\tau_k) U_C(\theta_k) = e^{-i\tau_k H_M} e^{-i\theta_k H_C},
\ee
for every layer $k\in[p]$. A schematic illustration of the QAOA circuit is shown below.
\be
\scalebox{0.85}{\Qcircuit @C=0.3em @R=0.5em{
\lstick{|0\rangle} & \gate{H} & \multigate{2}{U_C(\theta_1)} & \multigate{2}{U_M(\tau_1)} & \qw & \cdots & \qw & \multigate{2}{U_C(\theta_p)} & \multigate{2}{U_M(\tau_p)} & \meter & \cw\\
\lstick{\vdots}    & \gate{H} & \ghost{U_C(\theta_1)}        & \ghost{U_M(\tau_1)}        & \qw &        & \qw & \ghost{U_C(\theta_p)}        & \ghost{U_M(\tau_p)}        & \meter & \cw\\
\lstick{|0\rangle} & \gate{H} & \ghost{U_C(\theta_1)}        & \ghost{U_M(\tau_1)}        & \qw & \cdots & \qw & \ghost{U_C(\theta_p)}        & \ghost{U_M(\tau_p)}        & \meter & \cw
}
}
\ee

\subsection{Tilted-QAOA}

Unlike the standard QAOA algorithm that minimizes the expectation loss function, here we utilize our quantum tilted loss in Definition~\ref{def:QTL} as the loss function, and our solutions are obtained through minimization of QTL
\begin{align}
\bm\theta^*, \bm\tau^* &\in \argmin_{\bm\theta, \bm\tau \in[0, 2\pi)^{\otimes p}} \mathcal{L}_\gamma(H_C, \Psi_p(\bm\theta, \bm\tau)) \\\nonumber
&= \argmin_{\bm\theta, \bm\tau \in[0, 2\pi)^{\otimes p}} \frac{1}{\gamma}\log \tr \left( e^{\gamma H_C} \Psi_p(\bm\theta, \bm\tau) \right).
\end{align}
And when the parameters are optimized, the resulting state is close to the ground state of the cost Hamiltonian, as discussed in Lemma~\ref{lem:limit} and Lemma~\ref{lem:faithfulness}, and the optimal cut $x^* \in \{0,1\}^{n}$ is then obtained. We name our modification of the QAOA algorithm with tilted loss as \textit{tilted-QAOA}. The main part of the algorithm lies in efficient estimation of the quantum tilted loss. Since for the MaxCut problem, the cost Hamiltonian is diagonal. We can utilize Eq.~\ref{eqTiltDiagonal} in Section ~\ref{subsecFiniteTilt} to first estimate the expectation values and then calculate the classical moment-generating function as the estimation of tilted loss, as described in Algorithm~\ref{algo:estimation_diagonal}.

\begin{algorithm}[H]
\caption{Tilted-QAOA: diagonal QTL estimation}
\label{algo:estimation_diagonal}
\begin{algorithmic}[1]
\REQUIRE cost Hamiltonian $H_C$; circuit depth $p$; parameters $\theta, \tau$; shot budget $N_{shot}$; tilt parameter $\gamma$
\STATE Prepare the depth-$p$ QAOA state $|\psi_p(\theta, \tau)\rangle$
\FOR{$s = 1$ \textbf{to} $N_{shot}$}
    \STATE Measure the state in the computational basis to obtain a random bitstring outcome $z_s \in \{0,1\}^n$
    \STATE Compute the individual classical energy sample $E_s = \langle z_s | H_C | z_s \rangle$
\ENDFOR
\STATE Obtain the set of independent energy samples $\{E_s\}_{s=1}^{N_{shot}}$
\IF{$\gamma = 0$}
    \STATE Estimate the empirical expectation loss:
    \STATE $\hat{L}_0(\theta, \tau) \gets \frac{1}{N_{shot}} \sum_{s=1}^{N_{shot}} E_s$
\ELSE
    \STATE Estimate the empirical tilted loss:
    \STATE $\hat{L}_\gamma(\theta, \tau) \gets \frac{1}{\gamma} \log \left( \frac{1}{N_{shot}} \sum_{s=1}^{N_{shot}} e^{\gamma E_s} \right)$
\ENDIF
\ENSURE $\hat{L}_\gamma(\theta, \tau)$
\end{algorithmic}
\end{algorithm}

The tilted-QAOA algorithm has a hybrid quantum-classical procedure resembling the conventional QAOA algorithm, except for the loss function being modified to our QTL. With respect to the classical optimizations, the gradients of QTL are hard from a statistical perspective, as discussed in Section~\ref{subsecGradientQTL}. Hence, we construct our heuristic optimizer for the tilted-QAOA algorithm. We first utilize the '\texttt{Pennylane}' package \cite{bergholm2022pennylane} to build the QAOA circuit and call '\texttt{qml.grad}' with the finite difference method as the estimator of gradients. Based on those gradients, we implement a gradient-based optimization strategy with Polyak’s momentum~\cite{polyak1964methods}, combined with gradient clipping and a learning rate decay schedule, as explained in Algorithm~\ref{alg:gradient_QTL}.

\begin{remark}
    Although we prove in Appendix~\ref{app:qtl-parameter-shift} that an exact parameter-shift rule exists for QTL gradients, its practical use on hardware requires estimating the ratio of two extremely noisy partition function values ($Z_{\gamma}$), which can amplify errors. For the structured, low-depth QAOA circuits used in this demonstration, we found that standard finite-difference methods provided a more operationally stable gradient approximation under finite-shot constraints.
\end{remark}

\begin{algorithm}[H]
\caption{Tilted-QAOA: heuristic optimization}
\label{alg:gradient_QTL}
\begin{algorithmic}[1]
\REQUIRE cost Hamiltonian $H_C$; circuit depth $p$; initial parameters $\bm{\theta}^{(0)}, \bm{\tau}^{(0)}$; shots per objective evaluation $N_{\mathrm{shot}}$; tilt parameter $\gamma$; total optimization steps $T$; clipping parameter $C$; initial momentum $\bm{\nu}^{(0)}=\bm{0}$; momentum ratio $r_{\text{mom}}$; learning rate $\eta_{\text{lr}}$;  decay offset $\eta_{\text{off}}$; decay power $\eta_{\text{pow}}$; learning penalty $\eta_{\text{lp}}$;
\FOR{$t=0$ \TO $T-1$}
\STATE estimate $\widehat{L}_\gamma(\bm{\theta}^{(t)}, \bm{\tau}^{(t)})$ using Algorithm~\ref{algo:estimation_diagonal};
\STATE obtain $\widehat{g}_{\gamma}(\bm{\theta}^{(t)}, \bm{\tau}^{(t)})\gets\nabla_{\bm{\theta}, \bm{\tau}}\widehat{L}_\gamma(\bm{\theta}^{(t)}, \bm{\tau}^{(t)})$;
\STATE calculate the norm of gradient $\|\widehat{g}(\bm{\theta}^{(t)}, \bm{\tau}^{(t)})\|$;
\IF{$\|\widehat{g}_{\gamma}(\bm{\theta}^{(t)}, \bm{\tau}^{(t)})\|> C$}
\STATE $\widehat{g}_{\gamma}(\bm{\theta}^{(t)}, \bm{\tau}^{(t)}) \gets C\frac{\widehat{g}_{\gamma}(\bm{\theta}^{(t)}, \bm{\tau}^{(t)})}{\|\widehat{g}_{\gamma}(\bm{\theta}^{(t)}, \bm{\tau}^{(t)})\|} $;
\ENDIF
\STATE calculate the momentum;
\STATE $\bm{\nu}^{(t+1)}_{\gamma} \gets r_{\text{mom}} \bm{\nu}^{(t)}_{\gamma} + (1 - r_{\text{mom}}) \widehat{g}_{\gamma}(\bm{\theta}^{(t)}, \bm{\tau}^{(t)})$;
\STATE calculate the learning rate $\eta^{(t)} \gets \frac{\eta_{\text{lr}} }{(t + 1 + \eta_{\text{off}})^{\eta_{\text{pow}}}}$;
\STATE apply the fixed-tilt learning penalty
$\eta^{(t)} \gets \frac{\eta^{(t)}}{1+\eta_{\mathrm{lp}}|\gamma|}$;
\STATE update the parameters;
\STATE $\bm{\theta}^{(t+1)} \gets \bm{\theta}^{(t)} - \eta^{(t)} \cdot \bm{\nu}^{(t+1)}_{\gamma}|_{\bm{\theta}}$;
\STATE $\bm{\tau}^{(t+1)} \gets \bm{\tau}^{(t)} - \eta^{(t)} \cdot \bm{\nu}^{(t+1)}_{\gamma}|_{\bm{\tau}}$;
\ENDFOR
\ENSURE $\bm{\theta}^{(T)}, \bm{\tau}^{(T)}$.
\end{algorithmic}
\end{algorithm}

\subsection{Fixed tilt under finite-shot estimation}
When the tilt parameter $\gamma$ is fixed to be several specific values, we have the schematic framework for the tilted-QAOA algorithm elaborated in Algorithm~\ref{algo:tilted_qaoa}.
We benchmark our tilted-QAOA algorithm across a random set of Erd\H{o}s--R\'enyi ensemble $G(n,p_{\mathrm{edge}})$. We fixed the size $n=8$ and the probability that an edge exists between two nodes by $p_{\mathrm{edge}} = 0.43$. 
We study fixed-tilt performance over a range of tilt parameter $|\gamma| \in [0,4]$, using the mean final cut ratio as the performance metric. 
And we use QAOA depth $p=2$ and optimize each run for $T=100$ iterations. For each setting, results are averaged over five Erdős--Rényi graph instances and five independent random initializations of the QAOA angles. The final mean cut ratio is computed from the finite-shot output distribution produced by the simulator, and is therefore itself a finite-sampling estimate rather than an exact statevector expectation. We minimize the loss function with a negative tilt, corresponding to energy minimization. For convenience, the horizontal axis in the figures reports the tilt $|\gamma|$, corresponding to the tilted objective $L_{-\gamma}$. Figure~\ref{fig:FixedGammaShotsCombine} shows the mean final cut ratio as a function of $|\gamma|$ under different shot budgets. Since $\gamma$ is fixed during each run, the tilt-dependent penalty is constant along the optimization trajectory.

\begin{algorithm}[H]
\caption{Tilted-QAOA: fixed tilt experiment}
\label{algo:tilted_qaoa}
\begin{algorithmic}[1]
\REQUIRE tilt parameters set $\Gamma$; the exact MaxCut value obtained from brute-force search $C_{\max}$; cut function $\text{Cut}(x)$;
\FORALL{$\gamma \in \Gamma$}
\STATE use Algorithm~\ref{alg:gradient_QTL} to obtain $\bm{\theta}^{(T)}, \bm{\tau}^{(T)}$;
\STATE prepare $|\psi_p(\bm{\theta}^{(T)}, \bm{\tau}^{(T)})\rangle$ and measure it in the computational basis to estimate the expected cut from samples with probability vector $P_{\bm{\theta}, \bm{\tau}}(x)$ and string $x\in\{0,1\}^{\otimes n}$;
\STATE calculate the cut value $C^{\text{QAOA}}_{\gamma} \gets \sum_{x}P_{\bm{\theta}, \bm{\tau}}(x)\text{Cut}(x)$;
\STATE calculate mean final ratio $\text{MFR}_{\gamma} \gets \frac{C^{\text{QAOA}}_{\gamma}}{C_{\max}}$.
\ENDFOR
\ENSURE $\{\text{MFR}_{\gamma}\}_{\gamma\in\Gamma}$.
\end{algorithmic}
\end{algorithm}

\begin{figure}[h]
    \centering
    \includegraphics[width=\linewidth]{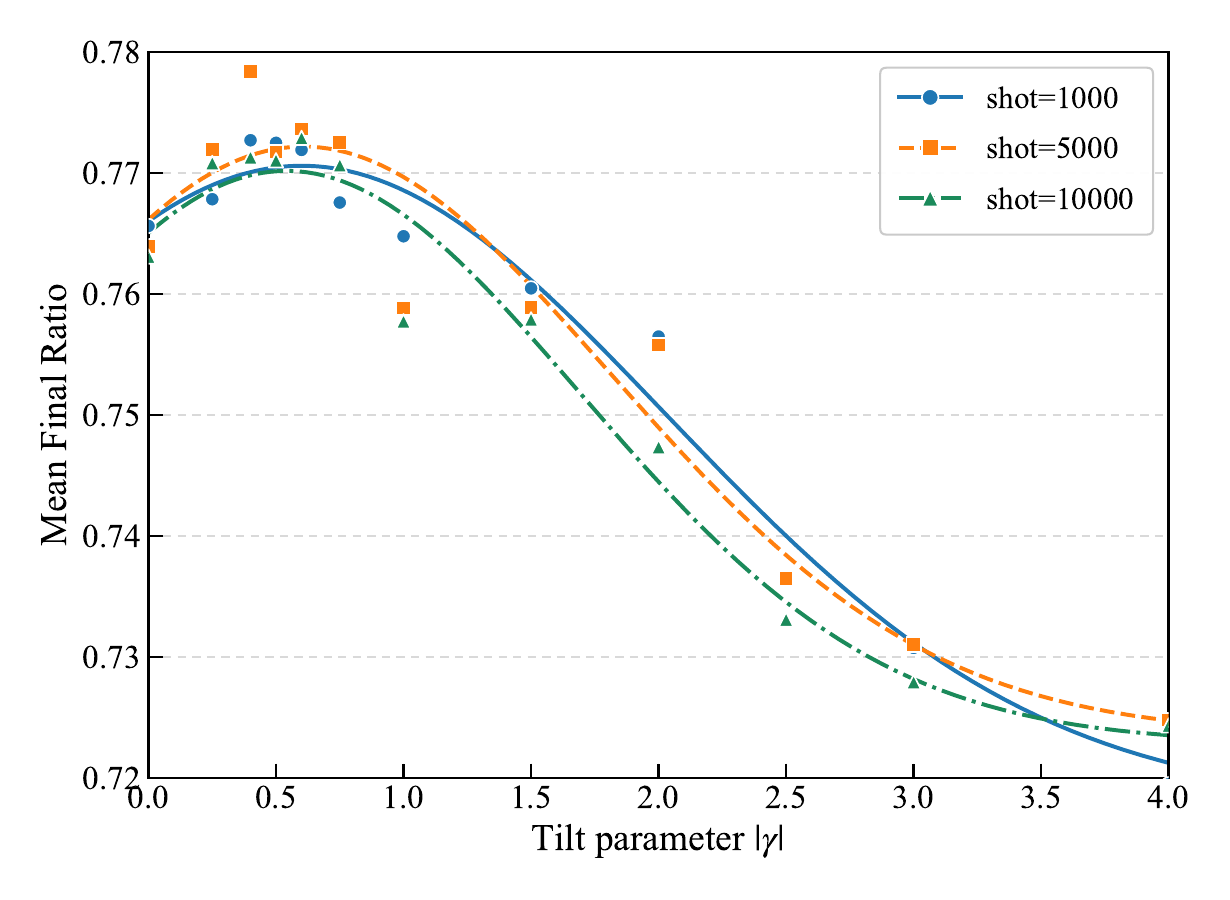}
    \caption{The shot budget affects the estimation of the QTL, hence affects the performance of the tilted-QAOA algorithm. We benchmark the algorithm under different shot budgets $N_{\mathrm{shot}} \in \{1000,\,5000,\,10000\}$. The curves show weighted robust fits over $|\gamma| \in [0,4]$ using shifted exponential-quadratic models of the form $d + (a + b|\gamma|)e^{-c|\gamma|^2}$. We utilize this function to model the relationship between the mean final ratio and the tilt parameter $|\gamma|$, to observe the behavior of the tilted-QAOA algorithm at different shot budgets. }
    \label{fig:FixedGammaShotsCombine}
\end{figure}

The dependence on $|\gamma|$ remains clearly non-monotonic: performance first improves as the tilt increases from zero, reaches its best values at a moderate tilt, and then degrades again for larger $|\gamma|$. In the present experiment, the preferred fixed tilt lies in a relatively small intermediate range, roughly around $|\gamma|\approx 0.4$--$0.7$, rather than at the largest tilt values. The optimal choice is also not entirely universal across shot budgets: the three curves are broadly similar, but their peak heights and local behavior differ slightly, with the intermediate shot budget attaining the best observed peak in this figure. At the same time, increasing the shot budget does not produce a uniformly better fixed-$|\gamma|$ curve across the full range of tilts, indicating that the benefit of stronger tilting is not simply unlocked monotonically by more samples. This behavior is consistent with the trainability--estimability trade-off developed in Sections~\ref{subsecFiniteTilt} and~\ref{sec:gradient_resolvability}: increasing $|\gamma|$ can reshape the optimization landscape in a potentially favorable way, but it also makes the tilted objective statistically harder to estimate from finite samples. In this regime, the main empirical message is therefore the existence of a shot-dependent moderate optimal tilt, rather than a monotonic advantage of either larger tilts or larger shot budgets. Because the current results are averaged over a limited number of instances, the precise location and height of the maxima should not be over-interpreted.

\subsection{Ascending tilt improves low-shot performance}
\label{subsec:scheduled_qtl}

The observation of tilt performance under finite-shot estimation motivates an ascending tilt schedule, as elaborated in Algorithm~\ref{algo:ascending_tilted_qaoa}. Rather than fixing a single tilt throughout training, we begin in the weak-tilt regime and gradually increase the tilt toward the stronger tail-focusing regime. Intuitively, this allows the optimizer to first benefit from the smoother, more expectation-like landscape at early iterations, and then progressively sharpen the objective once the parameters have entered a promising region. In this sense, the schedule acts as a continuation strategy: it starts from an easier objective and slowly transforms it into a more selective one. The ascending schedule is designed to combine both advantages: early-stage stability from the weak-tilt regime and late-stage sharpening from the stronger-tilt regime.

The scheduled variant follows the continuation idea of Ascending-CVaR, which progressively sharpens the tail objective and warm-starts later subproblems from earlier ones \cite{Kolotouros2022}. Ascending QTL implements the same principle: one first optimizes a smoother subproblem and then warm-starts the next, sharper subproblem using the previously obtained parameters. In our setting, the sharpening is implemented by the linear tilt schedule in the following. The tilt evolves during training from a smooth objective to a sharper one,
\begin{align}
\label{eq:ascending}
|\gamma_t| = |\gamma_{\text{init}}| + \frac{t}{T-1}\bigl(|\gamma_{\text{end}}|-|\gamma_{\text{init}}|\bigr),&
 \\\nonumber
t=0,1,\dots,T-1,&
\label{eq:linear_tilt_schedule}
\end{align}
with $\gamma_{\text{init}}=0$ and $\gamma_{\text{end}}<0$, where $T$ is the total number of optimization steps. For the fixed-versus-scheduled comparison in Figure~\ref{fig:fixed-ascending}, we made additional modifications: the smaller penalty and slower decay keep updates larger when the scheduled tilt becomes stronger late in training. The training starts from the expectation-value objective and gradually moves toward a more selective tilted objective. Figure~\ref{fig:fixed-ascending} compares fixed tilt with a linear schedule from 0 to $\gamma_{\text{end}}$, using the same optimization steps and total shot budget. For each $|\gamma|$, we first average performance over the fixed set of random initializations for each problem instance, and then report the mean ±1 standard error of the mean (SEM) across problem instances. In the $5000$-shot regime, the scheduled objective generally attains higher final cut ratios than fixed-tilt training at moderate-to-large final tilt magnitudes. This suggests that starting from a smoother objective and gradually increasing selectivity can improve optimization stability when loss estimates are noisy.

\begin{algorithm}[H]
\caption{Tilted-QAOA: ascending tilt experiment}
\label{algo:ascending_tilted_qaoa}
\begin{algorithmic}[1]
\REQUIRE cost Hamiltonian $H_C$; circuit depth $p$; initial parameters $\bm{\theta}^{(0)}, \bm{\tau}^{(0)}$; shots per objective evaluation $N_{\mathrm{shot}}$; tilt parameters set $\Gamma$; total optimization steps $T$; clipping parameter $C$; initial momentum $\bm{\nu}^{(0)}=\bm{0}$; momentum ratio $r_{\text{mom}}$; learning rate $\eta_{\text{lr}}$; learning penalty $\eta_{\text{lp}}$; decay offset $\eta_{\text{off}}$; decay power $\eta_{\text{pow}}$; the exact MaxCut value obtained from brute-force search $C_{\max}$; cut function $\text{Cut}(x)$;
\FORALL{$|\gamma_{\text{end}}| \in \Gamma$}
\FOR{$t=0$ \TO $T-1$}
\STATE calculate ascending tilt parameter;
$\gamma^{(t)}\gets \frac{t\gamma_{\text{end}}}{T-1}$.
\STATE estimate $\widehat{L}_{\gamma^{(t)}}(\bm{\theta}^{(t)}, \bm{\tau}^{(t)})$ using Algorithm~\ref{algo:estimation_diagonal};
\STATE obtain $\widehat{g}_{\gamma^{(t)}}(\bm{\theta}^{(t)}, \bm{\tau}^{(t)})\gets \nabla_{\bm{\theta}, \bm{\tau}}\widehat{L}_{\gamma^{(t)}}(\bm{\theta}^{(t)}, \bm{\tau}^{(t)})$;
\STATE calculate the norm of gradient $\|\widehat{g}_{\gamma^{(t)}}(\bm{\theta}^{(t)}, \bm{\tau}^{(t)})\|$;
\IF{$\|\widehat{g}_{\gamma^{(t)}}(\bm{\theta}^{(t)}, \bm{\tau}^{(t)})\|> C$}
\STATE $\widehat{g}_{\gamma^{(t)}}(\bm{\theta}^{(t)}, \bm{\tau}^{(t)}) \gets C\frac{\widehat{g}_{\gamma^{(t)}}(\bm{\theta}^{(t)}, \bm{\tau}^{(t)})}{\|\widehat{g}_{\gamma^{(t)}}(\bm{\theta}^{(t)}, \bm{\tau}^{(t)})\|} $;
\ENDIF
\STATE calculate the momentum;
\STATE $\bm{\nu}_{\gamma^{(t)}}^{(t+1)} \gets r_{\text{mom}} \bm{\nu}_{\gamma^{(t)}}^{(t)} + (1 - r_{\text{mom}}) \widehat{g}_{\gamma^{(t)}}(\bm{\theta}^{(t)}, \bm{\tau}^{(t)})$;
\STATE calculate the learning rate $\eta^{(t)} \gets \frac{\eta_{\text{lr}} }{(t + 1 + \eta_{\text{off}})^{\eta_{\text{pow}}}}$;
\STATE add penalty to learning rate for ascending $\gamma^{(t)}$;
\STATE $\eta^{(t)}_{\gamma^{(t)}} \gets  \frac{\eta^{(t)}}{(1 + \eta_{\text{lp}} \cdot |\gamma^{(t)}|)}$;
\STATE update the parameters;
\STATE $\bm{\theta}^{(t+1)} \gets \bm{\theta}^{(t)} - \eta^{(t)}_{\gamma^{(t)}} \cdot \bm{\nu}_{\gamma^{(t)}}^{(t+1)}|_{\bm{\theta}}$;
\STATE $\bm{\tau}^{(t+1)} \gets \bm{\tau}^{(t)} - \eta^{(t)}_{\gamma^{(t)}} \cdot \bm{\nu}_{\gamma^{(t)}}^{(t+1)}|_{\bm{\tau}}$;
\ENDFOR
\STATE prepare $|\psi_p(\bm{\theta}^{(T)}, \bm{\tau}^{(T)})\rangle$ and measure it in the computational basis to estimate the expected cut from samples with $P_{\bm{\theta}, \bm{\tau}}(x)$ and string $x\in\{0,1\}^{\otimes n}$;
\STATE calculate the cut value $C^{\text{QAOA}}_{\gamma_{\text{end}}} \gets \sum_{x}P_{\bm{\theta}, \bm{\tau}}(x)\text{Cut}(x)$;
\STATE calculate mean final ratio $\text{MFR}_{\gamma_{\text{end}}} \gets \frac{C^{\text{QAOA}}_{\gamma_{\text{end}}}}{C_{\max}}$.
\ENDFOR
\ENSURE $\{\text{MFR}_{\gamma_{\text{end}}}\}_{\gamma_{\text{end}}\in\Gamma}$.
\end{algorithmic}
\end{algorithm}

\begin{figure}[h]
  \centering

  \includegraphics[width=\linewidth]{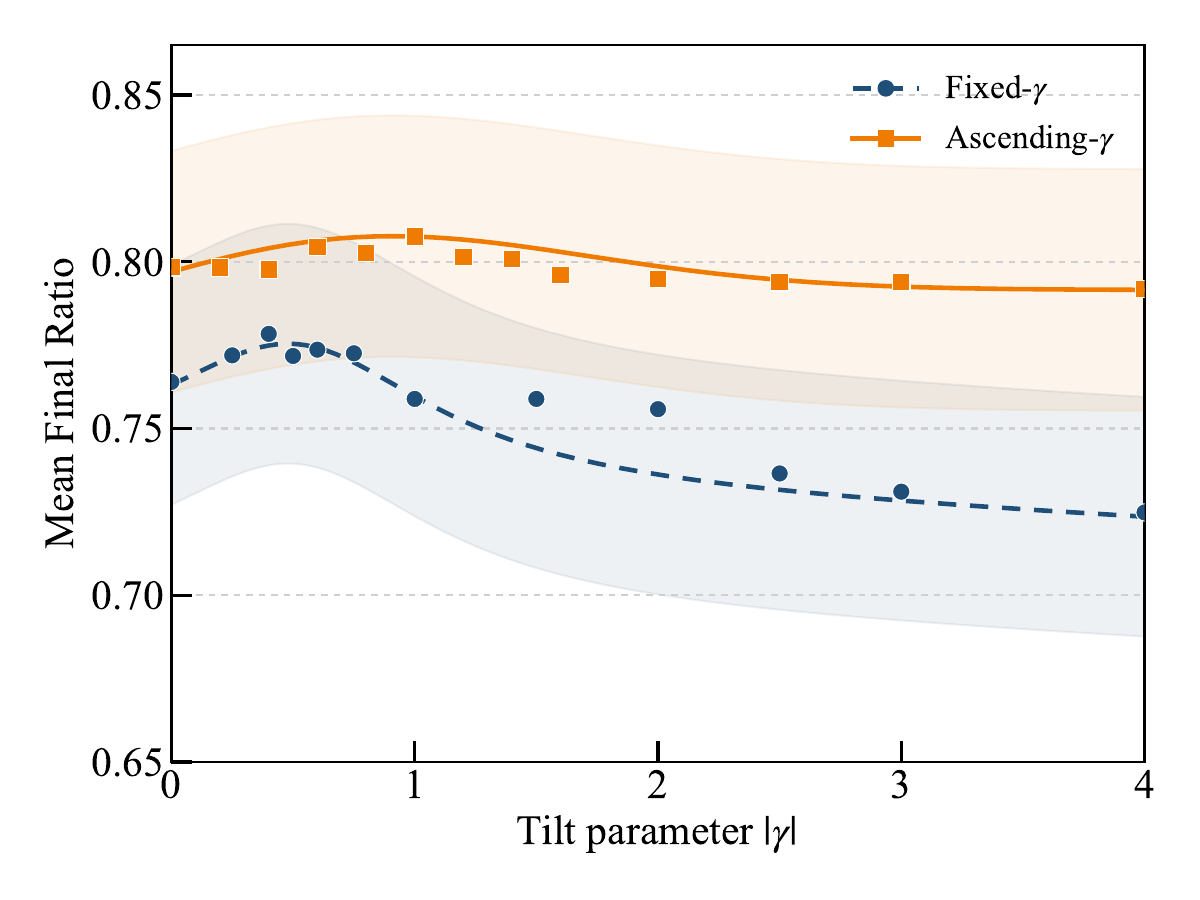}
    \caption{Mean final cut ratio versus final tilt magnitude $\gamma_{\text{end}}$ in $N_{\mathrm{shot}}=5000$. Fixed uses $\gamma = \gamma_{\text{end}}$ for all steps; ascending uses linear ascending $\gamma$ as in \cref{eq:ascending}. Both methods use the same optimizer, number of optimization steps, total shot budget, and the same shifted exponential quadratic model for fitting. Error bars denote $\pm 1$ SEM across problem instances.}
    \label{fig:fixed-ascending}
\end{figure}

\section{Conclusion and Discussions}

We introduced the Quantum Tilted Loss (QTL) as a tunable nonlinear objective for variational quantum optimization, analyzing it from three  perspectives: objective design, optimization geometry, and statistical estimability. The central message is that tilting is useful because it fundamentally changes the path to the target solution. We established that QTL preserves the global minima of the original problem while dynamically reshaping the local curvature of the optimization landscape. Furthermore, while previous tunable loss functions, such as the popular CVaR~\cite{Barkoutsos2020} or Gibbs objectives~\cite{gibbs_objective} have largely relied on heuristics and specific numerical demonstrations to outperform standard mean estimation, our framework provides a rigorous theoretical unification. We clarified the exact mathematical relationships between standard expectation-value training, tail-focused objectives like CVaR, and Gibbs formulations.

 Our analysis demonstrates that nonlinear objective design is not a free lunch. By tuning the tilt parameter, we can dynamically trade geometric flatness for statistical variance, shifting the primary algorithmic bottleneck from vanishing gradients to sample complexity. While stronger tilts successfully amplify the gradient signal $\|\nabla \mathcal{L}_\gamma \|$ needed to escape flat regions and guide descent methods, they simultaneously increase the statistical cost of estimating both the loss function and its gradient from finite measurement samples. This exposes a fundamental trainability-estimability trade-off: an exponential explosion in sampling variance is the price to pay when deploying aggressive nonlinear transformations to bypass barren plateaus in the near-term hardware regime.

The practical message is therefore not that one should universally tilt as aggressively as possible. Instead, exponential tilting must be treated as a controllable resource whose benefits depend entirely on the available measurement budget, the underlying observable structure, and the current stage of optimization. Our QAOA experiments on the MaxCut problem are consistent with these claims, demonstrating that when the shot budget is strictly limited, employing an ascending tilt schedule can outperform than relying on a fixed tilt. This approach allows the optimizer to balance reliable gradient estimation during the initial iterations with a sharper focus on tail outcomes in the final stages of training.

More broadly, these findings suggest that loss design is a tool just as important as ansätz design or optimizer selection for building variational quantum algorithms. Recent literature has made it clear that trainability, sampling costs, and potential quantum advantage cannot be treated as separate problems; they must be balanced together to find a practical advantage. Within this context, the QTL framework provides a mathematically rigorous and analytically tractable testbed for understanding exactly how nonlinear objectives reshape variational landscapes, and critically, where the real operational bottlenecks move when they do.

We expect that the following directions of future work will become relevant in the context of quantum optimization.

\begin{enumerate}
    \item \textbf{Ascending tilt schedules:} While our numerical analysis utilized fixed and linear ascending tilts, an important open problem is how to design schedules that adapt to circuit depth, system size, and shot budget in a principled way. Because the optimal tilt is inherently state-dependent, a natural extension is to dynamically adjust the parameter $\gamma_t$ using online empirical estimates of the noise-to-signal ratio. By autonomously maintaining statistical resolvability at each optimization step, this approach would explicitly bridge the QTL framework with advanced stochastic optimization, ensuring efficient landscape navigation under strict finite-shot constraints.
    \item \textbf{Beyond Ground-State Energy Minimization:} Because exponential tilting explicitly targets and isolates specific spectral regions, the QTL framework naturally extends to tasks such as quantum tomography, thermal-state learning, and Hamiltonian learning. A particularly concrete direction lies in Gibbs-state preparation: integrating the QTL with free-energy objectives within dynamic parameterized quantum circuits~\cite{deshpande2024dynamic} could provide a rigorous, analytically tractable mechanism for adaptive thermal state generation.
    \item \textbf{Estimating  Non-Commuting Hamiltonians:} Operationally, the QTL framework is currently most tractable for diagonal Hamiltonians, perturbative tilt regimes, or highly degenerate observables. An open direction is the development of efficient estimation protocols for general, non-commuting Hamiltonians. Promising approaches include  classical shadow tomography for near-term surrogates \cite{huang2020predicting,Stefan2022}, or  quantum signal processing (QSP) and block-encoding techniques in early fault-tolerant architectures. 
\end{enumerate}

\section*{Acknowledgements} YQ and JL thank Adrián Pérez-Salinas for discussions on barren plateaus during his visit at CQT. YQ thanks Richard Küng, Diego García-Martín and Zoë Holmes for discussions and feedback. Parts of the manuscript were revised with the help of LLMs. This work is supported by the National Research Foundation, Singapore, and A*STAR under its CQT Bridging Grant and its Quantum Engineering Programme under grant NRF2021-QEP2-02-P05. JL is supported by the National Research Foundation through the NRF Investigatorship on Quantum-Enhanced Agents (Grant No. NRF-NRFI09-0010),   the RIE25 Japan-Singapore Joint Call on Quantum R25J4IR111, the Singapore Ministry of Education Tier 1 Grant RT4/23 and RG91/25 (S) and the National Quantum Office, hosted in A*STAR, under its Centre for Quantum Technologies Funding Initiative (S24Q2d0009).

\bibliography{bibliography}

\appendix
\setcounter{subsection}{0}
\setcounter{table}{0}
\setcounter{figure}{0}

\vspace{2cm}
\onecolumngrid
\vspace{2cm}

\setcounter{equation}{0}
\setcounter{table}{0}
\setcounter{section}{0}
\setcounter{definition}{0}
\setcounter{figure}{0}

\newpage

\section*{Supplemental material}

\section{Classical Variational Representations and Entropic Risk}\label{app:classical_risk}

\subsection{Statistics and machine learning}\label{subsec:classml}
Li et al. \cite{li2020tilted, li2023tilted} replaced the usual empirical risk minimization (ERM) with a $\gamma$-tilted empirical risk, an exponential-tilting of the per-sample/objective values that smoothly interpolates between ERM and worst-case (or best-case) behavior and is known to relate to VaR/CVaR-type risks.

\begin{definition}[Tilted Empirical Risk Minimization (TERM) \citep{li2020tilted,li2023tilted}]\label{def:term}
For $\gamma \in \mathbbm{R}$, the $\gamma$-tilted loss in ERM is defined as the
tilted empirical risk minimization, given by
\be
\widetilde{R}_{\gamma}(L):=
  \frac{1}{\gamma}\log \left(\frac{1}{n} \sum_{i \in [n]}
  e^{\gamma L(h({\bf x}^{(i)}), y_i)} \right),
\ee
where $L(h({\bf x}^{(i)}), y_i)$ is the loss function on hypothesis
$h({\bf x}^{(i)})$ and true label $y_i$, and $n$ is the number of training
samples.
\end{definition}

\begin{definition}
[LogSumExp (LSE)]\label{def:lse}
LSE is a smooth function that maps a vector of real numbers $x_1,\dots,x_n$ to
\be
\operatorname{LSE}(x_1,\dots,x_n) = \log\Big(\sum_{i=1}^n e^{x_i}\Big).
\ee
\end{definition}
LSE is often used as a numerically stable way to work with sums of exponentials in log space, and as a smooth approximation of the maximum:
\be
\max_i x_i \le \operatorname{LSE}(x_1,\dots,x_n) \le \max_i x_i + \log n.
\ee
The LogSumExp function arises as the convex conjugate of the negative Shannon entropy, and possesses the variational representation (see Sec.~3.6 in Ref.~\cite{Wainwright2007}). Let $\Delta_n := \{p\in \mathbb{R}^+ : \, \sum_{i=1}^n p_i =1\}$, then
\be\label{Eq:lse_var}
\operatorname{LSE}(x) = \sup_{p\in\Delta_n}  \left\{\sum_{i=1}^n p_i x_i + H(p)\right\},
\ee
where $H(p) = -\sum_{i=1}^n p_i \log p_i$ is the Shannon entropy of $p$. The unique maximizer is the softmax distribution
\be
p_i^* (x) = \frac{e^{x_i}}{\sum_{j=1}^n e^{x_j}}.
\ee
Write the entropy in terms of the KL divergence to the uniform distribution $u = (1/n, \dots, 1/n)$ via
$H(p) = \log n - D_{\operatorname{KL}}(p \| u)$, where $D_{\operatorname{KL}}$ denotes KL divergence between two distributions. The representation in \cref{Eq:lse_var} can be rewritten as
\be
\operatorname{LSE}(x) = \log n + \sup_{p\in\Delta_n}  \left\{\sum_{i=1}^n p_i x_i - D_{\operatorname{KL}}(p \| u)\right\},
\ee
Now consider TERM for $\gamma>0$, the variational representation is given by
\be\label{eq:term_var}
\widetilde{R}_{\gamma}(x) = \sup_{p\in\Delta_n}  \left\{\sum_{i=1}^n p_i x_i - \frac{1}{\gamma} D_{\operatorname{KL}}(p \| u)\right\}.
\ee
That is, $\widetilde{R}_\gamma$ can be viewed as a KL-regularized worst-case reweighting of the empirical distribution. For $\gamma < 0$, an analogous expression holds with the supremum replaced by an infimum. The corresponding representation in terms of empirical risk is proved in Lemma~7 of Ref.~\cite{li2023tilted}.

The finite-sample TERM is the empirical analogue of a classical \emph{entropic risk measure}. In classical probability and mathematical finance, entropic risk is a canonical convex risk measure defined as below.
\begin{definition}[Entropic Risk Measure \cite{AhmadiJavid2011, FLLMER2011, Fllmer2016}]\label{def:erm}
Let $(\Omega,\mathcal{F},P)$ be a probability space and
$x:\Omega\to\mathbb{R}$ a loss random variable. For $\gamma\in\mathbb{R}$, the entropic risk is defined as
\be
\rho_\gamma(x):=
\frac{1}{\gamma}\log \mathbb{E}_P\big[e^{\gamma x}\big],
\ee
whenever the logarithm is finite.
\end{definition}
When $P$ is the empirical distribution over the training samples, the quantity $\rho_\gamma(x)$ reduces exactly to the TERM $\widetilde{R}_\gamma(x)$. For small $|\gamma|$, $\rho_\gamma(x)$ is close to the usual expectation $\mathbb{E}_P [x]$, while $\gamma>0$ puts exponentially increasing weight on large losses, making $\rho_\gamma(x)$ sensitive to tail events. A key feature is its variational representation
\be\label{eq:erm_var}
\rho_\gamma(x) =\sup_{Q \ll P}
\left\{ \mathbb{E}_Q[x] - \frac{1}{\gamma} D_{\operatorname{KL}}(Q \| P)\right\},
\ee
for $\gamma>0$, where the supremum is over all probability measures $Q$ absolutely continuous with respect to $P$. The KL-regularized variational representation in \cref{eq:erm_var} places QTL in the broader distributionally robust optimization paradigm, where one optimizes against distributions in an ambiguity set around a reference law \cite{Kuhn2025}. 
In words, $\rho_\gamma(x)$ is the worst-case expected loss over all distributions $Q$ in a KL-regularized neighbourhood of the reference distribution $P$. The maximizer $Q_\gamma$ is given by the exponential tilting (or called the Esscher transform \cite{esscher1932probability})
\be\label{eq:exptilt}
\frac{dQ_\gamma^i}{dP} = \frac{e^{\gamma x_i}}{\mathbb{E}_P[e^{\gamma x}]},
\ee
which is the continuous analogue of the softmax weights in the finite-sample representation of LogSumExp. Specializing \cref{eq:erm_var}  to the empirical distribution $P = u$ recovers the variational form of LogSumExp and TERM in \cref{Eq:lse_var} and \cref{eq:term_var}.

\subsection{Economic expected utility theory}\label{sec:risk_section}
Classical economic treatments of decision-making under uncertainty provide a standard framework for evaluating random outcomes and formalizing risk preferences \cite{mas1995microeconomic,Gollier2001,Fllmer2016}. 
Three core ingredients are emphasized: (a) expected utility as a criterion for ranking uncertain payoffs and its equivalent representation via certainty equivalents; (b) the role of utility curvature in determining risk attitudes, summarized locally by the Arrow--Pratt measure of absolute risk aversion; and (c) the special case of constant absolute risk aversion (CARA), for which exponential utility yields a closed-form certainty equivalent and naturally leads to the entropic (exponential) risk measure. 
We briefly review these notions below for completeness.

\textit{Utility, certainty equivalents, and risk.}
In classical decision theory, a random payoff $X$ is evaluated through a utility function
$u:\mathbb{R}\to\mathbb{R}$ by means of its expected utility.
When $u$ is strictly increasing, this evaluation can be expressed in terms of a
\emph{certainty equivalent}, defined as the deterministic value $\mathrm{CE}(X)$
that renders the agent indifferent between $X$ and a sure outcome.
\begin{align}
    u(\mathrm{CE}(X)) = \mathbb{E}[u(X)] ,
\end{align}
or equivalently,
\begin{align}
    \mathrm{CE}(X)
    = u^{-1}\!\left(\mathbb{E}[u(X)]\right).
\end{align}
The certainty equivalent provides a scalar summary of the full distribution of $X$
and coincides with the expectation only in the risk-neutral case.

\textit{Risk attitudes and absolute risk aversion.}
The curvature of $u$ determines the agent's attitude toward risk:
affine utilities correspond to risk neutrality, concave utilities to risk aversion,
and convex utilities to risk seeking.
A local measure of risk sensitivity is given by the absolute risk aversion
\begin{align}
    A(x) := -\frac{u''(x)}{u'(x)} .
\end{align}
Utilities with constant absolute risk aversion (CARA) satisfy $A(x)\equiv\kappa>0$,
which uniquely (up to affine transformations) leads to exponential utility functions.

\textit{Exponential utility and entropic risk.}
For CARA utilities, a convenient representative is the exponential utility
\begin{align}
    u(x) \propto -e^{-\kappa x},
    \qquad \kappa>0 .
\end{align}
In this case, the certainty equivalent of a random variable $X$ admits a closed form,
\begin{align}
    \mathrm{CE}(X)
    = -\frac{1}{\kappa}\log \mathbb{E}\!\left[e^{-\kappa X}\right].
\end{align}
When $X$ is interpreted as a loss, this naturally induces the entropic risk measure
\begin{align}
    \rho_\kappa(X)
    := \frac{1}{\kappa}\log \mathbb{E}\!\left[e^{\kappa X}\right],
\end{align}
which reduces to the expectation in the limit $\kappa\to 0$ and increasingly
emphasises tail events as $|\kappa|$ grows.

\section{Proof of properties of QTL}\label{app:propery-QTL}

Lemma~\ref{lem:prop} [Basic Properties of QTL]
\textit{
The QTL satisfies the following:
\begin{enumerate}
    \item Non-negativity: $\mathcal{L}_\gamma(O,\rho) \ge 0$ if the eigenvalues of $O$ are non-negative, i.e., $O\succeq 0$.
    \item Additivity: $ \mathcal{L}_\gamma(O,\rho)=\mathcal{L}_\gamma(O_A,\rho_A)+\mathcal{L}_\gamma(O_B,\rho_B)$ if $\rho=\rho_A\otimes\rho_B$ and $O=O_A\otimes\mathbb{I}_B+\mathbb{I}_A\otimes O_B$.
    \item Shift: $\mathcal{L}_\gamma(O+c\mathbb{I},\rho)=\mathcal{L}_\gamma(O,\rho)+c$ for any constant $c\in\mathbb R$.
    \item Monotonicity in $\gamma$: $\mathcal{L}_\gamma(O,\rho)$ is non-decreasing on $\mathbb{R}$.
\end{enumerate}
}
\begin{proof}
Denote $\mathcal{L}_\gamma(O,\rho) =  \frac{1}{\gamma}\log \tr \left( e^{\gamma O} \rho\right) $. We have:

\begin{enumerate}
\item For $\gamma>0$, if $O\succeq 0$ then $e^{\gamma O}\succeq \mathbb{I}$; and for $\gamma<0$, $e^{\gamma O}\preceq \mathbb{I}$. Hence, $\tr(\rho e^{\gamma O})\ge \tr(\rho)=1$ when $\gamma>0$, yielding $\mathcal{L}_\gamma(O,\rho)\ge 0$; Similarly, $\mathcal{L}_\gamma(O,\rho)\ge 0$ when $\gamma<0$.
For $\gamma=0$, $\mathcal{L}_\gamma(O,\rho)=\tr(\rho O)\ge 0$.
\item The summands commute: $[O_A\otimes\mathbb{I}_B,\mathbb{I}_A\otimes O_B]=0$, hence $e^{\gamma(O_A\otimes\mathbb{I}_B+\mathbb{I}_A\otimes O_B)}=e^{\gamma O_A}\otimes e^{\gamma O_B}$. Then
\begin{align}
\mathcal{L}_\gamma(O,\rho)
&=\frac{1}{\gamma}\log\tr\big((\rho_A\otimes\rho_B)(e^{\gamma O_A}\otimes e^{\gamma O_B})\big)\\\nonumber
&=\frac{1}{\gamma}\log\Big(\tr(\rho_A e^{\gamma O_A})\tr(\rho_B e^{\gamma O_B})\Big)\\\nonumber
&= \mathcal{L}_\gamma(O_A,\rho_A)+\mathcal{L}_\gamma(O_B,\rho_B).
\end{align}
\item Let $c\in\mathbb R$. For $\gamma\neq 0$, since $[O,\mathbb I]=0$ we have
$e^{\gamma(O+c\mathbb I)}=e^{\gamma c}\,e^{\gamma O}$. Hence
\begin{align}
\mathcal{L}_\gamma(O+c\mathbb I,\rho)
&=\frac{1}{\gamma}\log\tr\!\big(\rho\,e^{\gamma(O+c\mathbb I)}\big)\\\nonumber
&=\frac{1}{\gamma}\log\Big(e^{\gamma c}\tr(\rho e^{\gamma O})\Big)\\\nonumber
&=\frac{1}{\gamma}\big(\gamma c+\log\tr(\rho e^{\gamma O})\big)\\\nonumber
&=c+\mathcal{L}_\gamma(O,\rho).
\end{align}
For $\gamma=0$, using $\mathcal{L}_0(O,\rho)=\tr(\rho O)$ and $\tr(\rho)=1$, $\mathcal{L}_0(O+c\mathbb I,\rho)=\tr(\rho O)+c\tr(\rho)=\mathcal{L}_0(O,\rho)+c.$

\item Let the spectral decomposition be $O=\sum_i o_i \Pi_i$, where $\{\Pi_i\}$ are orthogonal projectors.
Define $p_i := \tr(\rho \Pi_i)$ with $p_i\ge 0, \,\sum_i p_i = 1$. 
Then for any $\gamma\in\mathbb{R}$, $Z(\gamma):=\tr(\rho e^{\gamma O})
=\sum_i \tr(\rho \Pi_i) e^{\gamma o_i}
=\sum_i p_i e^{\gamma o_i}$. Now define the $\gamma$-tilted distribution
\begin{align}
q_i^{(\gamma)} := \frac{p_i e^{\gamma o_i}}{Z(\gamma)}.
\end{align}
A direct differentiation gives
\begin{align}
\frac{\mathrm d}{\mathrm d\gamma}\log Z(\gamma)
=\frac{\sum_i p_i o_i e^{\gamma o_i}}{\sum_i p_i e^{\gamma o_i}}
=\sum_i q_i^{(\gamma)} o_i.
\end{align}
Moreover, since
\begin{align}
\log\frac{q_i^{(\gamma)}}{p_i}=\gamma o_i-\log Z(\gamma),
\end{align}
the KL divergence satisfies
\begin{align}
D\left(q^{(\gamma)}\|p\right)
&:=\sum_i q_i^{(\gamma)}\log\frac{q_i^{(\gamma)}}{p_i}\\\nonumber
&=\sum_i q_i^{(\gamma)}\big(\gamma o_i-\log Z(\gamma)\big)\\\nonumber
&=\gamma \sum_i q_i^{(\gamma)} o_i - \log Z(\gamma).
\end{align}
Therefore, 
\begin{align}
\frac{\mathrm d}{\mathrm d\gamma}\mathcal{L}_\gamma(O,\rho)
&=\frac{\mathrm d}{\mathrm d\gamma}\left(\frac{1}{\gamma}\log Z(\gamma)\right)\\\nonumber
&=\frac{\gamma\,\frac{\mathrm d}{\mathrm d\gamma}\log Z(\gamma)-\log Z(\gamma)}{\gamma^2}\\\nonumber
&=\frac{D\!\left(q^{(\gamma)}\|p\right)}{\gamma^2}\ \ge\ 0,
\end{align}
since $D(\cdot\|\cdot)\ge 0$ and $\gamma^2>0$. Hence $\mathcal{L}_\gamma(O,\rho)$ is non-decreasing on all $\mathbb{R}$.
\end{enumerate}
\end{proof}

Lemma~\ref{Lemma:convergence} [Convergence of QTL]
\textit{
Assume the matrix elements of the operator $O(x)$ are continuously differentiable with respect to the parameters $x \in \mathcal{X} \subseteq \mathbb{R}^d$. Under this condition, the Quantum Tilted Loss converges to the standard expectation value in the zero-tilt limit as
\begin{align}
    \lim _{\gamma \rightarrow 0} \mathcal{L}_\gamma(O,\rho) := \mathcal{L}_0(O,\rho)= \tr \left( O \rho \right).
\end{align}
}

\begin{proof}
By the continuous differentiability assumption, the matrix exponential $e^{\gamma O}$ is analytic. Applying L’Hôpital’s rule with respect to the tilt parameter $\gamma$ yields
\begin{align}
\lim _{\gamma \rightarrow 0} \mathcal{L}_\gamma(O,\rho) 
    &= \lim _{\gamma \rightarrow 0} \frac{1}{\gamma}\log \tr \left( e^{\gamma O} \rho\right) \\\nonumber
    &= \lim _{\gamma \rightarrow 0} \frac{ \tr \left( O e^{\gamma O} \rho\right)}{\tr \left( e^{\gamma O} \rho\right)} \\\nonumber
    &= \tr \left( O \rho \right).
\end{align}
\end{proof}

Lemma~\ref{lem:limit} [$\gamma$-Limits of QTL]
\textit{
Let $O = \sum_i \lambda_i \Pi_i$ be the spectral decomposition of an observable $O$, where $\Pi_i$ are the corresponding eigenspace projectors. Given a quantum state $\rho$, let $p_i = \mathrm{Tr}(\rho \Pi_i)$. Define $m_\rho$ and $M_\rho$ as the minimum and maximum eigenvalues of $O$ that have a strictly positive  weight on $\rho$ as $m_\rho = \min \{ \lambda_i \mid p_i > 0 \}, M_\rho = \max \{ \lambda_i \mid p_i > 0 \}$. The QTL exhibits the following limiting behavior: 
\begin{align}
    \lim_{\gamma \to -\infty} \mathcal{L}_\gamma(O,\rho) &= m_\rho \quad (\text{minimum})\\\nonumber
    \lim_{\gamma \to 0} \mathcal{L}_\gamma(O,\rho) &= \tr(O\rho) \quad (\text{expectation value}) \\\nonumber
    \lim_{\gamma \to +\infty} \mathcal{L}_\gamma(O,\rho) &= M_\rho \quad (\text{maximum}) 
\end{align}
}
\begin{proof}
The limit as $\gamma \rightarrow 0$ was established in Lemma~\ref{Lemma:convergence}. Let $\{(\lambda_i,P_i)\}_i$ be the spectral decomposition of $O$, where $\{P_i\}_i$ denotes the complete spectral PVM of $O$ (so $\sum_i P_i=\mathbb{I}$). For a normalized state $\rho$, $p_i=\tr(P_i\rho)$ is a probability distribution supported on the eigenspaces that intersect $\mathrm{supp}(\rho)$. Thus, 
$0 \leq p_i \leq 1$ with $\sum_i p_i=1$.
Let
\be
Z(\gamma) = \tr\big(e^{\gamma O}\rho\big)
= \sum_i e^{\gamma \lambda_i} p_i.
\ee
Since $Z(\gamma)$ is a convex combination of the numbers $\{e^{\gamma \lambda_i}\}$ over that support, we thus have
\be
 Z(\gamma) \in [\min\{e^{\gamma m_\rho}, e^{\gamma M_\rho}\}, \max\{e^{\gamma m_\rho}, e^{\gamma M_\rho}\} ]  \qquad\text{for all }\gamma\in\mathbb{R}.
\ee
After taking logarithm, the limits as $\gamma\to -\infty$ yields
\begin{align}
\lim_{\gamma\to -\infty}\frac{1}{\gamma}\log Z(\gamma) &= \lim_{\gamma\to -\infty}\frac{1}{\gamma}\log \big(\sum_i e^{\gamma \lambda_i} p_i\big)\\\nonumber
&\leq \lim_{\gamma\to -\infty}\frac{1}{\gamma}\log \big(\sum_i e^{\gamma m_\rho} p_i\big)\\\nonumber
&= m_\rho.
\end{align}
On the other hand,
\begin{align}
\lim_{\gamma\to -\infty}\frac{1}{\gamma}\log Z(\gamma) &= \lim_{\gamma\to -\infty}\frac{1}{\gamma}\log \big(\sum_i e^{\gamma \lambda_i} p_i\big)\\\nonumber
&\geq \lim_{\gamma\to -\infty}\frac{1}{\gamma}\log \big(e^{\gamma m_\rho} p_{m_\rho}\big)\\\nonumber
&= m_\rho + \lim_{\gamma\to -\infty}\frac{1}{\gamma}\log \big( p_{m_\rho}\big)\\\nonumber
&= m_\rho.
\end{align}
For the limit as $\gamma \rightarrow +\infty$: This argument follows the same reasoning process as in the case where $\gamma \to -\infty$ above. Following the same steps, it leads to $\lim_{\gamma \to +\infty} \mathcal{L}_\gamma(O,\rho) = M_\rho$.
\end{proof}

Lemma~\ref{lem:faithfulness} [Faithfulness of QTL]
\textit{Let $O$ be an observable with minimum eigenvalue $o_{\min}$ and corresponding eigenspace projector $\Pi_{\min}$. For any non-zero tilt parameter $\gamma \in \mathbb{R} \setminus \{0\}$, the global minimum of the Quantum Tilted Loss over the set of all density operators $\mathcal{S}$ is exactly the ground state energy of $O$:
\begin{align}
    \min_{\rho \in \mathcal{S}} \mathcal{L}_\gamma(O,\rho) = o_{\min}.
\end{align}
Furthermore, a state $\rho^\star \in \mathcal{S}$ achieves this minimum if and only if it is entirely supported on the ground state eigenspace, $\tr(\rho^\star \Pi_{\min}) = 1$.}

\begin{proof}
Consider the objective $\min_{\rho \in \mathcal{S}} \frac{1}{\gamma} \log \tr(e^{\gamma O} \rho)$. For $\gamma > 0$, the prefactor $1/\gamma$ is positive, so minimizing the loss is equivalent to minimizing the linear functional $\tr(e^{\gamma O} \rho)$. Since $f(x) = e^{\gamma x}$ is strictly increasing for $\gamma > 0$, the minimum eigenvalue of $e^{\gamma O}$ is $e^{\gamma o_{\min}}$. This minimum expectation value is achieved if and only if $\rho$ is supported entirely on the eigenspace $\Pi_{\min}$. Substituting this minimizer yields $\frac{1}{\gamma} \log(e^{\gamma o_{\min}}) = o_{\min}$. 

Conversely, for $\gamma < 0$, the prefactor $1/\gamma$ is negative, so minimizing the loss requires maximizing $\tr(e^{\gamma O} \rho)$. Since $f(x) = e^{\gamma x}$ is strictly decreasing for $\gamma < 0$, the maximum eigenvalue of $e^{\gamma O}$ corresponds to the smallest eigenvalue of $O$, which is again $e^{\gamma o_{\min}}$. This maximum is identically achieved if and only if $\rho$ is supported on $\Pi_{\min}$, yielding the exact same global minimum value $o_{\min}$.
\end{proof}

Lemma~\ref{lem:gibbs-variational} [Gibbs variational bound]
\textit{
Fix a real number $\gamma \neq 0$. Let $\rho$ be a full-rank density operator. Then the quantum tilted loss has the following variational representation:
For $\gamma>0$, it is given by a supremum:
\begin{align}
\mathcal{L}_\gamma 
\ge \sup_{\sigma} \Big\{ \tr[\sigma O] - \frac{1}{\gamma}D(\sigma\Vert \rho) \Big\}.
\end{align}
For $\gamma<0$, it is given by an infimum:
\begin{align}
\mathcal{L}_\gamma  \le \inf_{\sigma}\Big\{\tr[\sigma O] - \frac{1}{\gamma} D(\sigma\Vert \rho) \Big\},
\end{align}
where the optimization is over all density operators $\sigma$, and $D(\sigma\Vert\rho):=\tr[\sigma(\log\sigma-\log\rho)]$ is the Umegaki quantum relative entropy. Moreover, equality holds for both cases iff $[\rho,O]=0$, in which case the unique optimizer is $\sigma^\star\propto e^{\log\rho+\gamma O}$.
}

\begin{proof}
The Gibbs variational principle (a.k.a. the Donsker–Varadhan variational formula) states that for any Hermitian operator $K$ and density operator $Q$:
\be
\log\tr[e^{\log Q+K}]=\sup_{\sigma}\Big\{\tr[\sigma K]-D(\sigma\Vert Q)\Big\},
\ee
where the supremum is attained at $\sigma^\star\propto e^{\log Q+K}$. By Golden–Thompson inequality $\tr e^{A+B}\le \tr(e^A e^B)$ for Hermitian $A,B$ it yields that $\mathcal{L}_\gamma \ge \tfrac{1}{\gamma}\log\tr[e^{\log Q+K}]$ for $\gamma >0$ ( and $\le$ for $\gamma <0$). The lemma is obtained by setting $K=\gamma O$ and $Q=\rho$, and then dividing the entire expression by $\gamma$. When $\gamma<0$, dividing by a negative number flips the supremum to an infimum, yielding the second case.
\end{proof}

\section{Relations to other loss functions}

\subsection{Proof of Theorem~\ref{thm:cvar_vs_tilt}}\label{app:cvar_relation}\label{app:cvar_vs_tilt}

Theorem~\ref{thm:cvar_vs_tilt} [Negative $\gamma$ regime]
\textit{
Let $E(\bm{\theta})$ be a real-valued random energy with law $P_{\bm{\theta}}$. 
For any $\gamma < 0$ and any $\alpha\in(0,1)$, define the tilted loss
\begin{align}
    \mathcal{L}_\gamma(E,P_{\bm{\theta}}) 
    &:= \frac{1}{\gamma}\log \EX_{P_{\bm{\theta}}} \left[ e^{\gamma E} \right] \\\nonumber
    &= \frac{1}{\gamma}\log \left( \int_{-\infty}^{\infty} e^{\gamma E} \, dP_{\bm{\theta}}(E) \right)
\end{align}
and the lower-tail Conditional Value-at-Risk $\mathrm{CVaR}_{\alpha}(\bm{\theta})$ as in \cref{eq:cvar_integral}.  
Then the following inequality holds
\begin{align}
    \mathrm{CVaR}_{\alpha}(\bm{\theta})
    \geq
    \mathcal{L}_\gamma(\bm{\theta})
    + \frac{1}{\gamma}\log\!\left( \frac{1}{\alpha} \right).
\end{align}
}
\begin{proof}
We utilize the unified quantile integral representation for expectations and $\mathrm{CVaR}_\alpha$. By the probability integral transform, the moment-generating function can be expressed as an integral over the uniform interval $(0, 1)$ via the generalized inverse CDF $F_{\bm{\theta}}^{-1}(u)$ as
\begin{align}
\mathbb{E}_{P_{\bm{\theta}}}\!\left[e^{\gamma E(\bm{\theta})}\right]
    = \int_{0}^{1} e^{\gamma F_{\bm{\theta}}^{-1}(u)} \, du.
\end{align}
Since the integrand is strictly positive, we can lower-bound the integral by restricting the domain to the lowest $\alpha$-fraction $u \in (0, \alpha)$
\begin{align}
\mathbb{E}_{P_{\bm{\theta}}}\!\left[e^{\gamma E(\bm{\theta})}\right]
    &\ge \int_{0}^{\alpha} e^{\gamma F_{\bm{\theta}}^{-1}(u)} \, du \nonumber\\
    &= \alpha \left( \frac{1}{\alpha} \int_{0}^{\alpha} e^{\gamma F_{\bm{\theta}}^{-1}(u)} \, du \right).
\end{align}
Because the exponential function $x \mapsto e^{\gamma x}$ is convex for any $\gamma \in \mathbb{R}$, we apply Jensen's inequality to the normalized integral (which represents an expectation under the uniform measure on $(0, \alpha)$)
\begin{align}
\frac{1}{\alpha} \int_{0}^{\alpha} e^{\gamma F_{\bm{\theta}}^{-1}(u)} \, du
    &\ge \exp\!\left( \gamma \left( \frac{1}{\alpha} \int_{0}^{\alpha} F_{\bm{\theta}}^{-1}(u) \, du \right) \right) \nonumber\\
    &= \exp\!\left(\gamma\,\mathrm{CVaR}_{\alpha}(\bm{\theta})\right),
\end{align}
where we substituted the rigorous definition of $\mathrm{CVaR}_\alpha$ valid for both discrete and continuous distributions. Combining both inequalities yields
\begin{align}
\mathbb{E}_{P_{\bm{\theta}}}\!\left[e^{\gamma E(\bm{\theta})}\right]
    &\ge \alpha\,\exp\!\left(\gamma\,\mathrm{CVaR}_{\alpha}(\bm{\theta})\right).
\end{align}
Taking logarithms on both sides, we obtain:
\begin{align}
\log \mathbb{E}_{P_{\bm{\theta}}}\!\left[e^{\gamma E(\bm{\theta})}\right]
    &\ge \log \alpha + \gamma\,\mathrm{CVaR}_{\alpha}(\bm{\theta}).
\end{align}
Dividing by $\gamma < 0$ reverses the direction of the inequality:
\begin{align}
\mathcal{L}_\gamma(E,P_{\bm{\theta}}) 
    &= \frac{1}{\gamma}\log \mathbb{E}_{P_{\bm{\theta}}}\!\left[e^{\gamma E(\bm{\theta})}\right] \nonumber\\
    &\le \mathrm{CVaR}_{\alpha}(\bm{\theta})
       + \frac{1}{\gamma}\log \alpha.
\end{align}
Rewriting $\log \alpha = -\log (1/\alpha)$ and rearranging gives:
\begin{align}
\mathrm{CVaR}_{\alpha}(\bm{\theta})
    \;\geq\;
    \mathcal{L}_\gamma(\bm{\theta})
    + \frac{1}{\gamma}\log\!\left(\frac{1}{\alpha}\right),
\end{align}
which is the claimed inequality.
\end{proof}

\subsection{Proof of Theorem~\ref{thm:gibbs_tight}}\label{app:gibbs_proof}

\begin{lemma}\label{lem:lower_bound_gibbs}
Let 
$O = \sum_{i=1}^d E_i |E_i\rangle\!\langle E_i|$ 
with eigenvalues 
$E_{\min} = E_1 \le \cdots \le E_{\max} := E_d$.
For $\beta>0$, let the Gibbs state be
\begin{align}
\rho_\beta = \frac{e^{-\beta O}}{Z(\beta)},
\qquad 
Z(\beta) = \sum_{i=1}^d e^{-\beta E_i}.
\end{align}
For $\gamma<0$, define
\begin{align}
&\mathcal{L}_\gamma(\rho_\beta)
    := \frac{1}{\gamma}\log \Tr\!\big(e^{\gamma O}\rho_\beta\big),
\quad f(\gamma)
    := \mathcal{L}_\gamma(\rho_\beta)
    + \frac{1}{\gamma}\log\!\left(\frac{1}{\alpha}\right),
\quad \alpha\in(0,1).
\end{align}
Then, for all $\gamma<0$,
\begin{align}
E_{\min} + \frac{1}{\gamma}\log\!\left(\frac{1}{\alpha}\right) \leq f(\gamma)
\leq
E_{\min} + \frac{1}{\gamma}\log\!\left(\frac{p_{\min}}{\alpha}\right) ,
\end{align}
where $p_{\min}  = \langle E_{\min}|\rho_\beta|E_{\min}\rangle .$
\end{lemma}

\begin{proof}

Define the outcome probabilities of measuring $\rho_\beta$ in the energy eigenbasis of $O$ as
\begin{align}
p_i := \langle E_i|\rho_\beta|E_i\rangle 
= \frac{e^{-\beta E_i}}{Z(\beta)},
\end{align}
Then we have
\begin{align}
\Tr\!\big(e^{\gamma O}\rho_\beta\big)
    = \sum_{i=1}^d p_i e^{\gamma E_i}.
\end{align}
Since $\gamma<0$ and $E_i \ge E_1$ for all $i$, it holds that 
$e^{\gamma E_i} \le e^{\gamma E_1}$, and hence
\begin{align}
p_1 e^{\gamma E_1}
\;\le\;
\sum_{i=1}^d p_i e^{\gamma E_i}
\;\le\;
e^{\gamma E_1},
\end{align}
where we used $\sum_i p_i = 1$.
Taking logarithms,
\begin{align}
\log p_1 + \gamma E_1
\;\le\;
\log \Tr\!\big(e^{\gamma O}\rho_\beta\big)
\;\le\;
\gamma E_1.
\end{align}
Dividing by $\gamma<0$ (which reverses both inequalities) gives
\begin{align}
E_1 
\;\le\;
\mathcal{L}_\gamma(\rho_\beta)
= \frac{1}{\gamma}\log \Tr\!\big(e^{\gamma O}\rho_\beta\big)
\;\le\;
E_1 + \frac{1}{\gamma}\log p_1 .
\end{align}
Adding $\frac{1}{\gamma}\log(1/\alpha)$ to all terms yields
\begin{align}
E_1 + \frac{1}{\gamma}\log\!\left(\frac{1}{\alpha}\right)
\;\le\;
f(\gamma)
\;\le\;
E_1 + \frac{1}{\gamma}\log\!\left(\frac{p_1}{\alpha}\right),
\end{align}
and in particular the claimed upper bound on $f(\gamma)$.
\end{proof}

Theorem~\ref{thm:gibbs_tight} [Tightness of the entropic $\mathrm{CVaR}$ bound for Gibbs states] 
\textit{
Under the definitions above, for any confidence level $\alpha \in (0,p_1]$, the $\mathrm{CVaR}$ inequality becomes tight:
\be
    \mathrm{CVaR}_\alpha(X) = E_{\min} = \sup_{\gamma<0} \left\{ \mathcal{L}_\gamma(\rho_\beta) + \frac{1}{\gamma}\log\!\left(\frac{1}{\alpha}\right) \right\}.
\ee
}

\begin{proof}
For the measurement outcome $X$ with $\Pr(X = E_i) = p_i$, the lower-tail $\mathrm{CVaR}_\alpha(X)$ at level $\alpha$ must be evaluated using the proper definition for discrete distributions. The ground state is at energy $E_1$ with probability $p_1$. For any $\alpha \in (0, p_1]$, the $\alpha$-quantile is exactly $q_\alpha = E_1$.

Using the discrete definition of $\mathrm{CVaR}_\alpha$ that correctly splits the probability atom as
\begin{align}
\mathrm{CVaR}_\alpha(X) &= \frac{1}{\alpha} \left( \mathbb{E}\bigl[X \mathbb{I}_{\{X < E_1\}}\bigr] + E_1 \bigl(\alpha - \Pr(X < E_1)\bigr) \right).
\end{align}
Since $E_1$ is the minimum possible energy, the strict lower tail is empty ($\Pr(X < E_1) = 0$), so the expectation $\mathbb{E}\bigl[X \mathbb{I}_{\{X < E_1\}}\bigr] = 0$. This simplifies directly to:
\begin{align}
\mathrm{CVaR}_\alpha(X) &= \frac{1}{\alpha} \left( 0 + E_1 (\alpha - 0) \right) = E_1 = E_{\min}
\qquad\text{for all } \alpha\in(0,p_1] .
\end{align}
Then let's check the behaviour of $f (\gamma)$ by computing the following limit
\begin{align}
\lim_{\gamma \to -\infty} f(\gamma)
&= \lim_{\gamma \to -\infty} \left( 
      \mathcal{L}_\gamma(\rho_\beta)
      + \frac{1}{\gamma}\,\log\!\left(\frac{1}{\alpha}\right)
    \right) \\\nonumber
&= E_{\min} + 0 = E_{\min}.
\end{align}
where we used Lemma~\ref{lem:limit} noticing that $\rho_\beta$ and $O$ have the same support. Now we use that $\alpha \le p_1$. Then $p_1/\alpha \ge 1$ and hence 
$\log(p_1/\alpha) \ge 0$. Since $\gamma<0$, we have
\begin{align}
\frac{1}{\gamma}\log\!\left(\frac{p_1}{\alpha}\right)
\;\le\; 0,
\end{align}
so the upper bound from lemma~\ref{lem:lower_bound_gibbs} simplifies to
\begin{align}
f(\gamma)
\;\le\;
E_1
\qquad\text{for all } \gamma<0 \text{ and } \alpha\le p_1 .
\end{align}
Combining this with the previous inequality 
$\sup_{\gamma<0} f(\gamma) \le E_1$, we conclude
\begin{align}
\sup_{\gamma<0} f(\gamma) = E_1.
\end{align}
Finally, since for $\alpha\le p_1$ we have 
$\mathrm{CVaR}_\alpha(X) = E_1$ and 
$\sup_{\gamma<0} f(\gamma) \le \mathrm{CVaR}_\alpha(X)$ from the previous step, it follows that
\begin{align}
\mathrm{CVaR}_\alpha(X)
    = E_1
    = \sup_{\gamma<0} f(\gamma),
\end{align}
which proves the claim.
\end{proof}

\subsection{Proof of Theorem~\ref{thm:empirical_cvar_tilted}}\label{app:empirical_cvar_tilt}

Theorem~\ref{thm:empirical_cvar_tilted} [Empirical lower-tail CVaR vs.\ tilted loss]
\textit{
Let $\alpha\in(0,1)$, let $\rho(\bm{\theta})$ be any parametrised quantum state, and let $O$ be a Hamiltonian diagonal in the computational basis. Measuring $\rho(\bm{\theta})$ in this basis yields $K$ independent energy samples $\{E_k(\bm{\theta})\}_{k=1}^K$ from $P_{\bm{\theta}}$. Then for every $\gamma<0$,
\begin{align}
\widehat{\mathrm{CVaR}}_{\alpha,K}(\bm{\theta})
\;\ge\;
\widehat{\mathcal{L}}_{\gamma,K}(\bm{\theta})
    + \frac{1}{\gamma}\log\!\left(\frac{1}{\alpha}\right).
\end{align}
}

\begin{proof}
Let $m := \lceil \alpha K\rceil$.  
By convexity of the exponential function and Jensen's inequality,
\begin{align}
\frac{1}{m}
    \sum_{i=1}^{m} e^{\gamma E_{(i)}(\bm{\theta})}
    \;\ge\;
    \exp\!\left(
        \gamma\,\widehat{\mathrm{CVaR}}_{\alpha,K}(\bm{\theta})
    \right).
\end{align}
Since all terms are positive and the sum over the $m$ smallest exponentials
is bounded above by the full sum, we have
\begin{align}
\frac{1}{K}\sum_{k=1}^{K} e^{\gamma E_k(\bm{\theta})}
    &\;\ge\;
    \frac{1}{K}\sum_{i=1}^{m} e^{\gamma E_{(i)}(\bm{\theta})}\\\nonumber
    & \;\ge\;
    \frac{m}{K}\,
    \exp\!\left(
        \gamma\,\widehat{\mathrm{CVaR}}_{\alpha,K}(\bm{\theta})
    \right).
\end{align}

Taking logarithms and dividing by $\gamma<0$ (which reverses the inequality),

\begin{align}
\widehat{\mathcal{L}}_{\gamma,K}(\bm{\theta})
    &= \frac{1}{\gamma}
        \log\!\left(
            \frac{1}{K}\sum_{k=1}^{K} e^{\gamma E_k(\bm{\theta})}
        \right)\\\nonumber
    &\;\le\;
    \widehat{\mathrm{CVaR}}_{\alpha,K}(\bm{\theta})
    + \frac{1}{\gamma}\log\!\left(\frac{m}{K}\right).
\end{align}
Rearranging,
\begin{align}
\widehat{\mathrm{CVaR}}_{\alpha,K}(\bm{\theta})
    \;\ge\;
    \widehat{\mathcal{L}}_{\gamma,K}(\bm{\theta})
    + \frac{1}{\gamma}\log\!\left(\frac{K}{m}\right).
\end{align}
Since $m = \lceil \alpha K\rceil$ we have $m/K \ge \alpha$, hence
$K/m \le 1/\alpha$ and $\log(K/m) \le \log(1/\alpha)$.  
Because $\gamma<0$, multiplying by $1/\gamma$ reverses the inequality and yields
\begin{align}
\frac{1}{\gamma}\log\!\left(\frac{K}{m}\right)
    \;\ge\;
    \frac{1}{\gamma}\log\!\left(\frac{1}{\alpha}\right).
\end{align}
Combining with the previous bound gives
\begin{align}
\widehat{\mathrm{CVaR}}_{\alpha,K}(\bm{\theta})
    \;\ge\;
    \widehat{\mathcal{L}}_{\gamma,K}(\bm{\theta})
    + \frac{1}{\gamma}\log\!\left(\frac{1}{\alpha}\right),
\end{align}
which proves the claim.
\end{proof}

\subsection{Proof of Theorem~\ref{thm:QTL-vs-petz}}\label{app:petz-renyi-tilt}

Theorem~\ref{thm:QTL-vs-petz}
\textit{
Let $\rho,\sigma\in\mathcal S$ be density operators and let
$\gamma\in(0,1)\cup(1,\infty)$. Define
\be
L_\gamma(\log\sigma,\rho)
:=
\frac{1}{\gamma}\log\Tr(\sigma^\gamma\rho).
\ee
If $\gamma>1$, assume 
$\operatorname{supp}(\sigma)\subseteq \operatorname{supp}(\rho)$ so that 
$D_\gamma(\sigma\Vert\rho)<\infty$. Then
\be\label{eq:QTL-petz}
\mathcal{L}_\gamma(\log \sigma,\rho) \le \frac{\gamma-1}{\gamma}D_\gamma(\sigma\Vert\rho).
\ee
Equality holds iff $\rho$ is pure.
}
\begin{proof}
For $t\in[0,1]$ and any $\gamma>0$, $t^{\,1-\gamma}\ge t$. By functional calculus this gives
$\rho^{1-\gamma}-\rho\succeq 0$ on $\operatorname{supp}(\rho)$. Since $\sigma^\gamma\succeq 0$,
\begin{align}
\tr(\sigma^\gamma\rho) \le \tr(\sigma^\gamma\rho^{1-\gamma}).
\end{align}
Indeed, the right-hand side gives
\be
\frac{1}{\gamma}\log\Tr(\sigma^\gamma\rho^{1-\gamma})
=
\frac{\gamma-1}{\gamma}D_\gamma(\sigma\|\rho).
\ee

Let $B:=\rho^{1-\gamma}-\rho$, interpreted on
$\supp(\rho)$ when $\gamma>1$. For every nonzero
eigenvalue $\lambda\in(0,1]$ of $\rho$,
\be
\lambda^{1-\gamma}\ge \lambda ,
\ee
for all $\gamma\in(0,1)\cup(1,\infty)$. Hence
$B\succeq 0$ on the relevant support. Since
$\sigma^\gamma\succeq0$, and since for $\gamma>1$ we assume
$\supp(\sigma)\subseteq\supp(\rho)$, we obtain
\begin{equation}
\Tr(\sigma^\gamma\rho^{1-\gamma})
-
\Tr(\sigma^\gamma\rho)
=
\Tr(\sigma^\gamma B)
\ge 0.
\end{equation}
Taking logarithms and dividing by $\gamma>0$ proves \cref{eq:QTL-petz}.

It remains to characterize equality. Equality holds iff
\be
\Tr(\sigma^\gamma B)=0.
\ee
For positive semidefinite operators this is equivalent to
$\supp(\sigma)\subseteq\ker B$ on the relevant support. The
eigenvalues of $B$ are $\lambda^{1-\gamma}-\lambda$, which
vanish only at $\lambda=1$ among the nonzero eigenvalues of
$\rho$. For $0<\gamma<1$, $B$ also vanishes on
$\ker(\rho)$, but the assumption
$\Tr(\sigma^\gamma\rho)>0$ rules out equality coming only from
that kernel. Therefore equality is possible iff $\rho$ has an
eigenvalue equal to $1$, which for a density operator is
equivalent to $\rho$ being pure. Conversely, if $\rho$ is pure,
then $B=0$ on the support relevant to the trace, so equality
holds. For $\gamma>1$, the support assumption further forces
$\sigma=\rho$ in the equality case.
\end{proof}

\section{Additional details for QTL estimation}
\label{app:qtl_estimation_appendix}

In this appendix we provide the technical details omitted from Sec.~\ref{sec:estimation}. We treat four topics:
(i) exact reductions for structured observables, 
(ii) cumulant expansions in the small-tilt regime,
(iii) concentration bounds for diagonal observables under sampling access,
and (iv) gradient estimation via ratio estimators.

\subsection{Concentration bounds in the diagonal case}
\label{app:diagonal_qtl_bounds}

We now turn to the nonperturbative diagonal setting. Assume that
\begin{align}
O=\sum_{z\in\{0,1\}^n} E(z)\ket{z}\!\bra{z}
\end{align}
is diagonal in the computational basis. Measuring $\rho(\bm\theta)$ in the
computational basis yields a sample $z\sim p_{\bm\theta}(z)$ and hence a real
random variable
\begin{align}
X:=E(z).
\end{align}
Recall that
\begin{align}
Z_\gamma:=\Tr(e^{\gamma O}\rho(\bm\theta))
=
\mathbb{E}[e^{\gamma X}],
\qquad
\mathcal{L}_\gamma
:=
\frac{1}{\gamma}\log Z_\gamma.
\end{align}
Given $m$ i.i.d.\ samples $X_1,\dots,X_m$, define the empirical estimator as
\begin{align}\label{def:empirical_qtl}
\widehat Z_\gamma:=\frac{1}{m}\sum_{i=1}^m e^{\gamma X_i},
\qquad
\widehat{\mathcal{L}}_\gamma:=\frac{1}{\gamma}\log \widehat Z_\gamma.
\end{align}

We give a worst-case Hoeffding bound on estimating QTL with diagonal observable as below.

\begin{theorem}[Sample complexity for estimating diagonal QTL]\label{thm:qtl_hoeffding_appendix}
Assume $X\in[\lambda_{\min},\lambda_{\max}]$ almost surely, and let
\begin{align}
\Delta:=\lambda_{\max}-\lambda_{\min}.
\end{align}
Fix $\varepsilon\in(0,1)$, $\delta\in(0,1)$, and $\gamma\neq 0$.  Assume the target precision satisfies $\epsilon \leq 1/|\gamma|$.
If
\begin{align}
m \ge
\frac{2}{\gamma^2\varepsilon^2}
\bigl(e^{|\gamma|\Delta}-1\bigr)^2
\log\!\Big(\frac{2}{\delta}\Big),
\label{eq:qtl_hoeffding_m}
\end{align}
then
\begin{align}
\Pr\!\left(
|\widehat{\mathcal{L}}_\gamma-\mathcal{L}_\gamma|
\le \varepsilon
\right)\ge 1-\delta.
\end{align}
\end{theorem}

\begin{proof}
Under the \cref{def:empirical_qtl}, if $\gamma<0$ and $X\in[\lambda_{\min},\lambda_{\max}]$, the map $x\mapsto e^{\gamma x}$ is decreasing, hence
\begin{align}
Y \in \big[e^{\gamma\lambda_{\max}},\, e^{\gamma\lambda_{\min}}\big].
\end{align}
Hoeffding's inequality yields, for any $t>0$,
\begin{align}
\Pr \Big(|\widehat Z_\gamma - Z_\gamma|\ge t\Big)
\le
2\exp \left(-\frac{2m t^2}{(e^{\gamma\lambda_{\min}}-e^{\gamma\lambda_{\max}})^2}\right).
\label{eq:hoeffding_Z}
\end{align}
Set the RHS of \cref{eq:hoeffding_Z} to $\delta\in(0,1)$, yielding
\begin{align}
t_\delta:=(e^{\gamma\lambda_{\min}}-e^{\gamma\lambda_{\max}})\sqrt{\frac{\ln(2/\delta)}{2m}},
\end{align}
so that $\Pr(|\widehat Z_\gamma - Z_\gamma|\le t_\delta)\ge 1-\delta$. 
If $t_\delta< Z_\gamma$ then
\begin{align}
Z_\gamma-t_\delta \le \widehat Z_\gamma \le Z_\gamma+t_\delta \Longrightarrow 
1-\frac{t_\delta}{Z_\gamma}\le \frac{\widehat Z_\gamma}{Z_\gamma}\le 1+\frac{t_\delta}{Z_\gamma}.
\end{align}
Taking logs gives
\begin{align}
\log \left(1-\frac{t_\delta}{Z_\gamma}\right)
\le
\log \left(\frac{\widehat Z_\gamma}{Z_\gamma}\right)
\le
\log \left(1+\frac{t_\delta}{Z_\gamma}\right).
\end{align}
Since
\begin{align}
\widehat{\mathcal L}_\gamma-\mathcal L_\gamma
=\frac{1}{\gamma}\log\!\left(\frac{\widehat Z_\gamma}{Z_\gamma}\right),
\end{align}
we obtain
\begin{align}
\big|\widehat{\mathcal L}_\gamma-\mathcal L_\gamma\big|
\le
\frac{1}{|\gamma|}
\max\!\left\{
\log\!\left(1+\frac{t_\delta}{Z_\gamma}\right),\,
-\log\!\left(1-\frac{t_\delta}{Z_\gamma}\right)
\right\}.
\label{eq:log_bound}
\end{align}
Because $Y\ge e^{\gamma\lambda_{\max}}$ almost surely, $Z_\gamma=\mathbb E[Y]\ge e^{\gamma\lambda_{\max}}$.
Thus,
\begin{align}
\frac{t_\delta}{Z_\gamma}
&\le
\frac{t_\delta}{e^{\gamma\lambda_{\max}}}\nonumber\\
&=
\frac{e^{\gamma\lambda_{\min}}-e^{\gamma\lambda_{\max}}}{e^{\gamma\lambda_{\max}}}\sqrt{\frac{\ln(2/\delta)}{2m}}\nonumber\\
&=
\left(e^{|\gamma|\Delta}-1\right)\sqrt{\frac{\ln(2/\delta)}{2m}}
=:\eta.
\label{eq:eta}
\end{align}
Using the condition $\epsilon \leq 1/|\gamma|$ we get $\eta\le 1/2$. Then $\log(1+\eta)\le \eta$ and $-\log(1-\eta)\le \frac{\eta}{1-\eta}\le 2\eta$,
so \cref{eq:log_bound} implies
\begin{align}
\big|\widehat{\mathcal L}_\gamma-\mathcal L_\gamma\big|
\le \frac{2\eta}{|\gamma|}
=
\frac{2}{|\gamma|}\left(e^{|\gamma|\Delta}-1\right)\sqrt{\frac{\ln(2/\delta)}{2m}}.
\label{eq:final_bound}
\end{align}
Thus, to achieve $\big|\widehat{\mathcal L}_\gamma-\mathcal L_\gamma\big|\le \varepsilon$ with probability at least $1-\delta$, it suffices (up to constants) that
\begin{align}
m \ge \frac{2}{\gamma^2\varepsilon^2}\big(e^{|\gamma|\Delta}-1\big)^2\ln\left(\frac{2}{\delta}\right).
\end{align}
Same proof holds for $\gamma>0$.
\end{proof}

\subsection{Parameter shift rule for QTL}\label{app:qtl-parameter-shift}

\begin{theorem}[Parameter-shift rule for the quantum tilted loss]
\label{thm:qtl-parameter-shift}
Let
\be
\mathcal L_\gamma(\bm\theta)
:=
\frac{1}{\gamma}\log Z_\gamma(\bm\theta),
\qquad
Z_\gamma(\bm\theta)
:=
\Tr\!\big[\rho_{\bm\theta} e^{\gamma O}\big],
\ee
where $\rho_{\bm\theta} = U(\bm\theta)\rho_0 U(\bm\theta)^\dagger$. Assume that the parameter $\theta_k$ enters the circuit through a single gate $W_k(\theta_k)=e^{-i\theta_k V_k}$, so that
\be
U(\bm\theta)=U_-\,W_k(\theta_k)\,U_+,
\ee
with $U_\pm$ independent of $\theta_k$. Suppose moreover that the generator $V_k$ has two eigenvalues $\pm v_k$, i.e., $
V_k^2=v_k^2 \mathbb I $. Define the shifted parameters $\bm\theta_k^\pm
:=
\bm\theta \pm s_k \mathbf e_k, \,
s_k:=\frac{\pi}{4v_k}$, then the tilted moment admits the exact parameter-shift rule
\be
\partial_{\theta_k} Z_\gamma(\bm\theta)
=
v_k\Big(
Z_\gamma(\bm\theta_k^+)-Z_\gamma(\bm\theta_k^-)
\Big).
\label{eq:shift-Mgamma}
\ee
Consequently,
\be
\partial_{\theta_k}\mathcal L_\gamma(\bm\theta)
=
\frac{v_k}{\gamma}\,
\frac{
Z_\gamma(\bm\theta_k^+)-Z_\gamma(\bm\theta_k^-)
}{
Z_\gamma(\bm\theta)
}.
\label{eq:shift-Lgamma-M}
\ee
\end{theorem}

\begin{proof}
Fix all parameters except $\theta_k$. Since
\be
U(\bm\theta)=U_-\,e^{-i\theta_k V_k}\,U_+,
\ee
we may absorb the $\theta_k$-independent parts into an effective state and observable:
\be
\rho := U_+\rho_0 U_+^\dagger,
\qquad
A := U_-^\dagger e^{\gamma O} U_- .
\ee
Then
\be
Z_\gamma(\bm\theta)
=
f(\theta_k)
:=
\Tr\!\big[
A\,e^{-i\theta_k V_k}\rho\,e^{i\theta_k V_k}
\big].
\label{eq:def-f-theta}
\ee
Let $P_k:=\frac{V_k}{v_k}$, so that $P_k^2=\mathbb I$ by the assumption $V_k^2=v_k^2 \mathbb I$. Hence
\be
e^{-i\theta_k V_k}
=
\cos(v_k\theta_k)\,\mathbb I
-
i\sin(v_k\theta_k)\,P_k.
\label{eq:expansion-gate}
\ee
Substituting \cref{eq:expansion-gate} into \cref{eq:def-f-theta} and expanding shows that $f(\theta_k)$ is a trigonometric polynomial of the form
\be
f(\theta_k)
=
a+b\cos(2v_k\theta_k)+c\sin(2v_k\theta_k),
\label{eq:trig-form}
\ee
for some constants $a,b,c$ independent of $\theta_k$.

Differentiating \cref{eq:trig-form} yields
\be
f'(\theta_k)
=
2v_k\Big(
-b\sin(2v_k\theta_k)+c\cos(2v_k\theta_k)
\Big).
\label{eq:fprime}
\ee
On the other hand,
\begin{align}
f(\theta_k+s)-f(\theta_k-s)
&=
b\Big[
\cos(2v_k(\theta_k+s))-\cos(2v_k(\theta_k-s))
\Big]
\nonumber\\
&\quad
+
c\Big[
\sin(2v_k(\theta_k+s))-\sin(2v_k(\theta_k-s))
\Big]
\nonumber\\
&=
2\sin(2v_k s)\,
\Big(
-b\sin(2v_k\theta_k)+c\cos(2v_k\theta_k)
\Big),
\label{eq:shift-difference}
\end{align}
where we used $ \cos(x+y)-\cos(x-y) = -2\sin x \sin y,\,
\sin(x+y)-\sin(x-y) = 2\cos x \sin y$. 
Choosing $s=s_k=\frac{\pi}{4v_k}$ gives $\sin(2v_k s_k)=1$, so \cref{eq:shift-difference} becomes
\be
f(\theta_k+s_k)-f(\theta_k-s_k)
=
2\Big(
-b\sin(2v_k\theta_k)+c\cos(2v_k\theta_k)
\Big).
\label{eq:shift-special}
\ee
Comparing \cref{eq:fprime} and \cref{eq:shift-special}, we obtain
\be
f'(\theta_k)
=
v_k\Big(
f(\theta_k+s_k)-f(\theta_k-s_k)
\Big).
\ee
Recalling that $f(\theta_k)=Z_\gamma(\bm\theta)$ proves
\be
\partial_{\theta_k} Z_\gamma(\bm\theta)
=
v_k\Big(
Z_\gamma(\bm\theta_k^+)-Z_\gamma(\bm\theta_k^-)
\Big),
\ee
which is \cref{eq:shift-Mgamma}. It remains to differentiate the QTL. By the chain rule,
\be
\partial_{\theta_k}\mathcal L_\gamma(\bm\theta)
=
\frac{1}{\gamma}
\frac{\partial_{\theta_k}Z_\gamma(\bm\theta)}{Z_\gamma(\bm\theta)}.
\ee
Substituting \cref{eq:shift-Mgamma} gives
\be
\partial_{\theta_k}\mathcal L_\gamma(\bm\theta)
=
\frac{v_k}{\gamma}\,
\frac{
Z_\gamma(\bm\theta_k^+)-Z_\gamma(\bm\theta_k^-)
}{
Z_\gamma(\bm\theta)
},
\ee
which is \cref{eq:shift-Lgamma-M}. Equivalently, since $Z_\gamma(\bm\theta)=e^{\gamma \mathcal L_\gamma(\bm\theta)}$,
\be
\partial_{\theta_k}\mathcal L_\gamma(\bm\theta)
=
\frac{v_k}{\gamma}
\frac{
e^{\gamma \mathcal L_\gamma(\bm\theta_k^+)}
-
e^{\gamma \mathcal L_\gamma(\bm\theta_k^-)}
}{
e^{\gamma \mathcal L_\gamma(\bm\theta)}
}.
\label{eq:shift-Lgamma-exp}
\ee
In the common Pauli-generated case $e^{-i\theta_k P_k/2}$ with $P_k^2=\mathbb I$, one has $v_k=\tfrac12$ and $s_k=\pi/2$. 
\end{proof}

\section{Trainability}\label{app:proofs_trainability}

\subsection{Proof of Proposition~\ref{prop:warmup_projector_qtl}}
\label{app:proof_warmup_projector_qtl}

Propostion~\ref{prop:warmup_projector_qtl} [Properties of projector QTL benchmark]
\textit{
Let $O_G = \mathbb{I} - |0\rangle\langle 0|^{\otimes n}$ and the variational ansätz be the tensor-product circuit $V(\bm\theta)=\bigotimes_{j=1}^n e^{i\theta_j \sigma_x^{(j)}/2}$. This benchmark model has the following properties:
\begin{enumerate}
    \item Reduction: The projector QTL can be reduced as
\begin{align}
\mathcal L_\gamma(O_G,\rho(\bm\theta))
:=
\phi_\gamma\!\big(C_G(\bm\theta)\big) =\frac{1}{\gamma}\log\!\Big(1+(e^\gamma-1)C_G(\bm\theta)\Big),
\end{align}
where $C_G(\bm\theta) = 1 - \prod_{j=1}^n \cos^2(\theta_j/2)$.
    \item Monotonicity: The function $\phi_\gamma$ is strictly increasing on $[0,1]$ with $\phi_\gamma(0) = 0$ and $\phi_\gamma(x) > 0$ for $x > 0$. That is, the QTL is an exact scalar post-processing of the ordinary expectation value $C_G(\bm\theta)$---a monotone but nonlinear reshaping of the standard global cost.
    \item Faithfulness: $\mathcal L_\gamma(O_G,\rho(\bm\theta)) = 0 \iff C_G(\bm\theta) = 0$, so the tilted and untilted losses share the same global minimizers.
\end{enumerate}
}
\begin{proof}
Recall that
\be
O_G = \mathbb I - |0\rangle\langle 0|^{\otimes n}.
\ee
Since $|0\rangle\langle 0|^{\otimes n}$ is a projector, $O_G$ is also a projector:
\be
O_G^2 = O_G.
\ee
Therefore,
\be
e^{\gamma O_G}
= \sum_{m=0}^\infty \frac{\gamma^m}{m!} O_G^m
= \mathbb I + \sum_{m=1}^\infty \frac{\gamma^m}{m!} O_G
= \mathbb I + (e^\gamma - 1) O_G.
\ee
Taking expectation in the state $\rho(\bm\theta)$ gives
\be
Z_\gamma(\bm\theta)
= \Tr\!\big(e^{\gamma O_G}\rho(\bm\theta)\big)
= 1 + (e^\gamma - 1)\,\Tr\!\big(O_G\rho(\bm\theta)\big)
= 1 + (e^\gamma - 1)\,C_G(\bm\theta).
\ee
Hence
\be
\mathcal L_\gamma(O_G,\rho(\bm\theta))
= \frac{1}{\gamma}\log Z_\gamma(\bm\theta)
= \frac{1}{\gamma}\log\!\Big(1 + (e^\gamma - 1)\,C_G(\bm\theta)\Big)
= \phi_\gamma\!\big(C_G(\bm\theta)\big),
\ee
where
\be
\phi_\gamma(x) := \frac{1}{\gamma}\log\!\Big(1 + (e^\gamma - 1)x\Big), \qquad x \in [0,1].
\ee

For the explicit expression of $C_G$, since $e^{i\theta_j\sigma_x/2}|0\rangle = \cos(\theta_j/2)|0\rangle + i\sin(\theta_j/2)|1\rangle$, the overlap with the target state is $|\langle 0|^{\otimes n}|\psi(\bm\theta)\rangle|^2 = \prod_{j=1}^n \cos^2(\theta_j/2)$, giving $C_G(\bm\theta) = 1 - \prod_{j=1}^n \cos^2(\theta_j/2)$.

It remains to prove faithfulness. Since
\be
1 + (e^\gamma - 1)x > 0 \qquad \text{for all } x \in [0,1],
\ee
the function $\phi_\gamma$ is well defined on $[0,1]$. Differentiating gives
\be
\phi_\gamma'(x) = \frac{e^\gamma - 1}{\gamma\big(1 + (e^\gamma - 1)x\big)}.
\ee
For every $\gamma \neq 0$, the numerator and denominator have the same sign, so $\phi_\gamma'(x) > 0$ for all $x \in [0,1]$, establishing strict monotonicity. Moreover, $\phi_\gamma(0) = 0$, so $\phi_\gamma(x) > 0$ for all $x > 0$. Therefore $\mathcal L_\gamma = 0 \iff C_G = 0$.
\end{proof}

\subsection{Proof of Theorem~\ref{thm:critical_tilt_transition}}
\label{app:train}

Theorem~\ref{thm:critical_tilt_transition} [Gradient-variance regimes in the projector QTL benchmark]
\textit{
Consider the projector benchmark model above with i.i.d.\ uniform random initialization $\theta_j \sim \mathrm{Unif}(-\pi,\pi)$.
\begin{itemize}
\item Fixed small tilt: For fixed $|\gamma| \ll 1$, the gradient variance of $\mathcal L_\gamma$ exhibits the same exponential-in-$n$ decay as the standard global cost $C_G$:
\be
\mathrm{Var}[\partial_{\theta_k}\mathcal L_\gamma]
=
\frac{1}{8}\Big(\frac{3}{8}\Big)^{n-1}\big(1+\mathcal{O}(\gamma)\big).
\ee
\item Linear negative tilt schedule: There exists an explicit tilt schedule $\gamma(n)=\Theta(-n)$ under which the gradient variance admits a polynomial lower bound in $n$, in contrast to the exponential upper bound at fixed tilt.
\end{itemize}
More concretely, for the explicit choice
\be
\gamma(n)=2(n-1)\log(3/8),
\ee
the appendix proves the lower bound
\be
\mathrm{Var}[\partial_{\theta_k}\mathcal L_{\gamma(n)}]
=
\Omega\!\left(\frac{1}{n^2}\right).
\ee
}
\begin{proof}
In this proof we consider the shallow projector benchmark model under i.i.d.\ uniform random initialization $\theta_j\sim \mathrm{Unif}(-\pi,\pi)$.

We recall the benchmark model of Ref.~\cite{Cerezo2021}. The target state is
\be
|\psi_\mathrm{target}\rangle = |0\rangle^{\otimes n},
\ee
and the variational ansätz is the tensor-product circuit
\be
V(\bm\theta)
= \bigotimes_{j=1}^n e^{i\theta_j\sigma_x^{(j)}/2},
\qquad
|\psi(\bm\theta)\rangle
:= V(\bm\theta)|0\rangle^{\otimes n}.
\ee
Since
\be
e^{i\theta_j\sigma_x/2}|0\rangle
= \cos\frac{\theta_j}{2}|0\rangle + i \sin\frac{\theta_j}{2}|1\rangle,
\ee
we obtain
\be
|\psi(\bm\theta)\rangle
= \bigotimes_{j=1}^n
  \Big(\cos\frac{\theta_j}{2}|0\rangle + i \sin\frac{\theta_j}{2}|1\rangle \Big).
\ee
Let
\be
O_G := \mathbb{I} - |0\rangle\langle 0|^{\otimes n},
\qquad
\rho(\bm\theta)
:= |\psi(\bm\theta)\rangle\langle\psi(\bm\theta)|.
\ee
Define
\be
C_G(\bm\theta):= \langle\psi(\bm\theta)|O_G|\psi(\bm\theta)\rangle
= 1-\prod_{j=1}^n f_j(\theta_j),
\qquad
f_j(\theta_j):=\cos^2\!\Big(\frac{\theta_j}{2}\Big).
\ee
By Proposition~\ref{prop:warmup_projector_qtl},
\be\label{eq:app_qtl_projector_closed}
\mathcal{L}_\gamma(O_G,\rho(\bm{\theta}))
= \frac{1}{\gamma}
  \log\Big(1 + (e^\gamma-1)\,C_G(\bm\theta)\Big).
\ee

We divide the proof into three steps.

\paragraph*{Step 1: exact gradient formula and parity.}
Define
\be
\alpha := \frac{e^\gamma-1}{\gamma},
\qquad
\beta := e^\gamma-1.
\ee
Differentiating \cref{eq:app_qtl_projector_closed} with respect to $\theta_k$ gives
\be
g_k(\bm\theta) := \partial_{\theta_k}\mathcal{L}_\gamma
= \frac{1}{\gamma}
  \frac{\beta \partial_{\theta_k}C_G(\bm\theta)}
       {1+\beta C_G(\bm\theta)}
= \alpha
  \frac{\partial_{\theta_k}C_G(\bm\theta)}
       {1+\beta C_G(\bm\theta)}.
\ee
Let
\be\label{eq:app_B_theta}
B(\bm\theta_{-k}) := \prod_{j\neq k} f_j(\theta_j),
\ee
where $\bm\theta_{-k}$ denotes all angles except $\theta_k$. Then
\be
C_G(\bm\theta)=1-f_k(\theta_k) B(\bm\theta_{-k}).
\ee
Since
\be
f_k(\theta_k)
= \cos^2\!\Big(\frac{\theta_k}{2}\Big)
= \frac{1+\cos\theta_k}{2},
\ee
we have
\be
f_k'(\theta_k)
= -\frac{1}{2}\sin\theta_k,
\qquad
\partial_{\theta_k}C_G(\bm\theta)
= -f_k'(\theta_k)\,B(\bm\theta_{-k}).
\ee
Thus
\be\label{eq:app_gk_explicit}
g_k(\bm\theta)
= -\alpha \frac{ f_k'(\theta_k)B(\bm\theta_{-k})}{1+\beta (1- f_k(\theta_k)B(\bm\theta_{-k}))}.
\ee

Assuming the angles $\theta_j$ are i.i.d.\ uniformly distributed on $(-\pi,\pi)$, a parity argument shows
\be
\mathbb{E}[g_k(\bm\theta)]=0.
\ee
Indeed, conditioning on $B(\bm\theta_{-k})$, the numerator in \cref{eq:app_gk_explicit} is odd in $\theta_k$, while the denominator is even in $\theta_k$. Hence
\be
\mathrm{Var}(g_k(\bm\theta))=\mathbb{E}[g_k(\bm\theta)^2].
\ee
Squaring \cref{eq:app_gk_explicit} gives
\be\label{eq:app_gk_squared}
g_k(\bm\theta)^2
=\alpha^2\frac{f_k'(\theta_k)^2\,B(\bm\theta_{-k})^2}{\big(1+\beta  (1- f_k(\theta_k)B(\bm\theta_{-k}))\big)^2}.
\ee
Expanding the denominator, we obtain
\be\label{eq:app_denom_expanded}
1+\beta\bigl(1-f_k(\theta_k)B(\bm\theta_{-k})\bigr)
=e^\gamma+(1-e^\gamma)f_k(\theta_k)B(\bm\theta_{-k}).
\ee
Hence
\be\label{eq:app_gk_expanded}
g_k(\bm\theta)^2
= \alpha^2
\frac{f_k'(\theta_k)^2B(\bm\theta_{-k})^2}
{\big(e^\gamma+(1-e^\gamma)f_k(\theta_k)B(\bm\theta_{-k})\big)^2}.
\ee
This makes the competition explicit: the denominator is a sum of a constant term $e^\gamma$ and a global product term $f_k(\theta_k)B(\bm\theta_{-k})\in[0,1]$.

\paragraph*{Step 2: fixed small tilt gives exponential decay.}
For $|\gamma|\ll 1$, we have
\be
\beta = \gamma + \mathcal{O}(\gamma^2),
\qquad
\alpha = 1 + \mathcal{O}(\gamma),
\ee
so the denominator in \cref{eq:app_gk_expanded} becomes $1 + \mathcal{O}(\gamma)$. Therefore,
\be\label{eq:app_var_small_tilt_start}
\mathrm{Var}(g_k(\bm\theta))
= \mathbb{E}[g_k(\bm\theta)^2]
= \mathbb{E}_{\bm\theta_{-k}}\left[\mathbb{E}_{\theta_k}[g_k(\bm\theta)^2\,\big|\, \bm\theta_{-k}] \right]
= \alpha^2 \mathbb{E}_{\bm\theta_{-k}}\left[ B(\bm\theta_{-k})^2 \, \mathbb{E}_{\theta_k}[f_k'(\theta_k)^2]\right].
\ee
Since $f_k'(\theta_k)=-\tfrac12\sin\theta_k$ and $\mathbb{E}[\sin^2\theta_k]=1/2$, we find
\be
\mathbb{E}_{\theta_k}[f_k'(\theta_k)^2] = \frac{1}{8}.
\ee
On the other hand, from \cref{eq:app_B_theta},
\be
\mathbb{E}_{\bm\theta_{-k}}[B^2(\bm\theta_{-k})]
= \big(\mathbb{E}[f(\theta)^2]\big)^{n- 1}.
\ee
Using
\be
f(\theta)^2 = \cos^4\!\Big(\frac{\theta}{2}\Big)
= \frac{3 + 4\cos\theta + \cos 2\theta}{8},
\ee
and $\mathbb{E}[\cos\theta]=\mathbb{E}[\cos 2\theta]=0$, we obtain
\be
\mathbb{E}[f(\theta)^2] = \frac{3}{8},
\ee
hence
\be\label{eq:app_EB2}
\mathbb{E}_{\bm\theta_{-k}}[B^2(\bm\theta_{-k})]
= \left(\frac{3}{8}\right)^{n-1}.
\ee
Combining \cref{eq:app_var_small_tilt_start} and \cref{eq:app_EB2}, we obtain
\be\label{eq:app_var_small_tilt_final}
\mathrm{Var}(g_k(\bm\theta))
= \frac{1}{8}\left(\frac{3}{8}\right)^{n-1}\big(1+\mathcal{O}(\gamma)\big).
\ee
Thus, for any fixed small tilt $|\gamma| \ll 1$, the tilted loss $\mathcal{L}_\gamma$ exhibits the same exponentially vanishing gradient variance as the standard global cost.

\paragraph*{Step 3: a linear negative tilt schedule yields polynomial scaling.}
We now consider the negative-tilt regime. From \cref{eq:app_gk_squared}, for fixed $\theta_k$, define
\be
F(x)
= \frac{x}
{\big(1+\beta  (1- f_k \sqrt{x})\big)^2},
\qquad
x:= B^2(\bm\theta_{-k})\in[0,1].
\ee
A direct differentiation yields
\be
F''(x)=\frac{3(1+\beta) \beta f_k}{2\sqrt{x}\big(1+\beta  (1- f_k \sqrt{x})\big)^4}.
\ee
Since the sign of $F''(x)$ is determined by $\beta$, the function $F(x)$ is convex for $\beta>0$ and concave for $\beta\in(-1,0)$, equivalently for $\gamma<0$. For $\gamma<0$, the graph of a concave function lies above the chord joining $(0,F(0))$ and $(1,F(1))$, so
\be
\mathbb{E}[F(x)] \ge (1-\mathbb{E}[x])F(0) +\mathbb{E}[x]F(1)
= \mathbb{E}[x]
\frac{1}{\left(1+\beta\left(1-f_k \right)\right)^{2}}.
\ee
Using \cref{eq:app_EB2}, and averaging first over the global angles $\bm\theta_{-k}$ for fixed $f_k(\theta_k)$, we obtain
\be\label{eq:app_var_large_tilt_start}
\mathrm{Var}(g_k(\bm\theta))
=
\left(\frac{e^\gamma-1}{2\gamma}\right)^{2}\cdot\left(\frac{3}{8}\right)^{n-1}
\mathbb{E}_{\theta_k}\left[\frac{\sin^2\theta_k}{\left(1+\left(e^\gamma-1\right)\left(1-\cos^2\left(\frac{\theta_k}{2}\right)\right)\right)^{2}}\right].
\ee
It remains to evaluate the one-dimensional integral. Let $u=\theta/2$. Then
\be
\mathbb{E}_{\theta_k}\left[\frac{\sin^2\theta_k}{\left(1+\beta\left(1-\cos^2\left(\frac{\theta_k}{2}\right)\right)\right)^{2}}\right]
=
\frac{1}{2\pi}\int_{-\pi}^{\pi}\frac{\sin^2\theta}{\big(1+\beta\sin^2(\theta/2)\big)^2}d\theta
=
\frac{8}{\pi}\int_0^{\pi/2}\frac{\sin^2u\cos^2u}{(1+\beta\sin^2u)^2}du.
\ee
Define
\be
I_0(\beta):=\int_0^{\pi/2}\frac{du}{1+\beta\sin^2u}
=\frac{\pi}{2\sqrt{1+\beta}},
\ee
so that
\be
I_0'(\beta)
=
-\int_0^{\pi/2}\frac{\sin^2u}{(1+\beta\sin^2u)^2}du,
\ee
and
\be
J_0(\beta):=\int_0^{\pi/2}\frac{du}{(1+\beta\sin^2u)^2}
=I_0(\beta)+\beta I_0'(\beta).
\ee
Using
\be
\sin^2u\cos^2u=\sin^2u-\sin^4u
\ee
and
\be
\frac{\sin^4u}{(1+\beta\sin^2u)^2}
=\frac{1}{\beta^2}\left(1-\frac{2}{1+\beta\sin^2u}+\frac{1}{(1+\beta\sin^2u)^2}\right),
\ee
one obtains after substitution and simplification
\be
\frac{8}{\pi}\int_0^{\pi/2}\frac{\sin^2u\cos^2u}{(1+\beta\sin^2u)^2}du
=\frac{2}{\sqrt{1+\beta}\big(1+\sqrt{1+\beta}\big)^2}
=\frac{2}
{\sqrt{e^\gamma}\Big(1+\sqrt{e^\gamma} \Big)^2}.
\ee
Substituting into \cref{eq:app_var_large_tilt_start} yields
\be\label{eq:app_var_large_tilt_exact}
\mathrm{Var}(g_k(\bm\theta))
\ge
\left(\frac{e^\gamma-1}{2\gamma}\right)^{2}\cdot\left(\frac{3}{8}\right)^{n-1}
\frac{2}
{\sqrt{e^\gamma}\Big(1+\sqrt{e^\gamma} \Big)^2}.
\ee

We now choose the explicit linear negative tilt schedule
\be\label{eq:app_gamma_schedule}
\gamma(n)=2(n-1)\log \frac{3}{8},
\ee
which is negative and of order $-n$. Then
\be
\sqrt{e^{\gamma(n)}}=\left(\frac{3}{8}\right)^{n-1}.
\ee
Substituting this relation into \cref{eq:app_var_large_tilt_exact}, we obtain
\be\label{eq:app_var_poly_bound}
\mathrm{Var}(g_k(\bm\theta))
\geq
\frac{1}{2}
\frac{\big(1-\sqrt{e^{\gamma(n)}}\big)^2}{\gamma(n)^2}.
\ee
For $n\ge 2$, one has
\be
\sqrt{e^{\gamma(n)}}=\left(\frac{3}{8}\right)^{n-1}\le \frac{3}{8},
\ee
hence
\be
1-\sqrt{e^{\gamma(n)}}\ge \frac{5}{8}.
\ee
Using \cref{eq:app_gamma_schedule}, we conclude that
\be
\mathrm{Var}(g_k(\bm\theta))
\ge
\frac{1}{8}\left(\frac{5}{8}\right)^2 \frac{1}{\log^2(3/8)(n-1)^2}
=
\Omega\left(\frac{1}{n^2}\right).
\ee

This proves the existence of a linear negative tilt schedule under which the gradient variance improves from exponentially small to polynomially small. Combined with \cref{eq:app_var_small_tilt_final}, this establishes the two regimes stated in the theorem.
\end{proof}

\section{Gradient Resolvability and Signal-to-Noise Ratio}\label{app:gradient_resolvability}

\subsection{From landscape geometry to stochastic optimization}

While the previous sections demonstrate that nonlinear objective design can reshape the cost landscape to mitigate vanishing gradients, a non-vanishing  variance is a necessary but insufficient condition for successful optimization. We must address the operational reality of near-term Variational Quantum Algorithms: gradients are not accessed analytically, but are estimated stochastically from a finite number of quantum measurements. Even a geometrically steep landscape is untrainable if the gradient signal is drowned out by shot noise.

Consider the standard gradient descent update rule for a parameter vector $\bm{\theta}$ at step $t$
\begin{align}
    \bm{\theta}_{t+1} = \bm{\theta}_t - \eta \hat{g}(\bm{\theta}_t),
\end{align}
where $\eta$ is the learning rate and $\hat{g}(\bm{\theta}_t)$ is a stochastic estimator of the true gradient $\nabla \mathcal{L}(\bm{\theta}_t)$ for a general loss function $\mathcal{L}$. This estimator is typically computed by averaging $N$ independent measurement outcomes. We decompose the estimator into the true signal and a zero-mean noise term as
\be
    \hat{g}(\bm{\theta}) = \nabla \mathcal{L}(\bm{\theta}) + \bm{\xi}, \quad \text{where } \mathbb{E}[\bm{\xi}] = 0.
\ee
The statistical quality of this estimator is governed by its covariance matrix $\Sigma_N(\bm{\theta}) := \mathbb{E}[\bm{\xi}\bm{\xi}^T]$. Assuming independent samples, the variance of the estimator scales inversely with the shot count $N$, such that $\Sigma_N(\bm{\theta}) = \frac{1}{N}\Sigma_1(\bm{\theta})$, where $\Sigma_1$ represents the covariance of a single-shot gradient estimate.

For the optimization trajectory to follow the descent direction rather than degenerate into a random walk, the true gradient must be resolvable against the background noise. We quantify this via the \textit{Noise-to-Signal Ratio} (NSR)~\cite{mccandlish2018empirical}, denoted $\mathcal{R}$
\begin{align}
    \mathcal{R}(\bm{\theta}) = \frac{\Tr(\Sigma_N(\bm{\theta}))}{\|\nabla \mathcal{L}(\bm{\theta})\|^2} = \frac{\sigma^2_{\text{shot}}(\bm{\theta})}{N \|\nabla \mathcal{L}(\bm{\theta})\|^2},
\end{align}
where $\sigma^2_{\text{shot}}(\bm{\theta}) := \Tr(\Sigma_1(\bm{\theta}))$ is the total variance of the single-shot gradient estimator (which depends on both the state $\bm{\theta}$ and the specific measurement scheme). For the optimizer to effectively descend the landscape, we require $\mathcal{R} \lesssim 1$. This threshold aligns with the critical batch size requirements for convergence stability established in classical stochastic optimization \cite{mccandlish2018empirical, bottou2018optimization} applied to deep learning.

\subsection{Hypothesis Testing Interpretation}

The condition $\mathcal{R} \lesssim 1$ can be formally motivated by hypothesis testing. We frame the gradient estimation problem as the minimal requirement of distinguishing between two hypotheses regarding the observed vector $\hat{g}$.
\begin{itemize}
    \item $H_0$ (Noise Only): The true gradient is zero, and the observed vector is pure statistical fluctuation, i.e., $\hat{g} \sim \mathcal{N}(0, \Sigma_N)$.
    \item $H_1$ (Signal + Noise): A true gradient is present, i.e., $\hat{g} \sim \mathcal{N}(\nabla \mathcal{L}, \Sigma_N)$.
\end{itemize}

If $H_0$ holds, the optimization reduces to a random walk. The information-theoretic limit on distinguishing these hypotheses is governed by the Kullback-Leibler (KL) divergence. Assuming isotropic Gaussian noise for analytical clarity, the divergence is determined by the Mahalanobis distance between the distributions \cite{Cover2006} as
\begin{align}
    D_{\text{KL}}(H_1 \| H_0) &= \frac{1}{2} (\nabla \mathcal{L})^T \Sigma_N^{-1} (\nabla \mathcal{L}) \\\nonumber
    &= \frac{N \|\nabla \mathcal{L}\|^2}{2\sigma^2_{\text{shot}}} = \frac{1}{2\mathcal{R}}.
\end{align}
By the Chernoff-Stein Lemma \cite{Cover2006}, the probability of a Type II error (falsely concluding there is no gradient when one exists) decays exponentially with $D_{\text{KL}}$. Consequently, a small ratio $\mathcal{R} < 1$ implies a large divergence, ensuring the gradient direction is statistically distinguishable from measurement noise. Conversely, if $\mathcal{R} \gg 1$, the divergence vanishes, and the optimizer cannot reliably extract the descent direction.

\subsection{Sample Complexity and the Barren Plateau Connection}

We now connect this local resolvability condition to the global properties of the landscape. Following Definition~1 in~\cite{Arrasmith2021effectofbarren}, a generic cost function $\mathcal{L}$ exhibits a \textit{Barren Plateau} if the variance of the partial derivative with respect to a parameter $\theta_k$ vanishes exponentially with the system size $n$. That is, for $\bm{\theta}$ drawn from the initialization distribution we have
\begin{align}
    \mathrm{Var}_{\bm{\theta}}[\partial_k \mathcal{L}] \le F(n), \quad \text{where } F(n) \in \mathcal{O}\left(\frac{1}{b^n}\right) \text{ for } b > 1.
\end{align}
The operational difficulty caused by this vanishing variance is probabilistic. As detailed by Chebyshev's inequality~\cite{Arrasmith2021effectofbarren}, the probability that the gradient magnitude at a random point exceeds a threshold $c > 0$ is bounded by
\begin{align}
    \label{eq:chebyshev_bound}
    \Pr_{\bm{\theta}}\left( |\partial_k \mathcal{L}(\bm{\theta})| \ge c \right) \le \frac{\mathrm{Var}_{\bm{\theta}}[\partial_k \mathcal{L}]}{c^2}.
\end{align}
For the gradient to be resolvable ($\mathcal{R} \lesssim 1$), the signal magnitude $|\partial_k \mathcal{L}|$ must exceed the standard deviation of the estimator's noise. For a single parameter component $k$, the noise floor is $\sqrt{\sigma_k^2/N}$, where $\sigma_k^2$ is the single-shot variance associated with measuring $\partial_k$. 

Setting the threshold to this noise floor, $c = \sqrt{\sigma_k^2/N}$, allows us to bound the probability that the signal emerges from the noise as
\begin{align}
    \Pr_{\bm{\theta}}\left( \text{Signal} \ge \text{Noise} \right) &= \Pr_{\bm{\theta}}\left( |\partial_k \mathcal{L}| \ge \sqrt{\frac{\sigma_k^2}{N}} \right) \nonumber\\&\le \frac{N \cdot \mathrm{Var}_{\bm{\theta}}[\partial_k \mathcal{L}]}{\sigma_k^2}.
\end{align}
For optimization to succeed, this probability must be of order $\Theta(1)$. This requires the right-hand side to be lower-bounded by a constant, dictating the necessary sample complexity scales as
\begin{align}
    N \gtrsim \frac{\sigma_k^2}{\mathrm{Var}_{\bm{\theta}}[\partial_k \mathcal{L}]}.
\end{align}
This inequality formalizes the fundamental resource cost of variational optimization as a competition between the available shots $N$, the estimator noise $\sigma_k^2$, and the landscape geometry $\mathrm{Var}_{\bm{\theta}}[\partial_k \mathcal{L}]$. In a standard barren plateau, the denominator is exponentially suppressed, forcing $N$ to scale exponentially to maintain resolvability. 

This is precisely where the Quantum Tilted Loss (QTL) intervenes. By replacing the standard expectation value with the tunable objective $\mathcal{L}_\gamma$, we reshape the geometry to exponentially increase the signal strength $\mathrm{Var}_{\bm{\theta}}[\partial_k \mathcal{L}_\gamma]$. This effectively boosts the denominator of the bound, attempting to reduce the required sample complexity $N$. However, this geometric intervention simultaneously impacts the numerator $\sigma_k^2$, which we analyze below.

\subsection{Estimation Regimes: The Cost of Tilted Gradients}

While a sufficiently large tilt $|\gamma|$ recovers a non-vanishing gradient signal across the optimization landscape, we must critically assess the statistical cost of extracting this signal. In order to estimate the gradient we can use near-term regimes (sampling-based) approaches and fault-tolerant regimes approaches.

\paragraph{Near-Term Regime (NISQ).} 
In current variational algorithms targeting diagonal Hamiltonians, expectations are estimated via direct basis sampling. For the QTL, the gradient estimator takes the form of a re-weighted expectation: $\nabla \mathcal{L}_\gamma = \frac{1}{\gamma Z_\gamma} \nabla Z_\gamma$. The term $\nabla Z_\gamma$ involves the expectation $\mathbb{E}_{z \sim p_{\bm{\theta}}}[e^{\gamma E(z)} \nabla \log p_{\bm{\theta}}(z)]$. 

The variance of this estimator is dominated by the second moment of the random variable $Y = e^{\gamma E(z)} \nabla \log p_{\bm{\theta}}(z)$, which scales as
\begin{align}
    \mathbb{E}[Y^2] \approx \sum_z p_{\bm{\theta}}(z) \left(e^{\gamma E(z)}\right)^2 \|\nabla \log p_{\bm{\theta}}(z)\|^2.
\end{align}
Because of the squared exponential weight, the single-shot variance of the gradient estimator scales exponentially with the spectral gap $\Delta = \lambda_{\max} - \lambda_{\min}$ of the observable, yielding $\mathrm{Var}(\hat{g}) \propto e^{2|\gamma|\Delta}$. This mirrors the worst-case sample complexity derived for the objective itself in Theorem~\ref{thm:qtl_hoeffding_appendix}. 

A natural alternative is to employ gradient-free methods (e.g., Nelder-Mead \cite{Arrasmith2021effectofbarren}), which bypass the direct estimation of partial derivatives. However, these methods generally rely on resolving finite differences between evaluated loss values (e.g., $\mathcal{L}_\gamma(\bm{\theta}^+) - \mathcal{L}_\gamma(\bm{\theta}^-)$). As established, estimating $\mathcal{L}_\gamma$ entails an incoherent sampling complexity of $\mathcal{O}(e^{2|\gamma|\Delta})$. Thus, gradient-free approaches face the identical statistical bottleneck: the true signal of the loss difference is overwhelmed by the exponential variance of the estimator.

This exposes a fundamental \textit{trainability-estimability trade-off} in the sampling regime. Increasing $|\gamma|$ mitigates the barren plateau by amplifying the spatial signal, but it simultaneously triggers an exponential explosion in the variance of the stochastic estimator $\sigma_k^2$. Consequently, for aggressive global tilts ($|\gamma| = \Omega(n)$), the required shot count $N$ to keep $\mathcal{R} \lesssim 1$ remains exponentially large. The landscape is mathematically steep, but operationally flat under finite sampling budgets.

\paragraph{Fault-Tolerant Regime.} 
In the fault-tolerant era, this strict exponential sampling penalty can be relaxed. Rather than relying on incoherent sampling (where the statistical error scales as $1/\sqrt{N}$), coherent quantum algorithms such as \textit{Quantum Amplitude Estimation} (QAE)~\cite{brassard_2000} can be employed. QAE estimates the partition function $Z_\gamma$ and its gradients to within an additive error $\varepsilon$ using a query complexity that scales as $\mathcal{O}(1/\varepsilon)$. This provides a rigorous quadratic speedup over the $\mathcal{O}(1/\varepsilon^2)$ sample complexity required by classical measurement. Furthermore, because the reweighting operator $e^{\gamma O}$ is mathematically equivalent to an unnormalized thermal density matrix with inverse temperature $\gamma \equiv -\beta$, we can leverage fault-tolerant algorithms designed for Gibbs state preparation \cite{gibbs_sampling1, temme2011quantum}. These protocols utilize Quantum Phase Estimation to coherently encode the exponential weights $e^{\gamma E_k / 2}$ directly into quantum amplitudes. This circumvents the exponential variance explosion associated with classical reweighting and bypasses the severe inefficiencies of classical rejection sampling. Consequently, the statistical cost of evaluating the QTL is no longer strictly bottlenecked by the magnitude of the tilt, allowing the geometric benefits of landscape reshaping to be exploited without prohibitive measurement overhead.

\end{document}